\documentclass[10pt, a4paper]{article}
\pdfoutput=1
\usepackage{jheppub}
\usepackage{bm}

\usepackage{graphics} 
\usepackage{amsfonts}
\usepackage{amsmath}
\usepackage{transparent, colortbl}
\usepackage{booktabs}
\usepackage{multirow}
\usepackage{anysize}
\usepackage{array}
\usepackage{latexsym, amscd, amsthm}
\usepackage{pdfsync}
\usepackage{epsf}
\usepackage[all]{xy}
\usepackage{color}
\usepackage[usenames,dvipsnames]{xcolor}
\usepackage{enumerate}
\usepackage{float}
\restylefloat{table}
\usepackage{upgreek}
\usepackage{subcaption}
\usepackage{array}
\usepackage{multirow}

\DeclareSymbolFont{extraup}{U}{zavm}{m}{n}
\DeclareMathSymbol{\varheart}{\mathalpha}{extraup}{86}
\DeclareMathSymbol{\vardiamond}{\mathalpha}{extraup}{87}

\newtheorem{Theorem}{Theorem}[section]
\newtheorem{Corollary}[Theorem]{Corollary}
\newtheorem{Lemma}[Theorem]{Lemma}
\newtheorem{Proposition}[Theorem]{Proposition}
\theoremstyle{definition}
\newtheorem{Definition}[Theorem]{Definition}
\theoremstyle{remark}
\newtheorem*{Remark}{Remark}

\newcommand{\Hom}{{\rm Hom}}

\def\ad{\mathrm{ad}}

\newcommand{\Sym}{\mathrm{Sym}}

\newcommand{\im}{\mathrm{im}}

\newcommand{\bp}{\begin{Proposition}}
\newcommand{\ep}{\end{Proposition}}
\newcommand{\bl}{\begin{Lemma}}
\newcommand{\el}{\end{Lemma}}
\newcommand{\bt}{\begin{Theorem}}
\newcommand{\et}{\end{Theorem}}
\newcommand{\bd}{\begin{Definition}}
\newcommand{\ed}{\end{Definition}}
\newcommand{\End}{\mathrm{End}}
\newcommand{\Aut}{\mathrm{Aut}}

\newcommand{\Mat}{\mathrm{Mat}}

\newcommand{\eqdef}{\stackrel{{\rm def.}}{=}}

\newcommand{\cinf}{{{\cal \cC}^\infty(M,\R)}}

\DeclareFontFamily{U}{rsf}{}
\DeclareFontShape{U}{rsf}{m}{n}{<5> <6> rsfs5 <7> <8> <9> rsfs7 <10-> rsfs10}{}
\DeclareMathAlphabet\Scr{U}{rsf}{m}{n}

\def\bu{{\bf u}}
\def\cU{\mathcal{U}}

\def\cW{\mathcal{W}}

\def\Z{\mathbb{Z}}
\def\C{\mathbb{C}}
\def\R{\mathbb{R}}
\def\bD{\mathbb{D}}

\def\rk{{\rm rk}}

\def\deg{{\rm deg}}

\def\GL{\mathrm{GL}}
\def\dd{\mathrm{d}}

\def\tf{\mathrm{tf}}

\def\Ad{\mathrm{Ad}}

\def\ad{\mathrm{ad}}

\def\sp{\mathrm{sp}}

\def\bcD{\boldsymbol{\mathcal{D}}}

\def\momega{{\boldsymbol{\omega}}}

\def\balpha{{\boldsymbol{\alpha}}}

\def\U{\mathrm{U}}

\def\bt{\mathbf{t}}

\newcommand{\be}{\begin{equation*}}
\newcommand{\ee}{\end{equation*}}
\newcommand{\ben}{\begin{equation}}
\newcommand{\een}{\end{equation}}
\newcommand{\beqa}{\begin{eqnarray*}}
\newcommand{\eeqa}{\end{eqnarray*}}
\newcommand{\beqan}{\begin{eqnarray}}
\newcommand{\eeqan}{\end{eqnarray}}

\newcommand{\nn}{\nonumber}

\newcommand{\id}{\mathrm{id}}
\newcommand{\Tr}{\mathrm{Tr}}

\newcommand{\diag}{\mathrm{diag}}

\def\Diff{\mathrm{Diff}}

\def\cC{{\mathcal C}}

\def\Gr{\mathrm{Gr}}
\def\cK{\mathrm{\cal K}}
\def\cM{\mathrm{\cal M}}

\def\O{\mathrm{O}}
\def\Hol{\mathrm{Hol}}
\def\hol{\mathrm{hol}}
\def\cD{\mathcal{D}}

\def\cE{\check{E}}
\def\cX{\mathcal{X}}

\def\cP{\mathcal{P}}
\def\cN{\mathcal{N}}
\def\cG{\mathcal{G}}
\def\cT{\mathcal{T}}
\def\cF{\mathcal{F}}
\def\cC{\mathcal{C}}
\def\cH{\mathcal{H}}

\def\Sp{\mathrm{Sp}}
\def\G_2{\mathrm{G_2}}

\def\cL{\mathcal{L}}
\def\cS{\mathcal{S}}

\def\cV{\mathcal{V}}

\def\mf{\mathbf{f}}

\def\Alg{\mathrm{Alg}}
\def\tw{\mathrm{tw}}
\def\G{\mathrm{G}}

\def\Im{\mathrm{Im}}
\def\Ric{\mathrm{Ric}}

\def\cE{\mathcal{E}}
\def\Isom{\mathrm{Isom}}

\def\fS{\mathfrak{S}}
\def\iOmega{\mathit{\Omega}}

\def\sp{\mathrm{sp}}
\def\Ab{\mathrm{Ab}}
\def\VB{\mathrm{VB}}
\def\bcE{{\boldsymbol{\mathcal{E}}}}
\def\Herm{\mathrm{Herm}}
\def\TamedSymp{\mathrm{TamedSymp}}
\def\Symp{\mathrm{Symp}}
\def\AbVar{\mathrm{AbVar}}

\def\cX{\mathcal{X}}

\def\grad{\mathrm{grad}}

\def\bD{\mathbf{D}}

\def\bd{\boldsymbol{\dd}}
\def\bast{\boldsymbol{\ast}}

\def\btau{\boldsymbol{\tau}}

\def\Pic{\mathrm{Pic}}
\def\rS{\mathrm{S}}

\def\rT{\mathrm{T}}

\def\Met{\mathrm{Met}}
\def\Iso{\mathrm{Iso}}
\def\Re{\mathrm{Re}}
\def\Im{\mathrm{Im}}
\def\fJ{\mathfrak{J}}
\def\SH{\mathbb{S}\mathbb{H}}
\def\i{\mathbf{i}}
\def\sc{\mathrm{sc}}
\def\EH{\mathrm{EH}}

\def\bfJ{\boldsymbol{\fJ}}
\def\bcL{\boldsymbol{\cL}}
\def\Ric{\mathrm{Ric}}
\def\Sol{\mathrm{Sol}}
\def\Conf{\mathrm{Conf}}
\def\ub{\mathrm{ub}}

\def\Open{\mathrm{Open}}
\def\Set{\mathrm{Set}}

\def\bwedge{\boldsymbol{\wedge}}

\def\bD{\mathbf{D}}

\def\cl{\mathrm{cl}}

\def\Gp{\mathrm{Gp}}
\def\loslash{\bm{\oslash}}
\def\Vect{\mathrm{Vect}}
\def\sp{\mathrm{sp}}
\def\Tors{\mathrm{Tors}}
\def\curv{\mathrm{curv}}
\def\bAd{\mathbf{Ad}}
\def\bDelta{{\boldsymbol{\Delta}}}
\def\bXi{{\boldsymbol{\Xi}}}
\def\bD{\boldsymbol{\cD}}
\def\Div{\mathrm{Div}}
\def\bdelta{\boldsymbol{\delta}}
\def\Tor{\mathrm{Tor}}
\def\TorSymp{\mathrm{TorSymp}}
\def\TorComp{\mathrm{TorComp}}
\def\Comp{\mathrm{Comp}}

\def\bcX{\boldsymbol{\cX}}
\def\SB{\mathrm{SB}}
\def\DS{\mathrm{DS}}
\def\ES{\mathrm{ES}}
\def\UES{\mathrm{UES}}
\def\rT{\mathrm{T}}
\def\TamedSymp{\mathrm{TamedSymp}}

\title{Generalized Einstein-Scalar-Maxwell theories and locally geometric U-folds}

\author{C.~I.~Lazaroiu$^1$ and C.~S.~Shahbazi$^2$}

\affiliation{$^1$ Center for Geometry and Physics, Institute for Basic
  Science, Pohang, Republic of Korea 37673\\
$^2$ Institut de Physique Th\'eorique, CEA-Saclay, France.}

\emailAdd{calin@ibs.re.kr, carlos.shabazi-alonso@cea.fr}

\abstract{We give the global mathematical formulation of the coupling of
  four-dimensional scalar sigma models to Abelian gauge fields on a Lorentzian four-manifold, for the generalized situation when the \emph{duality structure} of the Abelian gauge theory is described by a flat symplectic vector bundle $(\cS,D,\omega)$ defined over the scalar manifold $\cM$. The
  construction uses a taming of $(\cS, \omega)$, which we find to be the correct mathematical object globally encoding the inverse gauge couplings and theta angles of the ``twisted'' Abelian gauge theory in a manner that makes no use of  duality frames. We show that global solutions of the equations of
  motion of such models give classical locally geometric U-folds. We
  also describe the groups of duality transformations and
  scalar-electromagnetic symmetries arising in such models, which
  involve lifting isometries of $\cM$ to the bundle $\cS$ and hence
  differ from expectations based on local analysis. The appropriate
  version of the Dirac quantization condition involves a discrete
  local system defined over $\cM$ and gives rise to a smooth bundle of
  polarized Abelian varieties, endowed with a flat symplectic
  connection. This shows in particular that a generalization of part of the
  mathematical structure familiar from $\cN=2$ supergravity is already
  present in such purely bosonic models, without any coupling to fermions
  and hence without any supersymmetry.}
\preprint{}

\begin{document}

\maketitle

\pagebreak

\vskip .6in


\section*{Introduction}

The construction of four-dimensional Einstein-Scalar-Maxwell theories,
namely nonlinear sigma models defined on an oriented four-manifold $M$
and coupled to gravity and to Abelian gauge fields
\cite{GaillardZumino,Andrianopoli,AndrianopoliUduality,
  AndrianopoliFlat, Ortin} is usually performed using a
simply-connected ``scalar manifold'' $\cM$ in which the scalar map
$\varphi:M\rightarrow \cM$ is valued. This assumption implies that the
behavior of the Abelian gauge theory under electro-magnetic duality
transformations is governed by a $2n$-dimensional symplectic vector
space $(\cS_0,\omega_0)$. The gauge field strengths and their
Lagrangian conjugates combine in a 2-form defined on $M$ and valued in
$\cS_0$, which transforms under duality transformations through the
action of the symplectic group $\Aut(\cS_0,\omega_0)\simeq
\Sp(2n,\R)$.

When the scalar manifold $\cM$ is not simply connected, one can
generalize such theories by promoting the symplectic vector space
$(\cS_0,\omega_0)$ to a flat symplectic vector bundle
$\Delta=(\cS,D,\omega)$ defined over $\cM$; such bundles correspond to
symplectic representations of the fundamental group of $\cM$. The
freedom to modify the theory in this manner arises because Abelian
gauge theories defined on oriented Lorentzian four-manifolds $(M,g)$
can be ``twisted'' by a flat symplectic vector bundle\footnote{Even
  when $\rk\cS=2$, such a generalized Abelian gauge theory is not the
  ordinary Maxwell theory defined on $(M,g)$, but a ``twisted''
  version thereof.}, which in the sigma model case is the
$\varphi$-pullback of the {\em duality structure} $\Delta$. In such a
generalized sigma model coupled to Abelian gauge fields, no
globally-defined ``duality frames'' need exist since $\Delta$ need not
admit any globally-defined flat symplectic frame. As a consequence,
the formulation of the theory must be reconsidered. In particular,
various statements regarding global transformations \cite{symmetries}
have to be modified, since the groups of scalar-electromagnetic
duality transformations and of scalar-electromagnetic symmetries no
longer have a simple direct product structure. Moreover, the Dirac
quantization condition for the Abelian gauge sector changes.

In this paper, we give a general geometric construction of the
coupling of four-dimensional nonlinear scalar sigma models to gravity
and to Abelian gauge fields having a non-trivial duality structure
$\Delta$, by giving a global formulation of the equations of motion
which is both coordinate and frame-free. This uses a taming \cite{Vaisman,Berndt,DuffSalamon} $J$ of the symplectic vector bundle $(\cS,\omega)$, which encodes in a
globally-valid and frame-free manner the information carried by the
gauge-kinetic functions (equivalently, by the inverse coupling
constants and theta angles). The ``duality frames'' of the traditional
formulation are replaced by {\em locally defined} flat symplectic
frames of $\Delta$. The gauge-kinetic matrix is defined only locally
and depends on choosing a local flat symplectic frame $\cE$ supported
on an open subset $\cU\subset \cM$, giving a smooth
$\cE$-dependent function $\tau^\cE:\cU\rightarrow \SH_n$ valued in the
Siegel upper half space $\SH_n$. Equivalently, the gauge-kinetic
matrix becomes a function ${\tilde \tau}^{\cE_0}:{\tilde
  \cM}\rightarrow \SH_n$ defined on the universal cover ${\tilde \cM}$
of the scalar manifold and dependent on a choice of symplectic basis
$\cE_0$ of the typical fiber $(\cS_0,\omega_0)$ of $(\cS,\omega)$. The
second perspective amounts to viewing ${\tilde \tau}^{\cE_0}$ as a
``multivalued function'' defined on $\cM$, whose monodromy is given by
the action of the holonomy representation of the flat symplectic
connection $D$ through matrix fractional transformations. The
equations of motion of the model are expressed using the flat
symplectic vector bundle $\Delta$ and its taming $J$ and make no
direct reference to locally-defined or frame-dependent quantities. Therefore,  the formulation that we present in this work can be used to provide an \emph{explicitly U-duality covariant} formulation of the bosonic sector of four-dimensional ungauged supergravity.

We describe global scalar-electromagnetic duality transformations and
scalar-electromagnetic symmetries of such models, which involve
lifting isometries from $\cM$ to $\cS$ (a process which generally is
obstructed and non-unique), showing that the corresponding groups need
not have the direct product structure typical of the case when the
duality structure is trivial. We give a geometric formulation of the
Dirac quantization condition using the theory of local systems,
showing that it produces a smooth bundle of polarized Abelian
varieties defined over the scalar manifold $\cM$. The Dirac lattice is
replaced by a flat fiber sub-bundle $\Lambda\subset \cS$ whose fibers
$\Lambda_p$ at points $p\in \cM$ are full lattices inside the fibers
$\cS_p$ of $\cS$. The bundle of Abelian varieties has fibers
$\cS_p/\Lambda_p$, the polarization being induced by the symplectic
pairing $\omega_p$ of $\cS_p$ together with the pointwise value $J_p$
of the taming. This shows that a generalization of part of the
structure characteristic of the projective special K\"{a}hler geometry
governing $\cN=2$ supergravity theories \cite{Strominger, Freed,
  Cortes} is already present for general scalar sigma models coupled
to gravity and to Abelian gauge fields, except that various conditions
relevant to the $\cN=2$ case are relaxed. The cubic form of
\cite{Freed} also has a generalization to our models, being replaced
by the {\em fundamental form} $\Theta$ of the {\em electromagnetic
  structure} $\Xi=(\cS,D,J,\omega)$. This is an $End(\cS)$-valued
1-form defined on $\cM$ which measures the failure of $J$ to be
covariantly constant with respect to $D$. As shown in the present
paper, a non-trivial duality structure is already allowed when the
model does not include fermions and hence cannot be supersymmetric. We
stress that such a structure {\em must}\footnote{Indeed, the bosonic
  sector of the $\cN=2$ theory should be a particular case of the
  bosonic sector of the $\cN=0$ theory.} be present if one wishes to
view $\cN=2$ supergravity coupled to matter as a particular case of
the $\cN=0$ theory, since a (highly constrained) non-trivial flat
symplectic vector bundle must be present \cite{Strominger, Freed, Cortes}
in the $\cN=2$ case.

We also show that globally-defined solutions of such models can be
interpreted as classical locally-geometric U-folds. This
interpretation --- which realizes proposals of \cite{GeometricUfolds}
--- already arises in the absence of any coupling to spinors (in
particular, it holds without the need for any supersymmetry), being a
topological consequence of the fact that the duality structure need
not be trivial. As we show in a separate paper, the models
constructed here can be further generalized such as to produce a
global geometric description of all classical locally-geometric
U-folds.

The paper is organized as follows. Section \ref{sec:sigmagravity}
recalls the global formulation of four-dimensional scalar sigma models
coupled to gravity. Section \ref{sec:scalarsigmamodel} discusses
duality and electromagnetic structures and gives the geometric
formulation of the general coupling of scalar sigma models to gravity
and to Abelian gauge fields, when the latter are governed by an
arbitrary duality structure. Section \ref{sec:duality} discusses
dualities and scalar-electromagnetic symmetries of such models, whose
formulation involves lifting isometries of the scalar manifold to the
vector bundle $\cS$. Section \ref{sec:Dirac} discusses the Dirac
quantization condition, showing that the Dirac lattice is replaced by
a discrete local system, giving rise rise to a smooth bundle of
polarized Abelian varieties defined over the scalar manifold
$\cM$. Section \ref{sec:GeometricU} discusses the U-fold
interpretation of global solutions of such theories as well as the
particular case when $(\cS,\omega)$ is symplectically trivial but
equipped with a flat symplectic connection of non-trivial
holonomy. Section \ref{sec:relation} contains a brief comparison with
the literature while Section \ref{sec:conclusions}
concludes. Appendices \ref{app:space_tamings} and
\ref{app:bundle_tamings} summarize information about tamings and
positive polarizations of symplectic vector spaces and symplectic
vector bundles. Appendix \ref{app:twistedconstructions} discusses
certain twisted constructions relevant to such models. Appendix
\ref{app:unbased} explains various constructions involving unbased
automorphisms of vector bundles while appendix \ref{app:localeom}
shows that our global models are {\em locally} indistinguishable from
those found in the literature. Appendix \ref{app:integral} discusses 
various categories of integral spaces and tori which are relevant to 
the formulation of Dirac quantization for such models. 

\paragraph{Notations and conventions.}

The symbols $\Gp,\Ab$ and $\Alg$ denote respectively the categories of
all groups, of Abelian groups and of unital $\R$-algebras. The symbol
$\Vect$ denotes the category of finite-dimensional $\R$-vector spaces
and linear maps. For any category $\cC$, let $\cC^\times$ denote its unit
groupoid, i.e. the category having the same objects as $\cC$ and whose
morphisms are the isomorphisms of $\cC$.  Given an Abelian group $A$,
let $\Tors(A)$ denote its torsion subgroup and $A_\tf\eqdef
A/\Tors(A)$ be the corresponding torsion-free group; the latter is
free when it is finitely-generated.

All manifolds considered are smooth, connected, Hausdorff and
paracompact (hence also second countable) while all bundles and
sections considered are smooth. Given a manifold $N$, let
$\Met_{p,q}(N)$ denote the set of all smooth pseudo-Riemannian metrics
of signature $(p,q)$ defined on $N$. Let $\cP(N)$ denote the set of
paths (=piecewise-smooth curves) $\gamma:[0,1]\rightarrow N$ and
$\Pi_1(N)$ denote the first homotopy groupoid of $N$. Given a vector
space $V$, let $V_N$ denote the corresponding trivial vector bundle
defined on $N$. Let $\cX(N)$ denote the $\cC^\infty(N,\R)$-module of
globally-defined smooth vector fields on $N$. By definition, a
Lorentzian manifold is a pseudo-Riemannian 4-manifold $(N,g)$ whose
pseudo-Riemannian metric $g$ has ``mostly plus'' signature
$(p,q)=(3,1)$.

Given vector bundles $\cS$ and $\cS'$ over a manifold $N$, we denote
by $Hom(\cS,\cS')$ and $Isom(\cS,\cS')$ the bundles of based linear
morphisms and isomorphisms from $\cS$ to $\cS'$ and by
$\Hom(\cS,\cS')\eqdef \Gamma(N,Hom(\cS,\cS'))$ and
$\Isom(\cS,\cS')\eqdef \Gamma(N,Isom(\cS,\cS'))$ the sets of based
vector bundle morphisms and isomorphisms. When $\cS'=\cS$, we set
$End(\cS)\eqdef Hom(\cS,\cS)$, $Aut(S)\eqdef Isom(\cS,\cS)$ and
$\End(\cS)\eqdef \Hom(\cS,\cS)$, $\Aut(S)\eqdef \Isom(\cS,\cS)$. We
use the notations $\Hom^\ub(\cS,\cS')$, $\Isom^\ub(\cS,\cS')$,
$\End^\ub(\cS)$ and $\Aut^\ub(\cS)$ for the sets of {\em unbased}
morphisms, isomorphisms etc.

Given a smooth map $f:N_1\rightarrow N_2$ and a vector bundle $\cS$ on
$N_2$, we denote the $f$-pullback of $\cS$ to $N_1$ by $\cS^f$.  Given
a smooth section $\xi\in \Gamma(N_2,\cS)$, we denote its $f$-pullback
by $\xi^f\in \Gamma(M_1,\cS^f)$, calling it the {\em topological
  pullback}. We reserve the notation $f^\ast$ for the {\em
  differential pullback} of tensors defined on $N_2$, an operation
which generally involves both the topological pull-back along $f$ and
insertions of the differential $\dd f$; an example is the differential
pull-back of differential forms valued in a vector bundle.

By definition, a Hermitian pairing is complex-linear in the first variable 
while an {\em anti-Hermitian pairing} is complex linear in
the second variable. A complex-valued pairing is anti-Hermitian iff its
conjugate is Hermitian. The canonical symplectic pairing $\omega_n$
on $\R^{2n}$ is defined through:
\ben
\label{omegan}
\omega_n((x_1,y_1),(x_2,y_2))=x_1y_2-y_1x_2~~\forall x_1,x_2,y_2,y_2\in \R^n
\een
and its matrix in the canonical basis is denoted by: 
\begin{equation*}
\label{Omegan}
\Omega_n\eqdef \left[\begin{array}{cc} 0 & I_n \\ -I_n & 0 \end{array} \right]\, .
\end{equation*}
The space of matrices of size $m\times n$ with coefficients in a ring $R$ is 
denoted by $\Mat(m,n,R)$ while the space of square matrices of size $n$ 
is denoted by $\Mat(n,R)\eqdef \Mat(n,n,R)$. The space of 
symmetric square matrices of size $n$ with coefficients in $R$
is denoted by $\Mat_s(n,R)$. The complex imaginary unit is denoted by $\i$. 

\section{Scalar nonlinear sigma models coupled to gravity}
\label{sec:sigmagravity}

We start by recalling the global construction of four-dimensional
scalar nonlinear sigma models coupled to gravity.  Let $M$ be an
oriented four-manifold and $\Met_{3,1}(M)$ be the set of Lorentzian
metrics defined on $M$. We assume that $\Met_{3,1}(M)$ is non-empty
(this requires $\chi(M)=0$ when $M$ is compact).

\subsection{Preparations}

The musical isomorphism $TM\simeq T^\ast M$ defined by $g\in
\Met_{3,1}(M)$ induces an isomorphism $T^\ast M\otimes T^\ast
M\simeq End(TM)$. Combining this with the fiberwise trace
$\Tr:End(TM)\rightarrow \R_M$ gives a map $\Tr_g:T^\ast M\otimes
T^\ast M\rightarrow \R_M$. Let $\nu:\Met_{3,1}(M)\rightarrow
\Omega^4(M)$ be the functional which associates to an element $g\in
\Met_{3,1}(M)$ its normalized volume form $\nu(g)$ and
$\mathrm{R}:\Met_{3,1}(M)\rightarrow \cC^\infty(M,\R)$ be the
functional which associates to $g$ its scalar curvature
$\mathrm{R}(g)$.

\begin{Definition} 
A {\em scalar structure} is a triplet $\Sigma=(\cM,\cG,\Phi)$, where
$(\cM,\cG)$ is a Riemannian manifold (called the {\em scalar
  manifold}) and $\Phi\in \cC^\infty(\cM,\R)$ is a smooth real-valued
function defined on $\cM$ (called the {\em scalar potential}).
\end{Definition}

\noindent Given a scalar structure $\Sigma=(\cM,\cG,\Phi)$, we denote
by $\Sigma_0\eqdef (\cM,\cG)$ the underlying scalar manifold.  Let
$\varphi^\ast(\cG)\in \Gamma(M,\Sym^2(T^\ast M))$ be the differential
pullback of $\cG$ through $\varphi$, also known as the {\em first
  fundamental form} of $\varphi$ (see \cite{BairdWood}). This is the
symmetric covariant 2-tensor defined on $M$ through the expression:
\ben
\varphi^\ast(\cG)_{m}(u,v)\eqdef \cG_{\varphi(m)}(\dd_m\varphi(u),\dd_m\varphi(v))
\een
for all $m\in M$ and $u,v\in T_mM$. The topological pullback of $\cG$
defines a Euclidean pairing $\cG^\varphi:(T\cM)^\varphi\otimes
(T\cM)^\varphi\rightarrow \R_M$ on the pulled-back bundle
$(T\cM)^\varphi$. The differential $\dd\varphi:TM\rightarrow T\cM$ is
an unbased morphism of vector bundles covering $\varphi$, which induces a
based morphism $\widehat{\dd\varphi}:TM\rightarrow (T\cM)^\varphi$; in
turn, this can be viewed as a 1-form $\widetilde{\dd
  \varphi}\in\Omega^1(M,(T \cM)^\varphi)=\Gamma(M,T^\ast M\otimes
(T\cM)^\varphi)$ defined on $M$ and valued in the pulled-back vector
bundle $(T\cM)^\varphi$. Then:
\be
\varphi^\ast(\cG)(X,Y)=\cG^\varphi(\widetilde{\dd\varphi}(X)\otimes \widetilde{\dd\varphi}(Y))~~\forall X,Y\in \cX(M)~~,
\ee
where $\widetilde{\dd\varphi}(X)$ and $\widetilde{\dd\varphi}(Y)$ are
sections of $(T\cM)^\varphi$. The {\em density} of $\varphi$ relative
to $\cG$ and $g\in \Met_{3,1}(M)$ is the function \cite{BairdWood}:
\be
e_\cG(g, \varphi)\eqdef \frac{1}{2}\Tr_g\varphi^\ast(\cG)\in \cC^\infty(M,\R)~~.
\ee

\begin{Definition}
The {\em modified density} of a smooth map $\varphi\in
\cC^\infty(M,\cM)$ relative to the Lorentzian metric $g\in
\Met_{3,1}(M)$ and to the scalar structure $\Sigma$ is defined
through:
\ben 
e_\Sigma(g,\varphi)\eqdef e_\cG(g,\varphi)+\Phi^\varphi=\frac{1}{2}\Tr_g \varphi^\ast(\cG)+\Phi^\varphi\in \cC^\infty(M,\R)~~,
\een
where $\Phi^\varphi=\Phi \circ \varphi\in \cC^\infty(M,\R)$.
\end{Definition}

\noindent Together with the Levi-Civita connection $\nabla^g$ of
$(M,g)$, the $\varphi$-pullback $(\nabla^\cG)^\varphi$ of the
Levi-Civita connection of $(\cM,\cG)$ induces a connection on the
vector bundle $T^\ast M\otimes (T\cM)^\varphi$, which we denote by
$\nabla$. The {\em second fundamental form} (see \cite{BairdWood}) of
the map $\varphi$ relative to $g$ and $\cG$ is the
$(T\cM)^\varphi$-valued tensor field on $M$ given by:
\be
\nabla\widetilde{\dd\varphi}\in \Omega^1(M,T^\ast M \otimes (T\cM)^\varphi)=\Gamma(M,T^\ast M\otimes T^\ast M\otimes (T\cM)^\varphi)~~,
\ee
i.e.:
\be
(\nabla \widetilde{\dd\varphi})(X,Y)=(\nabla^\cG)_X^\varphi(\widetilde{\dd\varphi}(Y))-\widetilde{\dd\varphi}(\nabla_X^g Y)~~\forall X,Y\in \cX(M)~~.
\ee
One can show that $\nabla\widetilde{\dd\varphi}$ is a section of the bundle
$\Sym^2(T^{\ast}M)\otimes (T \cM)^\varphi$. The {\em tension field} of
$\varphi$ relative to $g$ and $\cG$ is defined through (see
loc. cit.):
\ben
\label{tension}
\theta_\cG (g,\varphi)\eqdef \Tr_g \nabla\widetilde{\dd\varphi} \in \Gamma(M, (T\cM)^\varphi)~~.
\een

\begin{Definition}
The {\em modified tension field} of a smooth map $\varphi\in
\cC^\infty(M,\cM)$ relative to the Lorentzian metric $g\in \Met_{3,1}(M)$ and
to the scalar structure $\Sigma$ is defined through:
\ben
\theta_\Sigma(g,\varphi)\eqdef \theta_\cG(g,\varphi)-(\grad_\cG \Phi)^\varphi\in \Gamma(M, (T\cM)^\varphi)~~,
\een
where $\grad_\cG \Phi\in \cX(\cM)$ is the gradient vector field of
$\Phi$ with respect to $\cG$. 
\end{Definition}

\subsection{The action functional}
 
For any relatively compact open subset $U\subset M$, the action of the
scalar sigma model coupled to gravity determined by $\Sigma$ and
defined on $M$ is the functional
$S_\Sigma^U:\Met_{3,1}(M)\times\cC^\infty(M,\R)\rightarrow \R$ defined
through:
\begin{equation}
\label{S0}
S_\Sigma^U[g,\varphi] = S_{\EH}^U[g]+S_{\sc,\Sigma}^U[g, \varphi]\, ,
\end{equation} 
where:
\be
\label{SEH}
S_{\EH}^U[g]=\frac{1}{2\kappa}\int_U\nu(g) \mathrm{R}(g)\, ,
\ee
is the Einstein-Hilbert action ($\kappa$ being the Einstein constant) and:
\ben
\label{Ssc0}
S_{\sc,\Sigma}^U[g, \varphi]= - \int_U \nu(g) e_\Sigma(g,\varphi)\, .
\een
is the scalar sigma model action determined by $\Sigma$ in the
background metric $g\in \Met_{3,1}(M)$.

\subsection{Local expressions}

Let $d\eqdef \dim \cM$. Consider local charts $(U, x^\mu)_{\mu=0\dots
  3}$ on $M$ and $(\cU, y^A)_{A=1\ldots d}$ on $\cM$ chosen such that
$\varphi(U)\subset \cU$. Let $x=(x^0,\ldots x^3)$ and let
$y^A=\varphi^A(x)$ be the local expression of the map $\varphi$ in
these coordinates. Let $\cG_{AB}\in\cC^\infty(\cU,\R)$ be the local
coefficients of $\cG$ in the chart $(\cU,y^A)$, defined in the
convention $\cG|_\cU=\cG_{AB}\dd y^A\otimes \dd y^B$. Then:
\be
\varphi^\ast(\cG)(x)=_{U} \cG_{AB}(\varphi(x)) \partial_\mu \varphi^A(x)\partial_\nu \varphi^B(x) \dd x^\mu\otimes \dd x^\nu 
\ee
and:
\be
e_\Sigma(g, \varphi)(x)=_{U}\frac{1}{2}\cG_{AB}(\varphi(x))g^{\mu\nu}(x) \partial_\mu \varphi^A(x)\partial_\nu \varphi^B(x)+\Phi(\varphi(x))~~.
\ee
We also have $\theta_\cG(g,
\varphi)=_U\theta^A_\cG(g,\varphi)\partial_A^\varphi$, where
$\partial_A^\varphi\eqdef (\partial_A)^\varphi|_U$ form a local frame
of the bundle $(T\cM)^\varphi$ above $U$. Let $\Gamma^C_{AB}$ and
$K^\rho_{\mu\nu}$ be the Christoffel symbols of $\cG$ and $g$ computed
respectively in the chosen coordinates on $\cM$ and $M$:
\beqa
&& \nabla^\cG(\frac{\partial}{\partial y^A})=_\cU \Gamma_{AB}^C \frac{\partial}{\partial y^C}\nn\\
&& \nabla^g(\frac{\partial}{\partial x^\mu})=_U K_{\mu\nu}^\rho \frac{\partial}{\partial x^\rho}~~,
\eeqa
where $\nabla^\cG$ and $\nabla^g$ are the Levi-Civita connections of
$\cG$ and $g$, respectively. Then:
\ben
\label{taulocal}
\theta^A_\cG(g,\varphi)(x)=_U\partial^\mu \partial_\mu
\varphi^A(x)+\Gamma^A_{BC}(\varphi(x))\partial^\mu \varphi^B(x)\partial_\mu
\varphi^C(x)-g^{\mu\nu}(x)K^\rho_{\mu\nu}(x)\partial_\rho\varphi^A(x)~~.
\een
and:
\ben
\theta^A_\Sigma(g,\varphi)(x)=\theta^A_\cG(g,\varphi)(x)-\cG^{AB}(\varphi(x))(\partial_B\Phi)(\varphi(x))~~.
\een

\subsection{The equations of motion}

Let $\Ric(g)$ and $\G(g) \eqdef \Ric(g)-\frac{1}{2}\mathrm{R}(g)g $ be
the Ricci and Einstein tensors of $g$. 

\begin{Proposition}
The equations of motion (critical point equations) for $g$ and
$\varphi$ derived from the action functional \eqref{S0} are the {\em
  Einstein-Scalar (ES) equations}:
\beqan
\label{eom0}
&& \G(g) = \kappa \rT_\Sigma(g,\varphi)\, ,\nn \\
&& \theta_{\Sigma}(g,\varphi)  = 0\, ,
\eeqan
where: 
\ben
\label{TSigma}
\rT_\Sigma(g,\varphi)\eqdef \varphi^\ast(\cG) - e_\Sigma(g,\varphi)g \in \Gamma(M,\Sym^2(T^\ast M))\, ,
\een
is the {\em scalar stress-energy tensor}.
\end{Proposition}

\begin{proof}
Under an infinitesimal variation of $g$ and $\varphi$, we have:
\be
(\delta S_{\EH}^U) [g]=\frac{1}{2\kappa}\int_U \nu(g)\G_{\mu\nu}(g) \delta g^{\mu\nu}~~,
\ee
and:
\be
(\delta S_{\sc,\Sigma}^U) [g,\varphi]= \int_{U} \nu(g)\cG^\varphi(\delta\varphi, \theta_\Sigma(g,\varphi))-\frac{1}{2}\int_{U}\nu(g)\rT_{\mu\nu}(g,\varphi)\delta g^{\mu\nu}~~,
\ee
where the stress-energy tensor is given by:
\be
\rT_{\mu\nu}[g,\varphi]= - g_{\mu\nu} e_\Sigma(g,\varphi)+2 \frac{\delta e_\Sigma(g,\varphi)}{\delta g^{\mu\nu}}= \varphi^\ast(\cG)_{\mu\nu} - g_{\mu\nu} e_\Sigma(g,\varphi)\, .
\ee
Thus: 
\be
(\delta S_\Sigma^U)[g,\varphi]=\int_{U} \nu(g)\cG^\varphi(\delta\varphi, \theta_\Sigma(g,\varphi))+\frac{1}{2\kappa}\int_{U}\nu(g)\left[\G_{\mu\nu}(g)-\kappa \rT_{\mu\nu}(g,\varphi)\right]\delta g^{\mu\nu}~~.
\ee
Hence the critical points of $S_\Sigma^U$ are solutions of the equations: 
\be
\G_{\mu\nu}(g)=\kappa \rT_{\mu\nu}(g,\varphi)~~,~~\theta_\Sigma(g,\varphi)=0~~.
\ee
\end{proof}

\noindent Using the expressions of the previous subsection, one finds
that the equations of motion \eqref{eom0} coincide with the local
forms usually given in the physics literature.

\begin{Remark}
When $\Phi=0$, the second equation in \eqref{eom0} reduces to the {\em
  pseudo-harmonicity equation}:
\ben
\label{sceom0}
\theta_\cG(g, \varphi)=0~~,
\een
whose solutions are called {\em pseudoharmonic maps} (or wave maps)
from $(M,g)$ to $(\cM,\cG)$. The theory of pseudo-harmonic maps is a
classical subject (see \cite{BairdWood, KriegerSchlag}).
Most mathematical literature on Einstein equations coupled to scalar
fields considers only the case when the scalar manifold is given by
$\cM=\R$, endowed with a translationally-invariant metric $\cG$ on
$\R$ (which can be taken to be the standard metric).  In that
particular case, the system \eqref{eom0} is sometimes called the {\em
  Einstein-nonlinear scalar field system}. The initial value problem
for that particular case is discussed in \cite{CB1,CB2,CB3}.
\end{Remark} 

\subsection{Sheaves of configurations and solutions}

For any subset $U\subset M$, let $\Conf_\Sigma(U)\eqdef
\Met_{3,1}(U)\times \cC^\infty(U,\cM)$ and let $\Sol_\Sigma(U)\subset
\Conf_\Sigma(U)$ denote the set of smooth solutions of equations
\eqref{eom0} defined on $U$, where $U$ is endowed with the open
submanifold structure induced from $M$. Let $\Open(M)$ be the small
site\footnote{The collection of all open subsets of $M$ viewed as a
  category whose morphisms are inclusion maps.} of $M$. Since
equations \eqref{eom0} are local, the restriction to an open subset of
a solution is a solution. This gives contravariant functors
$\Conf_\Sigma:\Open(M)\rightarrow \Set$ and
$\Sol_\Sigma:\Open(M)\rightarrow \Set$ which satisfy the gluing
conditions and thus are sheaves of sets defined on $M$, called the
{\em sheaves of Einstein-Scalar (ES) configurations and solutions},
respectively. Given a metric $g\in \Met_{3,1}(M)$, we also have a
sheaf $\Conf_\Sigma^\sc$ of {\em scalar configurations} (defined
through $\Conf_\Sigma^\sc(U)\eqdef \cC^\infty(U,\cM)$) and a sheaf
$\Sol^g_{\Sigma,\sc}$ of {\em scalar solutions} relative to a metric
$g\in \Met_{3,1}(M)$, where $\Sol^g_{\Sigma,\sc}(U)$ is the set of all
solutions of the second equation in \eqref{eom0} which are defined on
$U$. The sheaf $\Sol_{\Sigma,\sc}^g$ is relevant when treating $g$ as 
a background metric. 

\subsection{Target space symmetries}

Let: 
\be
\Iso(\cM,\cG)\eqdef \{\psi\in \Diff(\cM)|\psi^\ast(\cG)=\cG\}
\ee
be the group of isometries of $(\cM,\cG)$, which is a Lie group by the
second Steenrod-Myers theorem. This group acts from the left on the
set $\Conf_\Sigma^\sc(M)=\cC^\infty(M,\cM)$ through postcomposition:
\ben
\label{scalaraction}
\psi \varphi \eqdef \psi \circ \varphi
\een
The action \eqref{scalaraction} preserves $\Phi^\varphi$ for all
$\varphi\in \cC^\infty(M,\cM)$ if and only if\footnote{To see this, it
  suffices to take $\varphi$ to be a constant map.} $\Phi\circ
\psi=\Phi$.

\begin{Definition}
The {\em automorphism group} of the scalar structure
$\Sigma=(\cM,\cG,\Phi)$ is defined through:
\be
\Aut(\Sigma) \eqdef \left\{\psi\in\Iso(\cM,\cG) | \Phi\circ \psi=\Phi\right\} \subset \Iso(\cM, \cG)~~.
\ee
\end{Definition}

\noindent Notice that $\Aut(\Sigma)$ is a closed subgroup of
$\Iso(\cM,\cG)$ and hence is a Lie group.

\begin{Proposition}
\label{prop:thetatf}
Let $\varphi\in \cC^\infty(M,\cM)$. For any $\psi\in\Iso(\cM,\cG)$, we have: 
\beqan
\label{thetatf}
&& e_\cG(g,\psi\varphi)=e_\cG(g,\varphi)~~\nn\\
&& \theta_\cG(g, \psi \varphi)=(\widehat{\dd \psi})^\varphi \circ \theta_\cG(g, \varphi)~~,
\eeqan
where $\widehat{\dd\psi}\in
\Hom(T\cM,(T\cM)^\psi)$ is the based  morphism of vector bundles 
induced by $\dd\psi$ and $(\widehat{\dd\psi})^\varphi\in
\Hom((T\cM)^\varphi,(T\cM)^{\psi\circ \varphi})$ is its
$\varphi$-pullback to $M$. When $\psi\in \Aut(\Sigma)$, we have:
\beqan
\label{thetatf2}
&& e_\Sigma(g,\psi\varphi)=e_\Sigma(g,\varphi)~~\nn\\
&& T_\Sigma(g,\psi\varphi)=T_\Sigma(g,\varphi)~~\nn\\
&& \theta_\Sigma(g,\psi\varphi)=(\widehat{\dd \psi})^\varphi \circ \theta_\Sigma(g, \varphi)~~.
\eeqan

\end{Proposition}

\begin{proof}
Follows by direct computation. 
\end{proof}

\noindent The proposition implies that the equations of motion
\eqref{eom0} are invariant under the action of the group
$\Aut(\Sigma)$. It follows that this action preserves the set
$\Sol_\Sigma(M)$ of global Einstein-Scalar solutions as well as the
set $\Sol_{\Sigma,\sc}^g(M)$ of global scalar solutions relative to
any metric $g\in \Met_{3,1}(M)$.

\section{Generalized scalar sigma models coupled to gravity and to Abelian gauge fields} 
\label{sec:scalarsigmamodel}

In this section we give the global mathematical formulation of the equations of motion of a scalar sigma model coupled to gravity and to an arbitrary number $n$ of Abelian gauge fields --- also known as the Einstein-Scalar-Maxwell (ESM) theory --- for the generalized case when the Abelian gauge fields have a non-trivial ``duality structure''. This theory is {\em locally} equivalent with that discussed in the supergravity literature, to which it reduces globally only when the duality structure is trivial. As we explain in a separate publication, the generalized models constructed in this section are physically allowed because the Maxwell equations defined on a non-simply connected oriented Lorentzian four-manifold can be ``twisted'' by a flat symplectic vector bundle.

\subsection{Introductory remarks} 
Before introducing the mathematical formalism, it may be worth
outlining how it captures the local description found in the
supergravity literature. Let $(M,g)$ be a Lorentzian four-manifold and
$\Sigma=(\cM,\cG,\Phi)$ be a scalar structure. When coupling the
scalar sigma model defined by $\Sigma$ on $(M,g)$ to $n$ Abelian gauge
fields with a trivial ``duality structure'', one specifies a symmetric
complex-valued gauge-kinetic matrix $\tau(p)=\theta(p)+i\gamma(p)\in
\Mat_s(n,\C)$ for each point $p$ of the scalar manifold $\cM$. The
real-valued symmetric matrices $\gamma(p)$ and $\theta(p)$ encode
respectively the inverse gauge couplings and the theta angles of the
Abelian gauge theory. Since the kinetic term in the electromagnetic
action written in any duality frame must produce a positive-definite
energy density, the imaginary part $\gamma(p)$ must be strictly
positive-definite, hence $\tau(p)$ is an element of the Siegel upper
half space $\SH_n$. This gives a function $\tau:\cM\rightarrow \SH_n$,
which we take to be smooth. Traditionally \cite{GaillardZumino}, one
arranges the Abelian gauge field strengths defined on $M$ into a
vector-valued two-form ${\hat F} \in \Omega^2(M,\R^n)$ (the
generalized ``Faraday tensor'') and associates to it the vector-valued
2-form ${\hat G} \eqdef \theta\circ\varphi|_U {\hat F} -\gamma\circ
\varphi|_U ~\ast_g {\hat F} \in \Omega^2(M,\R^n)$ of ``Lagrangian
conjugate field strengths'', where $\ast_g$ is the Hodge operator of
$(M,g)$. The vector-valued 2-form ${\hat \cV}={\hat F}\oplus {\hat
  G}\in\Omega^2(M,\R^{2n})$ obeys the algebraic {\em polarization
  condition}:
\be
\ast_g {\hat \cV}=- ({\hat J}\circ \varphi)~{\hat \cV} ~~,
\ee
where ${\hat J}:\cM\rightarrow \Mat(2n,\R)$ is the smooth matrix-valued
function given by:
\be
{\hat J}(p)\eqdef \left[\begin{array}{cc} \gamma(p)^{-1}\theta(p)~~&
    -\gamma(p)^{-1} \\ \gamma(p)+\theta(p)\gamma(p)^{-1}\theta(p) &~~
    -\theta(p)\gamma(p)^{-1}\end{array}\right]~~.
\ee
Consider the modified Hodge operator $\star_{g,{\hat J},\varphi}\eqdef
\ast_g\otimes {\hat J}^\varphi:\Omega^k(M,\R^{2n})\rightarrow
\Omega^{4-k}(M,\R^{2n})$ where ${\hat J}^\varphi\eqdef {\hat J}\circ
\varphi\in \cC^\infty(M,\Mat(2n,\R))$. This operator squares to the
identity on $\Omega^2(M,\R^{2n})$ and the polarization condition
states that ${\hat \cV}$ is self-dual with respect to
$\star_{g,J,\varphi}$:
\be
\star_{g,J,\varphi} {\hat \cV}={\hat \cV}~~.
\ee
The Maxwell equations for the field strengths (including the Bianchi
identities) amount to the condition $\dd {\hat \cV}=0$, which we shall
call the {\em electromagnetic equation}. This formulation allows one
to view ${\hat \cV}$ as a self-dual vector-valued 2-form field defined
on $M$, where self-duality is understood with respect to the modified
Hodge operator $\star$ (which, beyond the space-time metric $g$,
depends on both $J$ and $\varphi$). Under an electro-magnetic duality
transformation, ${\hat \cV}$ changes by a symplectic transformation
acting on $\R^{2n}$ (where $\R^{2n}$ is endowed with its canonical
symplectic form \eqref{omegan}), while the gauge-kinetic function
$\tau$ changes by the corresponding matrix modular transformation
\cite{GaillardZumino,Andrianopoli}. To obtain the global mathematical formulation of this construction, it suffices to notice (see \cite{Vaisman,
  Berndt}) that a choice of symplectic basis of a real symplectic
vector space $(\cS_0,\omega_0)$ of dimension $2n$ induces a bijection
between the Siegel upper half space $\SH_n$ and the space of tamings
of $(\cS_0,\omega_0)$, which takes an element $\tau(p)\in \SH_n$ into
the taming $J(p)\in \End_\R(\cS_0)$ whose matrix ${\hat J}(p)$ in the
given symplectic basis is expressed in terms of $\tau$ through the
formula given above. This means that $\tau$ can be recovered uniquely
from the map $\cM\ni p\rightarrow J(p)\in \End_\R(\cS_0)$ and from a
choice of symplectic basis of $(\cS_0,\omega_0)$. When $J$ is fixed,
changes of symplectic basis of $(\cS_0,\omega_0)$ correspond to matrix
modular transformations of $\tau$ (see Appendix
\ref{app:space_tamings}). In particular, symplectic bases of
$(\cS_0,\omega_0)$ can be identified with ``duality frames''.  The
polarization condition is independent of the choice of ``duality
frame'' and the duality-invariant information contained in $\tau$ is
encoded by the smooth map $J:\cM\rightarrow \fJ_+(\cS_0,\omega_0)$,
where $\fJ_+(\cS_0,\omega_0)$ is the space of tamings of
$(\cS_0,\omega_0)$.

In the generalized theory constructed below, the symplectic vector
space $(\cS_0,\omega_0)$ is promoted to a {\em duality structure},
i.e. to a flat symplectic vector bundle $\Delta=(\cS,D,\omega)$
defined on $\cM$ and having typical fiber $(\cS_0,\omega_0)$, where
$D$ is the flat $\omega$-compatible connection on $\cS$. The taming of
$(\cS_0,\omega_0)$ is promoted to a taming $J$ of the symplectic
vector bundle $(\cS,\omega)$ (where $J$ need {\em not} be
covariantly-constant with respect to $D$). The choice of such a bundle
taming defines an {\em electromagnetic structure}
$\Xi=(\cS,D,J,\omega)$ on $\cM$. The vector-valued 2-form ${\hat \cV}$
is promoted to a 2-form $\cV\in \Omega^2(M,\cS^\varphi)$ defined on
space-time and valued in the $\varphi$-pullback of the vector bundle
$\cS$. This 2-form is subject to a generalized version of the
polarization condition which can be expressed in a frame-free manner
using the $\varphi$-pullback $J^\varphi$ of the bundle taming:
\be
\ast_g\cV=-J^\varphi \cV~~\mathrm{i.e.}~~\star_{g,J,\varphi}\cV=\cV~~,
\ee
where $\star_{g,J,\varphi}\eqdef \star_g\otimes
J^\varphi:\Omega^k(M,\cS)\rightarrow \Omega^{4-k}(M,\cS)$ is the Hodge
operator of $(M,g)$ ``twisted'' by the complex vector bundle
$(\cS^\varphi, J^\varphi)$. Moreover, $\cV$ obeys the following
version of the electromagnetic equations:
\be
\dd_{D^\varphi}\cV=0~~,
\ee
where $\dd_{D^\varphi}:\Omega^k(M,\cS^\varphi)\rightarrow
\Omega^{k+1}(M,\cS^\varphi)$ is the exterior differential of $M$
twisted by the pulled-back flat connection $D^\varphi$. In this
generalized theory, ``duality frames'' correspond to {\em local} flat
symplectic frames of $\Delta$. The generalized theory reduces to the
ordinary one when $D$ is a trivial flat connection, i.e. when $D$ has
trivial holonomy along each loop in $\cM$. In particular, it is
necessarily equivalent with the latter when $\cM$ is
simply-connected. When $\pi_1(M)\neq 1$ and $D$ has non-trivial
holonomy, no globally-defined ``duality frame'' exists and the
generalized theory is only {\em locally} equivalent with the ordinary
one.  In that case, the Abelian gauge theory coupled to a background
scalar field $\varphi$ can be understood by ``patching'' ordinary
theories using electro-magnetic duality transformations and in this
sense its global solutions are of ``U-fold'' type
(cf. \cite{GeometricUfolds}).

We now proceed to give the mathematical formulation, referring the
reader to Appendices \ref{app:space_tamings} and
\ref{app:bundle_tamings} for information on tamings of symplectic
vector spaces and of vector bundles. The proof of the fact that the
generalized theory defined below is locally equivalent with that
described in the supergravity literature can be found in Appendix
\ref{app:localeom}.

\subsection{Duality structures}

Let $N$ be a manifold.

\paragraph{The category of duality structures defined on $N$.}

\begin{Definition}
A {\em duality structure} (or {\em flat symplectic vector bundle})
defined on $N$ is a triple $\Delta=(\cS,D,\omega)$, where:
\begin{enumerate}[(a)]
\itemsep 0.0em
\item $\cS$ is a real vector bundle of even rank defined on $N$.
\item $\omega:\cS\times_N\cS\rightarrow\R_N$ is a symplectic pairing
  on $\cS$, i.e. a non-degenerate and antisymmetric fiberwise-bilinear
  pairing defined on $\cS$.
\item $D:\Gamma(N,\cS)\rightarrow \Omega^1(N,\cS)$ is a flat and
  $\omega$-compatible connection on $\cS$.
\end{enumerate}
The {\em rank} of a duality structure $\Delta=(\cS,D,\omega)$ is
defined to be the rank of the underlying vector bundle $\cS$.
\end{Definition}

\

\noindent Compatibility of $D$ with $\omega$ means that the following
condition holds for all $ X\in \cX(N)$ and all $ \xi_1,\xi_2\in
\Gamma(N,\cS)$:
\ben
\label{Domega}
\omega(D_X \xi_1,\xi_2)+\omega(\xi_1,D_X \xi_2)=X (\omega(\xi_1,\xi_2))~~.
\een
Notice that the rank of a duality structure is necessarily even. 

\begin{Definition}
Let $\Delta_i=(\cS_i,D_i,\omega_i)$ with $i=1,2$ be two duality
structures defined on $N$. A {\em morphism of duality structures} from
$\Delta_1$ to $\Delta_2$ is a based morphism of vector bundles $f\in
\Hom(\cS_1,\cS_2)$ which satisfies the following two conditions: 
\begin{enumerate}[1.]
\itemsep 0.0em
\item $f:(\cS_1,\omega_1)\rightarrow (\cS_2,\omega_2)$ is a morphism
  of symplectic vector bundles, i.e. we have:
\be
\omega_2(f(\xi),f(\xi'))=\omega_1(\xi,\xi')~~\forall \xi,\xi'\in \Gamma(N,S_1)~~.
\ee
\item $f:(\cS_1,D_1)\rightarrow (\cS_2,D_2)$ is a morphism of vector
  bundles with connection, i.e.  we have $D_2\circ
  f=(\id_{\Omega^1(N)}\otimes f)\circ D_1$, where
  $f:\Gamma(N,\cS_1)\rightarrow \Gamma(N,\cS_2)$ denotes the map
  induced on sections.
\end{enumerate}
\end{Definition}

\noindent With this notion of morphism, duality structures defined on
$N$ form a category which we denote by $\DS(N)$. Notice that any
morphism $f:(\cS_1,D_1,\omega_1)\rightarrow (\cS_2,D_2,\omega_2)$ of
duality structures is a monomorphism of vector bundles, i.e. we have
$\ker f=0$. For any $x\in N$, the fiber
$f_x:(\cS_{1,x},\omega_{1,x})\rightarrow (\cS_{2,x},\omega_{2,x})$ is
a morphism in the category $\Symp$ of finite-dimensional symplectic
vector spaces over $\R$ (see Appendix \ref{app:sympspaces}). In
particular, $f_x$ is injective for all $x\in N$.

\paragraph{The parallel transport functor.}

Let $\Delta=(\cS,D,\omega)$ be
a duality structure of rank $2n$ defined on $N$. Let $T^\Delta_\gamma$
be the parallel transport of $D$ along a path $\gamma\in\cP(N)$. The
compatibility condition \eqref{Domega} amounts to the requirement that
$\omega$ is preserved by the parallel transport of $D$ in the sense
that the following relation holds for any $\gamma\in \cP(N)$:
\be
\omega_{\gamma(1)}(T^\Delta_\gamma(s_1),T^\Delta_\gamma(s_2))=\omega_{\gamma(0)}(s_1,s_2)~~\forall s_1,s_2\in \cS_{\gamma(0)}~~.
\ee
Let $\Delta_i=(\cS_i,D_i,\omega_i)$ with $i=1,2$ be two duality
structures defined on $N$.  A morphism of symplectic vector bundles
$f:(\cS_1,\omega_1)\rightarrow (\cS_2,\omega_2)$ is a morphism of
duality structures from $(\cS_1,D_1,\omega_1)$ to
$(\cS_2,D_2,\omega_2)$ iff it satisfies the following condition:
\ben
\label{fT}
f_{\gamma(1)}\circ T^{\Delta_1}_\gamma=T^{\Delta_2}_\gamma\circ f_{\gamma(0)}~~.
\een
Recall that $\Pi_1(N)$ denotes the first homotopy groupoid of $N$. 

\begin{Definition}
The category of {\em symplectic local systems defined on $N$}
is the category $[\Pi_1(N),\Symp^\times]$ of functors from
$\Pi_1(N)$ to the groupoid $\Symp^\times$ and natural transformations
of such. 
\end{Definition} 

\noindent An object of $[\Pi_1(N),\Symp^\times]$ can be viewed as a
$\Symp^\times$-valued local system defined on $N$ as explained, for
example, in \cite{Whitehead}.

\begin{Definition}
The {\em parallel transport functor} of the duality structure $\Delta$
is the functor $T_{\Delta}:\Pi_1(N)\rightarrow\Symp^\times$ which
associates to any point $x\in N$ the symplectic vector space
$T_\Delta(x)=(\cS_x,\omega_x)$ and to any homotopy class $c\in
\Pi_1(N)$ with fixed initial point $x$ and fixed final point $y$ the
invertible symplectic morphism
$T_\Delta(c)=T_\gamma^\Delta:(\cS_x,\omega_x)\stackrel{\sim}{\rightarrow}
(\cS_y,\omega_y)$, where $\gamma\in \cP(N)$ is any path which
represents the class $c$.
\end{Definition}

\noindent Notice that $T_\Delta$ is well-defined since the connection
$D$ is flat and symplectic. Moreover, $T_\Delta$ determines the
duality structure $\Delta$. Relation \eqref{fT} implies that a
morphism of duality structures $f:\Delta_1\rightarrow \Delta_2$ gives
a natural transformation from $T_{\Delta_1}$ to $T_{\Delta_2}$. In
fact, the map which takes $\Delta$ into $T_\Delta$ and $f$ into the
corresponding natural transformation gives an equivalence between the
categories $\DS(N)$ and $[\Pi_1(N),\Symp^\times]$.

\paragraph{Classification of duality structures.}

Let $x$ be any point of $N$ and $\Delta=(\cS,D,\omega)$ be a duality
structure of rank $2n$ defined on $N$. The holonomy of $D$ along
piecewise-smooth loops based at $x$ induces a morphism of groups:
\ben
\label{holrep}
\hol_D^x:\pi_1(N,x)\rightarrow \Sp(\cS_x,\omega_x)~~,
\een
whose image:
\be
\label{hol}
\Hol_D^x\eqdef \hol_D^x(\pi_1(N,x))\subset \Sp(\cS_x,\omega_x)~~
\ee
is the holonomy group of $D$ at $x$. Noticing that
$\pi_1(N,x)=\End_{\Pi_1(N)}(x)$, we have:
\be
\hol_D^x=T_\Delta|_{\pi_1(N,x)}~~,~~\Hol_D^x=T_\Delta(\pi_1(N,x))~~.
\ee
In fact, $D$ is uniquely determined by its holonomy representation
$\hol_D^x$, while the symplectic gauge equivalence class of $D$ is
determined by the orbit of $\hol_D^x$ under the conjugation action of
$\Sp(\cS_x,\omega_x)$. Since $N$ is connected, the isomorphism type of
$\Hol_D^x$ does not depend on $x$ and we denote it by $\Hol_D$. Let
${\tilde N}$ be the universal covering space of $N$, viewed as a
principal $\pi_1(N,x)$-bundle over $N$, where $\pi_1(N,x)$ acts freely
and transitively on the fibers through deck transformations. Then
$\cS$ is isomorphic with the vector bundle associated to ${\tilde N}$
through the holonomy representation $\hol_D^x$:
\be
\cS\simeq {\tilde N}\times_{\hol_D^x}\cS_x~~.
\ee
Moreover, the symplectic pairing $\omega$ is induced by 
$\omega_x$ while the flat symplectic connection $D$ is induced by the
trivial flat connection on the trivial bundle ${\tilde N}\times
\cS_x$. It follows that isomorphism classes of
duality structures of rank $2n$ defined on $N$ are in
bijection with points of the symplectic character variety:
\be
C_{\pi_1(N)}(\Sp(2n,\R))\eqdef \Hom(\pi_1(N),\Sp(2n,\R))/\Sp(2n,\R)~~,
\ee
where $\Sp(2n,\R)$ acts by conjugation. 

\paragraph{Duality frames.}

\noindent Let $\Delta=(\cS,D,\omega)$ be a duality structure of rank $2n$ defined on $N$. 

\begin{Definition}
A {\em duality frame} of $\Delta$ defined on an open subset
$\cU\subset N$ is a local frame $\cE\eqdef (e_1,\ldots e_n,f_1,\ldots
f_n)$ of $\cS$ defined on $\cU$ which satisfies the following two
conditions:
\begin{enumerate}[(a)]
\itemsep 0.0em
\item $\cE$ is a symplectic frame of $(\cS,\omega)$, i.e.: 
\be
\omega(e_i,e_j)=\omega(f_i,f_j)=0~~,~~\omega(e_i,f_j)=\delta_{ij}~~\forall i,j=1\ldots n~~.
\ee
\item $\cE$ is a flat frame of $(\cS,D)$, i.e.: 
\be
D(e_i)=D(f_i)=0~~\forall i=1\ldots n~~.
\ee
\end{enumerate}
\end{Definition}

\begin{Remark} 
Let $\cU\subset N$ be an open subset and $x\in \cU$ be a point. Then
$\cU$ supports a duality frame of $\Delta$ iff
$\hol_D^x(j_\ast(\pi_1(\cU,x)))=1$, where
$j_\ast:\pi_1(\cU,x)\rightarrow \pi_1(N,x)$ is the morphism of groups
induced by the inclusion $\cU\subset N$ (this morphism need not be
injective !). In particular, any open subset $\cU$ of $N$ for which
$j_\ast(\pi_1(\cU,x))=0$ (for example, a simply-connected domain)
supports a duality frame.
\end{Remark}

\paragraph{Trivial duality structures.}

\begin{Definition}
A duality structure defined on $N$ is called {\em trivial} if it admits a
globally-defined duality frame, i.e.  a duality frame defined on all
of $N$.
\end{Definition}

\noindent A duality structure $\Delta=(\cS,D,\omega)$ is trivial iff
$D$ is a trivial flat connection, i.e. iff $\Hol_D$ is the trivial
group. Equivalently, $\Delta$ is trivial iff $T_\Delta$ is a trivial
symplectic local system, which means that $T_\Delta(\Hom_{\Pi_1(N)}(x,y))$ is a
singleton subset of $\Hom_{\Symp^\times}(T_\Delta(x),T_\Delta(y))$ for
any $x,y\in N$.  In this case, $D$ is determined up to a symplectic
gauge transformation and $\Delta$ is isomorphic with the duality
structure whose underlying vector bundle is the trivial bundle
$\R^{2n}_N=N\times \R^{2n}$, whose flat connection is the trivial
connection $\dd|_{\Omega^0(N)}\otimes \id_{\R^{2n}}$ and whose
symplectic pairing is the constant pairing induced by the canonical
symplectic pairing $\omega_0$ of $\R^{2n}$. Any duality structure
defined on a simply-connected manifold is necessarily trivial (in this
case, the symplectic character variety is reduced to a point). When
$N$ is not simply-connected, a duality structure $\Delta$ defined on
$N$ can be non-trivial even if the underlying symplectic vector bundle
$(\cS,\omega)$ is symplectically trivial; that situation is discussed
in Section \ref{sec:Strivial}. We refer the reader to Appendix
\ref{app:twistedconstructions} for certain twisted constructions
associated to duality structures.

\subsection{Electromagnetic structures}

Let $N$ be a manifold. This subsection assumes familiarity with the
notion of tamings (a.k.a. positive compatible complex structures)
\cite{Vaisman, Berndt, DuffSalamon} on symplectic vector spaces and
symplectic vector bundles, for which we refer the reader to
loc. cit. and to Appendices \ref{app:space_tamings} and
\ref{app:bundle_tamings}. 

\paragraph{The category of electromagnetic structures defined on $N$.}

Recall that a complex structure on a real vector bundle $\cS$ defined
over $N$ is an endomorphism $J\in \End_\R(\cS)$ such that
$J^2=-\id_\cS$. When $(\cS,\omega)$ is a symplectic vector bundle, a
complex structure $J$ on $\cS$ is called compatible with $\omega$ if
$\omega(J\xi_1,J\xi_2)=\omega(\xi_1,\xi_2)$ for all $\xi_1,\xi_2\in
\Gamma(N,\cS)$. A compatible complex structure $J$ induces a symmetric
bilinear pairing $Q=Q_{J,\omega}$ on $\cS$, defined through:
\ben
\label{Qdef}
Q(\xi_1,\xi_2)\eqdef \omega(J\xi_1,\xi_2)=-\omega(\xi_1,J\xi_2)~~
\een
and a Hermitian form $h=h_{J,\omega}\eqdef Q+\i \omega$ on the complex
vector bundle $(\cS,J)$.  A compatible complex structure is called a
{\em taming} of $(\cS,\omega)$ if $Q$ is positive-definite, which is
equivalent with the requirement that $h$ is a Hermitian scalar product
on the complex vector bundle $(\cS,J)$.

\begin{Remark}
The symplectic geometry literature \cite{Vaisman, Berndt, DuffSalamon}
generally uses the convention that a taming of $(\cS,\omega)$ is an
$\omega$-compatible complex structure $J$ for which the symmetric
pairing \eqref{Qdef} is {\em negative} definite. By contrast, the
convention used in the present paper is that common in the
Abelian variety literature \cite{BL}. The two conventions are related
by either of the substitutions $\omega\rightarrow -\omega$ or
$J\rightarrow -J$.
\end{Remark}

\begin{Definition}
An {\em electromagnetic structure} (or {\em tamed flat symplectic
  vector bundle}) defined on $N$ is a quadruplet $\Xi\eqdef
(\cS,D,J,\omega)$, where $\Xi_0\eqdef (\cS,D,\omega)$ is a duality
structure defined on $N$ (called the {\em duality structure underlying
  $\Xi$}) and $J$ is a taming of the symplectic vector bundle
$(\cS,\omega)$. The {\em rank} of $\Xi$ is defined to be the rank of
the underlying duality structure $\Xi_0$, i.e. the rank of the real
vector bundle $\cS$.
\end{Definition}

\noindent Notice that $J$ need not be covariantly-constant with
respect to $D$. Thus $D$ is generally {\em not} a complex-linear
connection on the Hermitian vector bundle $(\cS,J,h)$, but only a
connection which is symplectic relative to $\Im h$. The space
$\fJ_+(\cS,\omega)\subset \End_\R(\cS)$ of tamings of $(\cS,\omega)$
is a non-empty and contractible topological space \cite[Sec. 2.6]{DuffSalamon}.

\begin{Definition}
Let $\Xi_1=(\cS_1,D_1,J_1,\omega_1)$ and
$\Xi_2=(\cS_1,D_1,J_1,\omega_1)$ be two electromagnetic structures
defined on $N$. A {\em morphism of electromagnetic structures} from
$\Xi_1$ to $\Xi_2$ is a morphism of duality structures
$f:(\cS_1,D_1,\omega_1)\rightarrow (\cS_2,D_2,\omega_2)$ such that
$J_2\circ f =f\circ J_1$.
\end{Definition}

\noindent Thus a morphism of electromagnetic structures is a morphism
of the underlying duality structures which is also a complex-linear
morphism between the underlying complex vector spaces $(\cS_1,J_1)$
and $(\cS_2,J_2)$. Equivalently, it is a morphism between the
associated Hermitian vector spaces which intertwines the flat
connections $D_1$ and $D_2$. With this notion of morphism,
electromagnetic structures defined on $N$ form a category which we
denote by $\ES(N)$. This fibers over the category of duality
structures $\DS(N)$, the fiber at a duality structure
$\Delta=(\cS,D,\omega)$ being given by the set of those
electromagnetic structures $\Xi$ whose underlying duality structure
$\Xi_0$ equals $\Delta$. This fiber can be identified with the set of
tamings $\fJ_+(\cS,\omega)$. Accordingly, the set of isomorphism
classes of $\ES(N)$ fibers over the disjoint union of character
varieties $\sqcup_{n\geq 0}C_{\pi_1(N)}(\Sp(2n,\R))$.

\paragraph{The fundamental form of an electromagnetic structure.}

\noindent Let $\Xi=(\cS,D,J,\omega)$ be an electromagnetic structure
defined on $N$ and $h=Q+\i\omega$ be the Hermitian scalar product
defined by $\omega$ and $J$ on $\cS$. Let
$D^\ad:\Gamma(N,End(\cS))\rightarrow \Omega^1(N,End(\cS))$ be the
adjoint connection of $D$, i.e. the connection induced by $D$ on the
vector bundle $End(\cS)$. For any
$A\in \End(\cS)=\Gamma(N,End(\cS))$, we have:
\be
D^\ad(A)=D\circ A-(\id_{\Omega^1(N)}\otimes A)\circ D~~,
\ee
where $A$ is viewed as a map $A:\Gamma(N,\cS)\rightarrow \Gamma(N,\cS)$. 
Thus: 
\be
D^\ad_X(A)=D_X\circ A-A\circ D_X~~\forall X\in \cX(N)~~.
\ee

\begin{Definition}
The {\em fundamental form} of $\Xi$ is the $End(\cS)$-valued 1-form
defined on $N$ through the relation:
\ben
\label{ThetaDef}
\Theta_\Xi \eqdef D^\ad(J)\in \Omega^1(N,End(\cS))~~.
\een
\end{Definition}

\noindent We have $\Theta_\Xi(X)=D_X \circ J-J\circ D_X\in
\Gamma(N,End(\cS))$ for all $X\in \cX(N)$. Let $\Theta\eqdef
\Theta_\Xi$.

\begin{Proposition}
For all $X\in \cX(N)$ and all $\xi_1,\xi_2\in \Gamma(N,\cS)$, we have:
\beqan
J\circ \Theta(X) &=& -\Theta(X)\circ J\label{p1}\\
\omega(\Theta(X)\xi_1,\xi_2) &=& -\omega(\xi_1,\Theta(X)\xi_2) \label{p2}\\
Q(\Theta(X)\xi_1,\xi_2) &=& Q(\xi_1,\Theta(X)\xi_2) \label{p3}~~.
\eeqan
In particular, $\Theta(X)$ is an antilinear self-adjoint endomorphism
of the Hermitian vector bundle $(\cS,J,h)$:
\ben
\label{p4}
h(\Theta(X)\xi_1,\xi_2)=\overline{h(\xi_1,\Theta(X)\xi_2)}~~.
\een
\end{Proposition}

\begin{proof}
The first relation follows by applying $D_X^{\ad}$ to the identity
$J^2=-\id_{\cS}$, using the fact that $D_X^{\ad}$ is a derivation
of the composition of $\End(\cS)$. For the second relation, we compute:
\beqa
&& \omega(\Theta(X)\xi_1,\xi_2)=\omega(D_X(J\xi_1),\xi_2)-\omega(J (D_X\xi_1),\xi_2)=\omega(D_X(J\xi_1),\xi_2)+\omega(D_X \xi_1,J\xi_2)=\nn\\
&& X[\omega(J\xi_1,\xi_2)+\omega(\xi_1,J\xi_2)]-\omega(J\xi_1,D_X\xi_2)-\omega(\xi_1,D_X(J\xi_2))=\omega(\xi_1,J(D_X\xi_2))-\omega(\xi_1,D_X(J\xi_2))\nn\\
&& =-\omega(\xi_1,\Theta(X)\xi_2)~~,
\eeqa
where we used $J$-invariance of $\omega$ as well as relation
\eqref{Domega}. For the third relation, we compute:
\be
Q(\Theta(X)\xi_1,\xi_2)=\omega(J\Theta(X)\xi_1,\xi_2)=-\omega(\Theta(X)J\xi_1,\xi_2)=\omega(J\xi_1,\Theta(X)\xi_2)=Q(\xi_1,\Theta(X)\xi_2)~~,
\ee
where in the second equality we used \eqref{p1} and in the third
equality we used \eqref{p2}. Property \eqref{p1} means that
$\Theta(X)$ is $J$-antilinear, while properties \eqref{p2} and
\eqref{p3} are equivalent with \eqref{p4}, which means that we have:
\ben
\label{p5}
\Theta(X)^\dagger=\Theta(X)~~\forall X\in \cX(N)~~,
\een
where $\Theta(X)^\dagger$ is the antilinear adjoint of $\Theta(X)$
with respect to $h$. Hence $\Theta(X)$ is an antilinear self-adjoint
endomorphism of the Hermitian vector bundle $(\cS,J,h)$.
\end{proof}

\begin{Definition} 
Let $\Theta_c\in \Omega^1(N,Hom(\cS\otimes \cS,\C))$ be defined through:
\ben
\label{hTd}
\Theta_c(X)(\xi_1,\xi_2)\eqdef h(\xi_1,\Theta(X)\xi_2) \in \cC^\infty(N,\C)~~,
\een
for all $X\in \cX(N)$ and all $\xi_1,\xi_2\in \Gamma(N,\cS)$. 
\end{Definition}

\noindent Let $\cS^\vee=Hom_\C(\cS,\C)$ be the complex-linear dual of
the complex vector bundle $(\cS,J)$. Let $\Sym^2_\C(\cS^\vee)$ denote the
second symmetric power (over $\C$) of the complex vector bundle $\cS^\vee$.

\begin{Proposition}
We have $\Theta_c\in \Omega^1(N,\Sym^2_\C(\cS^\vee))$. Thus, for every
$X\in \cX(N)$, the $\C$-valued pairing $\Theta_c(X)\in \Hom(\cS\otimes
\cS,\C)$ is symmetric and complex-bilinear with respect to the complex
structure $J$.
\end{Proposition}

\begin{proof}
Follows from \eqref{hTd} and \eqref{p4} using the fact that $h$ is
$J$-Hermitian while $\Theta(X)$ is $J$-antilinear.
\end{proof}

\noindent Since $h$ is non-degenerate, the bundle-valued one-form
$\Theta_c\in \Omega^1(N,\Sym^2_\C(\cS^\vee))$ encodes the same
information as the fundamental form $\Theta$.

\paragraph{Unitary electromagnetic structures.}

\begin{Definition}
An electromagnetic structure $\Xi=(\cS,D,J,\omega)$ is called
{\em unitary} if $\Theta_\Xi=0$, i.e. if $J$ is covariantly constant
with respect to $D$:
\be
D_X \circ J=J\circ D_X~~\forall X\in \cX(N)~~
\ee
\end{Definition}

\

\noindent By definition, the {\em category $\UES(N)$ of unitary
  electromagnetic structures} on $N$ is the full sub-category of
$\ES(N)$ whose objects are the unitary electromagnetic structures.
When an electromagnetic structure $\Xi$ is unitary, the connection $D$
preserves both $J$ and $\omega$, so $(\cS,D,J,h_{J,\omega})$ is a flat
Hermitian vector bundle, i.e. a Hermitian vector bundle endowed with a
flat complex-linear connection which is compatible with the Hermitian
pairing. This gives an equivalence of categories between $\UES(N)$ and
the category of flat Hermitian vector bundles defined on $N$.  A
unitary electromagnetic structure defines a parallel transport functor
valued in the unit groupoid of the category $\Herm$ of
finite-dimensional Hermitian vector spaces. This gives an equivalence
of categories between $\UES(N)$ and the functor category
$[\pi_1(N),\Herm^\times]$.  Thus isomorphism classes of unitary
electromagnetic structures of rank $2n$ are in bijection with
isomorphism classes of flat Hermitian vector bundles and hence
correspond to points of the character variety:
\be
C_{\pi_1(N)}(\U(n))\eqdef \Hom(\pi_1(N),\U(n))/\U(n)~~,
\ee
where $\U(n)$ acts by conjugation. In particular, unitary
electromagnetic structures over $N$ admit a finite-dimensional
classification (unlike arbitrary electromagnetic structures). 

\subsection{Scalar-duality and scalar-electromagnetic structures}

\begin{Definition} 
A {\em scalar-duality structure} is an ordered system $(\Sigma,\Delta)$, where:
\begin{enumerate}[1.]
\itemsep 0.0em
\item $\Sigma=(\cM,\cG,\Phi)$ is a scalar structure
\item $\Delta=(\cS,D,\omega)$ is an electromagnetic structure defined
  on $\cM$.
\end{enumerate}
\end{Definition}

\begin{Definition} 
A {\em scalar-electromagnetic structure} is an ordered system
$\cD=(\Sigma,\Xi)$, where:
\begin{enumerate}[1.]
\itemsep 0.0em
\item $\Sigma=(\cM,\cG,\Phi)$ is a scalar structure
\item $\Xi=(\cS,D,J,\omega)$ is an electromagnetic structure defined
  on $\cM$.
\end{enumerate}
In this case, $\cD_0\eqdef (\Sigma,\Xi_0)$ (where
$\Xi_0=(\cS,D,\omega)$ is the underlying duality structure of $\Xi$)
is called the {\em underlying scalar-duality structure of $\cD$}.
\end{Definition}

\noindent Let $\cD$ be a scalar-electromagnetic structure as in the
definition and $\sharp_\cG:T^\ast\cM\rightarrow T \cM$ be the musical
isomorphism defined by $\cG$.

\begin{Definition}
An $End(\cS)$-valued vector field defined on $\cM$ is a smooth global
section of the vector bundle $T\cM\otimes End(\cS)$.
\end{Definition}

\noindent Let $\cX(M,End(\cS))\eqdef \Gamma(\cM,T\cM\otimes End(\cS))$
be the $\cC^\infty(\cM,\R)$-module of $End(\cS)$-valued vector fields
defined on $\cM$.

\begin{Definition}
The {\em fundamental field} of the scalar-electromagnetic structure
$\cD=(\Sigma,\Xi)$ is the $End(\cS)$-valued vector field:
\ben
\label{PsiDef}
\Psi_\cD \eqdef (\sharp_\cG\otimes \id_{End(\cS)})(\Theta_\Xi) \in \cX(\cM,End(\cS))~~,
\een
where $\Theta_\Xi$ is the fundamental form of the electromagnetic
structure $\Xi$.
\end{Definition}

\subsection{Positively-polarized $\cS^\varphi$-valued 2-forms}

Let $(M,g)$ be a Lorentzian four-manifold. Let $\cD=(\Sigma,\Xi)$ be a
scalar-electromagnetic structure with underlying scalar structure
$\Sigma=(\cM,\cG,\Phi)$ and underlying electromagnetic structure
$\Xi=(\cS,D,J, \omega)$. Let $\varphi:M\rightarrow\cM$ be a smooth map. 

\begin{Definition}
The $\varphi$-pullback of the electromagnetic structure $\Xi$ defined
on $\cM$ is the electromagnetic structure
$\Xi^\varphi\eqdef(\cS^\varphi, D^\varphi,J^\varphi, \omega^\varphi)$
defined on $M$, where $D^\varphi$ denotes the pull-back of the
connection $D$.
\end{Definition}

\noindent The Hodge operator $\ast_g:\wedge T^\ast M\rightarrow \wedge
T^\ast M$ of $(M,g)$ induces the endomorphism $\bast_g\eqdef
\ast_g\otimes \id_{\cS^{\varphi}}$ of the {\em twisted exterior bundle} 
$\wedge_M(\cS^\varphi)\eqdef \wedge T^\ast M\otimes \cS^{\varphi}$. For ease of notation we will sometimes denote $\bast_g$ simply by $\ast_g$.

\begin{Definition}
The {\em twisted Hodge operator} of $\Xi^\varphi$ is the bundle
endomorphism of $\wedge_M(\cS^\varphi)$ defined through:
\be
\star_{g,J^\varphi}\eqdef\ast_g\otimes J^\varphi=\ast_g\circ J^\varphi=J^\varphi\circ \ast_g~~.
\ee
\end{Definition}

\

\noindent Let: 
\ben
\label{alphadef}
\alpha \eqdef \oplus_{k=0}^4 (-1)^k\id_{\wedge^k T^\ast
  M}
\een 
be the {\em main automorphism} of the exterior bundle $\wedge T^\ast M$. We have:
\ben
\label{HodgeSquare}
\star_{g,J^\varphi}^2=\alpha\otimes \id_{\cS^\varphi}~~.
\een
The operator $\star_{g,J^\varphi}$ preserves the sub-bundle
$\wedge^2_N(\cS^\varphi)\eqdef \wedge^2T^\ast M\otimes \cS^\varphi$ of
the twisted exterior bundle, on which it squares to plus the
identity. Accordingly, we have a direct sum decomposition:
\be
\wedge^2 T^\ast M\otimes \cS^\varphi=(\wedge^2 T^\ast M\otimes \cS^\varphi)^+\oplus (\wedge^2 T^\ast M\otimes \cS^\varphi)^-~~,
\ee 
where $(\wedge^2 T^\ast M\otimes \cS^\varphi)^\pm$ are the sub-bundles
of eigenvectors of $\star_{g,J^\varphi}$ corresponding to the
eigenvalues $\pm 1$.

\begin{Definition}
An $\cS^\varphi$-valued 2-form $\eta\in \Omega^2(M,\cS^\varphi)$
defined on $M$ is called {\em positively polarized} with respect to
$g$ and $J^\varphi$ if it is a section of the vector bundle $(\wedge^2
T^\ast M\otimes \cS^\varphi)^+$, which amounts to the requirement that
it satisfies the {\em positive polarization condition}:
\ben
\label{2formpol}
\star_{g,J^\varphi} \eta=\eta~~\mathrm{i.e.}~~\ast_g \eta=-J^\varphi\eta~~.
\een
\end{Definition}

\

\noindent For any open subset $U$ of $M$, let $g_U\eqdef g|_U$,
$\varphi_U\eqdef \varphi|_U$ and:
\be
\iOmega^{\Xi, g, \varphi} (U)\eqdef \Gamma(U,(\wedge^2 T^\ast M\otimes \cS^\varphi)^+)=\{\eta\in \Omega^2(U,\cS^\varphi)|\star_{g,J^\varphi} \eta=\eta\}
\ee
denote the $\cC^\infty(U,\R)$-submodule of $\Omega^2(U,\cS^\varphi)$
consisting of those 2-forms defined on $U$ which are positively
polarized with respect to $g|_U$ and $J^\varphi|_U$. This defines a sheaf of
$\cC^\infty$-modules $\iOmega^{\Xi, g, \varphi}$ on $M$, namely the
sheaf of smooth sections of the bundle $(\wedge^2 T^\ast M\otimes
\cS^\varphi)^+$. The globally-defined and positively-polarized
$\cS^\varphi$-valued forms are the global sections of this
sheaf. Notice that $\eta\in \Omega^2(M,\cS^\varphi)$ is positively
polarized iff $\star_{g,J^\varphi} \eta$ is (because
$\star_{g,J^\varphi}$ squares to the identity on
$\Omega^2(M,\cS^\varphi)$).

\subsection{Generalized four-dimensional Einstein-Scalar-Maxwell theories}
\label{sec:completetheoryI}

Let $\cD=(\Sigma,\Xi)$ be a scalar-electromagnetic structure with
underlying scalar structure $\Sigma=(\cM,\cG,\Phi)$ and underlying
electromagnetic structure $\Xi=(\cS,D,J, \omega)$. Let $Q$ denote the
Euclidean scalar product induced by $\omega$ and $J$ on $\cS$ (see
\eqref{Qdef}). Let $M$ be a four-manifold.

\paragraph{Preparations.}

Let $\varphi:M\rightarrow \cM$ be a smooth map and $\Xi^\varphi\eqdef
(\cS^\varphi, D^\varphi,J^\varphi,\omega^\varphi)$ be the pulled-back
electromagnetic structure defined on $M$. Before describing our
models, we introduce certain operations which enter the geometric
formulation of the equations of motion. Let $\wedge_M\eqdef \wedge
T^\ast M$ be the exterior bundle of $M$ and
$\wedge_M(\cS^\varphi)\eqdef \wedge_M \otimes \cS^\varphi$ be the
twisted exterior bundle. Let $\alpha$ be the main automorphism of $\wedge_M$ 
(see \eqref{alphadef}). 

\begin{Definition} 
The {\em twisted wedge product} on $\wedge_M(\cS^\varphi)$ is the
$\wedge_M$-valued fiberwise bilinear pairing
$\bwedge_\omega:\wedge_M(\cS^\varphi)\times_M\wedge_M(\cS^\varphi)\rightarrow
\wedge_M$ which is determined uniquely by the condition (we write
$\bwedge_\omega$ in infix notation):
\ben
\label{bwedge}
(\rho_1\otimes \xi_1) \bwedge_{\omega} (\rho_2\otimes \xi_2)=\alpha(\rho_2)\omega(\xi_1,\xi_2) \rho_1\wedge \rho_2
\een
for all $\rho_1,\rho_2\in \Omega(M)$ and all $\xi_1,\xi_2\in
\Gamma(M,\cS^\varphi)$.
\end{Definition}

\noindent Notice that the twisted wedge product depends only on the
pulled-back symplectic vector bundle $(\cS^\varphi,\omega^\varphi)$.
More information about this operation can be found in Appendix
\ref{app:twistedwedge}. Given any vector bundle $\cT$ on $M$, we
trivially extend the twisted wedge product to a $\cT \otimes
\wedge_M$-valued pairing (denoted by the same symbol) between the
bundles $\cT \otimes \wedge_M(\cS^\varphi)$ and
$\wedge_M(\cS^\varphi)$. Thus:
\be
(t\otimes \eta_1)\bwedge_\omega \eta_2\eqdef  t\otimes (\eta_1\bwedge_\omega \eta_2)~~,~~\forall t\in \Gamma(M,\cT)~~\forall \eta_1,\eta_2\in \wedge_M(\cS^\varphi)~~.
\ee

\noindent Let $g\in \Met_{3,1}(M)$ be a Lorentzian metric on $M$. The {\em
  exterior pairing} $(~,~)_g$ is the pseudo-Euclidean metric induced
by $g$ on the exterior bundle $\wedge_M$.

\begin{Definition}
The {\em twisted exterior pairing} $(~,~):=(~,~)_{g,Q^\varphi}$ is the
unique pseudo-Euclidean scalar product on the twisted exterior bundle
$\wedge_M(\cS^\varphi)$ which satisfies:
\be
(\rho_1\otimes \xi_1,\rho_2\otimes \xi_2)_{g,Q^\varphi}=(\rho_1,\rho_2)_g Q^\varphi(\xi_1,\xi_2)~~
\ee
for any $\rho_1,\rho_2\in \Omega(M)$ and any $\xi_1,\xi_2\in
\Gamma(M,\cS^\varphi)$.
\end{Definition}

\noindent For any vector bundle $\cT$ defined on $M$, we trivially
extend the twisted exterior pairing to a $\cT$-valued pairing (denoted
by the same symbol) between the bundles $\cT \otimes
\wedge_M(\cS^\varphi)$ and $\wedge_M(\cS^\varphi)$. Thus:
\be
(t\otimes \eta_1,\eta_2)_{g,Q^\varphi} \eqdef  t\otimes (\eta_1,\eta_2)_{g,Q^\varphi}~~\forall t\in \Gamma(M,\cT)~~\forall \eta_1,\eta_2\in \wedge_M(\cS^\varphi)~~.
\ee

\noindent The {\em inner $g$-contraction of 2-tensors} is the bundle
morphism $\oslash_g:(\otimes^2T^\ast M)^{\otimes 2}\rightarrow
\otimes^2 T^\ast M$ which is uniquely determined by the condition:
\be
(\omega_1\otimes\omega_2)\oslash_g (\omega_3\otimes \omega_4)=(\omega_2,\omega_3)_g\omega_1\otimes \omega_4~~\forall \omega_1,\omega_2,\omega_3,\omega_4\in \Omega^1(M)~~.
\ee
Viewing $\wedge^2 T^\ast M$ as the sub-bundle of antisymmetric
2-tensors inside $\otimes^2 T^\ast M$, the morphism $\oslash_g$
restricts to a morphism $\wedge^2 T^\ast M \otimes \wedge^2 T^\ast
M\stackrel{\oslash_g}{\rightarrow} \otimes^2 T^\ast M$, which is called the 
{\em inner $g$-contraction of 2-forms}. 

\begin{Definition}
\label{def:twistedinner}
The {\em twisted inner contraction} of $\cS^\varphi$-valued 2-forms
is the unique morphism of vector bundles
$\loslash:=\loslash_{g,J,\omega,\varphi}:\wedge_M^2(\cS^\varphi)\times_M\wedge_M^2(\cS^\varphi)\rightarrow
\otimes^2(T^\ast M)$ which satisfies:
\be
(\rho_1\otimes \xi_1)\loslash_{\Xi,g,\varphi} (\rho_2\otimes \xi_2)= Q^\varphi(\xi_1,\xi_2)\rho_1\oslash_g\rho_2
\ee
for all $\rho_1,\rho_2\in \Omega^2(M)$ and all $\xi_2,\xi_2\in
\Gamma(M,\cS^\varphi)$. 
\end{Definition}

\begin{Remark}
For any $\rho\in \Omega^2(M)$, we have $\rho\oslash_g\rho\in
\Gamma(M,\Sym^2(T^\ast M))$. For any $\eta\in
\Omega^2(M,\cS^\varphi)$, we have $\eta\loslash \eta\in
\Gamma(M,\Sym^2(T^\ast M))$.
\end{Remark}

\noindent Let $\Psi\eqdef \Psi_\cD\in \Gamma(\cM,T\cM\otimes End(\cS))$ be the
fundamental field of $\cD$ and $\Psi^\varphi\in
\Gamma(M,(T\cM)^\varphi\otimes End(\cS^\varphi))$ be its pullback
through $\varphi$. 

\begin{Definition} 
Let $\eta\in \Omega^2(M,\cS^\varphi)$ be a $\cS^\varphi$-valued 2-form. 
The {\em $\eta$-average} of $\Psi^\varphi$ is the following section of $(T\cM)^\varphi$: 
\be
(\Psi^\varphi \eta,\eta)_{g,Q^\varphi}\in \Gamma(M,(T\cM)^\varphi)~~. 
\ee
\end{Definition}

\noindent Let $\star\eqdef \star_{g,J,\varphi}$. Using relation \eqref{starwedgegen}, we find:
\be
(\Psi^\varphi\eta,\eta)\nu=(\Psi^\varphi \eta)\bwedge_\omega\star \eta~~\forall \eta\in \Omega^2(M,\cS^\varphi)~~.
\ee
When $\eta$ is positively polarized, we have $\star\eta=\eta$ and hence:  
\be
(\Psi^\varphi\eta,\eta)\nu=(\Psi^\varphi \eta)\bwedge_\omega\eta~~\forall \eta\in \iOmega^{\Xi,g,\varphi}(M)~~.
\ee

\paragraph{General Einstein-Scalar-Maxwell models.}

\begin{Definition}
The {\em sheaf of Einstein-Scalar-Maxwell (ESM) configurations}
$\Conf_\cD$ determined by $\cD$ on $M$ is the sheaf of sets defined through:
\be
\Conf_\cD(U)\eqdef \{(g,\varphi,\cV)|g\in \Met_{3,1}(U),\varphi\in \cC^\infty(U,\cM),\cV\in \iOmega^{\Xi,g,\varphi}(U)\}~~
\ee
for all open subsets $U$ of $M$, with the obvious restriction maps. An
element $(g,\varphi,\cV)\in \Conf_\cD(U)$ is called a {\em local ESM
  configuration} of type $\cD$ defined on $U$. The {\em set of global
  configurations} determined by $\cD$ on $M$ is the set:
\be
\Conf_\cD(M)\eqdef \{(g,\varphi,\cV)|g\in \Met_{3,1}(M),\varphi\in \cC^\infty(M,\cM),\cV\in \iOmega^{\Xi,g,\varphi}(M)\}~~.
\ee
of global sections of this sheaf. An element $(g,\varphi,\cV)\in
\Conf_\cD(M)$ is called a {\em global ESM configuration} of type
$\cD$ defined on $M$.
\end{Definition}

\begin{Definition}
Let $U$ be an open subset of $M$. The {\em Einstein-Scalar-Maxwell
  (ESM) equations} determined by $\cD$ on $U$ are the following equations
for local configurations $(g,\varphi,\cV)\in \Conf_\cD(U)$:
\begin{enumerate}[1.]
\item The Einstein equations:
\ben
\label{eins}
\G(g) =\kappa\, \mathrm{T}_\cD(g,\varphi,\cV)~~,
\een
with energy-momentum tensor $\mathrm{T}_\cD$ given by:
\begin{equation}
\rT_\cD(g,\varphi,\cV) \eqdef  T_\Sigma(g,\varphi)  + 2~\cV\loslash \cV ~~.
\end{equation}
\item The scalar equations:
\ben
\label{sc}
\theta_\Sigma(g,\varphi) - \frac{1}{2} (\bast \cV , \Psi^\varphi\cV) = 0~~.
\een
\item The electromagnetic equations:
\ben
\label{em}
\dd_{D^{\varphi}}\cV = 0~~.
\een
\end{enumerate}
A {\em local ESM solution} of type $\cD$ defined on $U$ is a smooth
solution $(g, \varphi,\cV)$ of these equations which is defined on $U$. A {\em global ESM
  solution} of type $\cD$ is a smooth solution of these equations
which is defined on $M$. The {\em sheaf of local ESM solutions}
$\Sol_\cD$ determined by $\cD$ is the sheaf of sets defined on $M$ whose
set of sections on an open subset $U\subset M$ is the set of all local
solutions of type $\cD$ defined on $U$.
\end{Definition}

\

\noindent Notice that $\Sol_\cD$ is a subsheaf of $\Conf_\cD$. In
Appendix \ref{app:localeom}, we show that the equations listed above
reproduce {\em locally} the ordinary equations of motion used the
supergravity literature (see, for example \cite{Andrianopoli,Ortin})
upon choosing a local duality frame of $(\cS,D,\omega)$. We stress
that such a frame need not exist globally, since the flat connection
$D$ may have non-trivial holonomy when the scalar manifold $\cM$ is
not simply-connected. We refer the reader to Section
\ref{sec:relation} for comparison with the usual model discussed in
the supergravity literature --- a model which assumes a trivial
duality structure. One motivation to consider the models defined above
is that they generalize the bosonic sector of all four-dimensional
ungauged supergravity theories and hence they afford a generalization of 
such theories.

\subsection{The sheaves of scalar-electromagnetic configurations and
  solutions relative to a space-time metric}

Let $\cD=(\Sigma,\Xi)$ be a scalar-electromagnetic structure, where
$\Sigma=(\cM,\cG,\Phi)$ and $\Xi=(\cS,D,J,\omega)$. Let $g\in
\Met_{3,1}(M)$ be a Lorentzian metric on $M$.

\begin{Definition}
The {\em sheaf of local scalar-electromagnetic configurations}
$\Conf^g_\cD$ relative to $g$ is the sheaf of sets defined on $M$
whose set of sections above any open subset $U$ of $M$ is defined through: 
\be
\Conf^g_\cD(U) \eqdef \{(\varphi,\cV)|\varphi\in \cC^\infty(U,\cM),\cV\in \iOmega^{\Xi,g,\varphi}(U)\}~~.
\ee
The {\em set of global scalar-electromagnetic configurations} relative
to $g$ is the set $\Conf^g_\cD(M)$ of global sections of this sheaf.
\end{Definition}

\begin{Definition}
The {\em sheaf of local scalar-electromagnetic solutions} relative to
$g$ is the sheaf of sets defined on $M$ whose set of sections
$\Sol^g_\cD(U)$ on an open subset $U$ of $M$ is defined as the set of
all pairs $(\varphi,\cV)\in \Conf_\cD^g(U)$ such that $\varphi$ and
$\cV$ satisfy equations \eqref{sc} and \eqref{em} on $U$. The {\em set
  of global scalar-electromagnetic solutions} relative to $g$ is the
set $\Sol^g_\cD(M)$ of global sections of $\Sol^g_\cD$.
\end{Definition}

\subsection{Electromagnetic field strengths}

Let $\cD=(\Sigma,\Xi)$ be a scalar-electromagnetic structure with
underlying scalar structure $\Sigma=(\cM,\cG,\Phi)$ and underlying
electromagnetic structure $\Xi=(\cS,D,J, \omega)$.

\begin{Definition}
An {\em electromagnetic field strength} on $M$ with respect to $\cD$
and relative to $g\in \Met_{3,1}(M)$ and to a map $\varphi\in
\cC^\infty(M,\cM)$ is an $\cS^\varphi$-valued 2-form $\cV\in
\Omega^2(M,\cS^\varphi)$ which satisfies the following two conditions:
\begin{enumerate}[1.]
\itemsep 0.0em
\item $\cV$ is positively polarized with respect to $g$ and $J^\varphi$,
  i.e. we have $\star_{g,J^\varphi}\cV=\cV$.
\item $\cV$ is $\dd_{D^\varphi}$-closed, i.e.:
\ben
\label{GaugeEOM0}
\dd_{D^\varphi} \cV=0~~.
\een
\end{enumerate}
The second condition is called {\em the electromagnetic equation}. 
\end{Definition}

\

\noindent For any open subset $U$ of $M$, let:
\ben
\iOmega^{\Xi,g,\varphi}_\cl(U)\eqdef \{\cV\in \iOmega^{\Xi,g,\varphi}(U) | \dd_{D^\varphi} \cV=0\}
\een
denote the set of electromagnetic field strengths defined on $U$,
which is an (infinite-dimensional) subspace of the $\R$-vector space
$\iOmega^{\Xi,g,\varphi}(U)$. Together with the obvious restriction
maps, this defines the {\em sheaf of electromagnetic field strengths}
$\iOmega^{\Xi,g,\varphi}_\cl$ relative to $\varphi$ and $g$, which is a
locally-constant sheaf of $\R$-vector spaces
defined on $M$. Notice that $\cV$ is a globally-defined electromagnetic
field strength iff $\star_{g,J^\varphi} \cV$ is. The positive
polarization condition implies that an electromagnetic field strength
$\cV$ is both $\dd_{D^\varphi}$ and $\updelta_{D^\varphi}$-closed:
\be
\dd_{D^\varphi}\cV=\updelta_{D^\varphi}\cV=0
\ee
and hence it satisfies the twisted d'Alembert equation
$\Box_{D^\varphi}\cV=0$, where the operators $\updelta_{D^\varphi}$
and $\Box_{D^\varphi}$ are defined in Appendix
\ref{app:twistedconstructions}.

\section{Scalar-electromagnetic dualities and symmetries}
\label{sec:duality}

In this section, we describe the groups of scalar-electromagnetic
duality transformations and scalar-electromagnetic symmetries 
for ESM models with non-trivial duality structure. As we shall see
below, the definitions of these groups and of their actions on the
sets of configurations and solutions are rather subtle in the
general case. In particular, the duality group constructed below can
differ markedly from that commonly found in the supergravity
literature (which considers only trivial duality structures). This
section uses certain constructions involving unbased automorphisms
of vector bundles, for which we refer the reader to Appendix
\ref{app:unbased}.

\subsection{Covariantly-constant unbased automorphisms of a vector bundle with connection}

Let $\cM$ be a manifold and $\cS$ be a real vector bundle over $\cM$.
Let $\Aut^\ub(\cS)$ denote the group of unbased automorphisms of
$\cS$. Let $\cX(M,\cS)\eqdef \cX(M)\otimes_{\cC^\infty(\cM,\R)}
\Gamma(\cM,\cS)$ and $\Omega(\cM,\cS)\eqdef
\Omega(M)\otimes_{\cC^\infty(\cM,\R)} \Gamma(M,\cS)$ denote the
$\cC^\infty(\cM,\R)$-modules of $\cS$-valued vector fields and
$\cS$-valued differential forms defined on $\cM$. Any unbased
automorphism $f\in \Aut^\ub(\cS)$ covering a diffeomorphism $f_0\in
\Diff(\cM)$ determines an $\R$-linear map
$\mf:\Gamma(\cM,\cS)\rightarrow \Gamma(\cM,\cS)$ as well as
$\R$-linear push-forward and pull-back operations
$f_\ast:\cX(\cM,\cS)\rightarrow \cX(\cM,\cS)$, and
$f^\ast:\Omega(\cM,\cS)\rightarrow \Omega(\cM,\cS)$ and an $\R$-linear
map $\bAd(f):\Gamma(\cM, End(\cS))\rightarrow \Gamma(\cM, End(\cS))$,
whose precise definitions and properties are given in Appendix
\ref{app:unbased}. We have $f^\ast(\xi)=\mf^{-1}(\xi)$ for all $\xi\in
\Omega^0(\cM,\cS)=\Gamma(\cM,\cS)$.  On the other hand, $f$ induces a based
isomorphism ${\hat f}:\cS\rightarrow \cS^{f_0}$, where $\cS^{f_0}$ is
the $f_0$-pullback of $\cS$. Let $D:\Gamma(\cM,\cS)\rightarrow
\Omega^1(\cM,\cS)$ be a connection on $\cS$.

\begin{Definition}
An unbased automorphism $f\in \Aut^\ub(\cS)$ is called {\em
  covariantly constant} with respect to $D$ if:
\ben
\label{Df}
f^\ast\circ D=D\circ \mf^{-1}~~.
\een
\end{Definition}

\begin{Proposition}
The following statements are equivalent for $f\in \Aut^\ub(\cS)$:
\begin{enumerate}[(a)]
\itemsep 0.0em
\item $f$ is covariantly constant with respect to $D$.
\item We have: 
\ben
\label{DfX}
D_{(f_0)_\ast(X)} ~\mf(\xi)=\mf(D_X(\xi))~~,~~\forall \xi\in \Gamma(\cM,\cS)~~\forall X\in \cX(\cM)~~.
\een
\item We have: 
\ben
\label{Dhatf}
D^{f_0}\circ {\hat f}={\hat f}\circ D~~,
\een
i.e. ${\hat f}$ is an isomorphism of bundles with connection from
$(\cS,D)$ to $(\cS^{f_0}, D^{f_0})$, where
$D^{f_0}:\Gamma(\cM,\cS^{f_0})\rightarrow \Omega^1(\cM,\cS^{f_0})$ is
the pulled-back connection.
\end{enumerate}
\end{Proposition}

\begin{proof}
Relation \eqref{pullpush} implies that \eqref{Df} is
equivalent with the condition
$\mf^{-1}(D_{(f_0)_\ast(X)}\xi)=D_X(\mf^{-1}(\xi))$ for all $X\in
\cX(\cM)$ and all $\xi\in \Gamma(\cM,\cS)$.  Multiplying both sides
with $\mf$ and replacing $\xi$ with $\mf(\xi)$ shows that this
relation is equivalent with \eqref{DfX}. For any $\psi\in \Diff(\cM)$,
the pull-back connection $D^\psi$ on $\cS^\psi$ is characterized by
the condition:
\be
D_X^\psi(\xi^\psi)=(D_{(\dd\psi)(X)}\xi)^\psi~~\forall \xi\in \Gamma(\cM,\cS)~~.
\ee
For any section $\xi\in \Gamma(\cM,\cS)$, we have $\xi^\psi=\Phi_\cS(\psi)^{-1}(\xi_\psi)$. 
Substituting this into the relation above shows that $D^\psi$ is given by: 
\ben
\label{Dpsi}
D^\psi_X=\Phi_\cS(\psi)^{-1}\circ D_{\psi_\ast(X)}\circ \Phi_\cS(\psi)~~\forall X\in \cX(\cM)~~,
\een
where we used the fact that $\psi_\ast(X)_\psi=(\dd\psi)(X)$.  Using
this relation and equation \eqref{hatf} shows that \eqref{Dhatf} is
equivalent with \eqref{DfX}.
\end{proof}

\noindent Covariantly constant automorphisms form a subgroup
$\Aut^\ub(\cS,D)$ of $\Aut^\ub(\cS)$. An unbased automorphism $f\in
\Aut^\ub(\cS)$ is covariantly constant iff it is compatible with the
parallel transport $T$ of $D$ in the sense that the following relation
holds for all $\gamma\in \cP(\cM)$:
\be
f_{\gamma(1)}\circ T_\gamma=T_{f_0\circ \gamma}\circ f_{\gamma(0)}~~.
\ee
Let $\Aut(\cS,D)$ denote the group of {\em based} automorphisms of
$\cS$ which are covariantly constant with respect to $D$. We have a
short exact sequence:
\be
1\rightarrow \Aut(\cS,D)\hookrightarrow \Aut^\ub(\cS,D)\longrightarrow \Diff^{\cS,D}(\cM)\rightarrow 1~~,
\ee
where $\Diff^{\cS,D}(\cM)$ is the subgroup of $\Diff^\cS(\cM)$
consisting of those diffeomorphisms of $\cM$ which admit a lift to a
covariantly constant unbased automorphism of $(\cS,D)$. When $D$ is
flat, the group $\Aut(\cS,D)$ of based covariantly constant
automorphisms is isomorphic with the commutant inside $\Aut(\cS_p)$ of
the holonomy group $\Hol_D^p$ of $D$ at any point $p\in \cM$.

\begin{Proposition} 
Let $D$ be a flat connection on $\cS$ and
$\dd_D:\Omega(\cM,\cS)\rightarrow \Omega(\cM,\cS)$ be the de Rham
operator twisted by $D$. For any $f\in \Aut^\ub(\cS,D)$, we have:
\beqan
&& f^\ast\circ \dd_D=\dd_D\circ f^\ast ~~\label{diffrel}\\
&& \dd_{D^{f_0}}\circ {\hat f} ={\hat f} \circ \dd_{D}\label{hatdiffrel}~~,
\eeqan
where ${\hat f}$ denotes $\id_{\Omega(\cM)}\otimes {\hat f}$. In
particular, $f^\ast$ induces $\R$-linear automorphisms of the twisted
cohomology spaces $H^k(M,\cS)$ while ${\hat f}$ induces $\R$-linear
isomorphisms $H^k(\cM,\cS)\stackrel{\sim}{\rightarrow}
H^k(\cM,\cS^{f_0})$.
\end{Proposition}

\begin{proof}
For any $\rho\in \Omega^k(\cM)$ and $\xi\in \Gamma(M,\cS)$, we have: 
\be
\dd_D(\rho\otimes \xi)=(\dd \rho)\otimes \xi+(-1)^k \rho\wedge D\xi~~.
\ee
Using  \eqref{Df} and Proposition \ref{prop:twistedpullback}, this gives:
\beqa
f^\ast(\dd_D(\rho\otimes \xi)) &=& f_0^\ast(\dd\rho)\otimes \mf^{-1}(\xi)+(-1)^k f_0^\ast(\rho)\wedge f^\ast(D\xi)\\
&=& (\dd f_0^\ast(\rho))\otimes \mf^{-1}(\xi)+(-1)^k f_0^\ast(\rho)\wedge D(\mf^{-1}(\xi))=\dd_D(f^\ast(\rho\otimes \xi))~~.
\eeqa
We also compute:
\beqa
\dd_{D^{f_0}}({\hat f}(\rho\otimes \xi))&=&\dd_{D^{f_0}}(\rho\otimes {\hat f}(\xi))=(\dd \rho)\otimes {\hat f}(\xi)+(-1)^k\rho\wedge D^{f_0}({\hat f}(\xi))\nn\\
&=& (\dd \rho)\otimes {\hat f}(\xi)+(-1)^k\rho\wedge {\hat f}(D\xi)={\hat f}(\dd_D(\rho\otimes \xi))~~,
\eeqa
where we used relation \eqref{Dhatf}.
\end{proof} 

\noindent Let $D^\ad:\Gamma(\cM,End(\cS))\rightarrow \Omega^1(\cM,End(\cS))$ be the adjoint connection of $D$.

\begin{Proposition}
For any $f\in \Aut^\ub(\cS,D)$, we have:
\ben
\label{fDad}
\Ad(f)^\ast\circ D^\ad=D^\ad \circ \bAd(f^{-1})~~.
\een
\end{Proposition}

\begin{proof}
For any $T\in \End(\cS)$, we have $D^\ad(T)\in
\Omega^1(\cM,End(\cS))$. For any $\xi\in \Gamma(\cM,\cS)$, relation
\eqref{Adfast} gives:
\be
\Ad(f)^\ast(D^\ad(T))(\xi)=f^\ast(D^\ad(T)(\mf(\xi)))=f^\ast(D(T\mf(\xi)))-f^\ast(TD(\mf(\xi)))=f^\ast(D(T\mf(\xi)))-(\mf^{-1}T\mf) f^\ast(D(\mf(\xi)))~~,
\ee
where in the last equality we used \eqref{fast}. Using \eqref{Df}, this gives: 
\be
\Ad(f)^\ast(D^\ad(T))(\xi)=D((\mf^{-1}T\mf)(\xi))-(\mf^{-1}T\mf) D(\xi)=(D\circ \bAd(f^{-1})(T)-\bAd(f^{-1})(T)\circ D)(\xi)=(D^\ad\circ \bAd(f^{-1})(T))(\xi)~,
\ee
where in the second equality we used \eqref{bAdxi}. Since $\xi$ and
$T$ are arbitrary, this gives \eqref{fDad}.
\end{proof}

\subsection{Symplectic, complex and unitary unbased automorphisms}

Let $(\cS,\omega)$ be a symplectic vector bundle over $\cM$ and $J$ be
a taming of $(\cS,\omega)$.

\begin{Definition}
An unbased automorphism $f\in \Aut^\ub(\cS)$ is called:
\begin{enumerate}[1.]
\item {\em symplectic}, if $f$ preserves $\omega$ in the sense that
  $f_p$ is a symplectic isomorphism from $(\cS_p,\omega_p)$ to
  $(\cS_{f_0(p)},\omega_{f_0(p)})$ for all $p\in \cM$.
\item {\em complex}, if $f$ preserves $J$ in the sense that $f_p\circ
  J_p=J_{f_0(p)}\circ f_p$ for all $p\in \cM$.
\item {\em unitary}, if $f$ is both symplectic and complex.
\end{enumerate}
\end{Definition}

\noindent Notice that $f$ is symplectic iff ${\hat f}$ is an
isomorphism of symplectic vector bundles from $(\cS,\omega)$ to
$(\cS^{f_0},\omega^{f_0})$ and $f$ is complex iff ${\hat f}$ is an
isomorphism of complex vector bundles from $(\cS,J)$ to
$(\cS^{f_0},J^{f_0})$. Finally, $f$ is unitary iff it is an
isomorphism of tamed symplectic vector bundles from $(\cS,J,\omega)$
to $(\cS^{f_0},J^{f_0},\omega^{f_0})$. Relation \eqref{bAdf2} implies
that $f$ is complex iff $\bAd(f)(J)=J$. Symplectic, complex and
unitary automorphisms form subgroups of $\Aut^\ub(\cS)$, which we
denote by respectively by $\Aut^\ub(\cS,\omega)$, $\Aut^\ub(\cS,J)$
and $\Aut^\ub(\cS,J,\omega)$. We have:
\be
\Aut^\ub(\cS,J,\omega)=\Aut^\ub(\cS,\omega)\cap \Aut^\ub(\cS,J)~~.
\ee
For any $f\in \Aut^\ub(\cS)$ and any complex structure $J$ of $\cS$,
the endomorphism $\bAd(f)(J)$ (see Appendix \ref{app:unbased}) is
again a complex structure of $\cS$. For any $f\in
\Aut^\ub(\cS,\omega)$ and any taming $J$ of $(\cS,\omega)$, the
endomorphism $\bAd(f)(J)$ is a taming of $(\cS,\omega)$ while
$\Ad({\hat f})(J)$ is a taming of $(\cS^{f_0},\omega^{f_0})$.
Relation \eqref{mfp2} implies that $f\in \Aut^\ub(\cS)$ is symplectic 
iff the following condition is satisfied: 
\ben
\label{fsymp}
\omega(\mf(\xi_1),\mf(\xi_2))=\omega(\xi_1,\xi_2)\circ f_0^{-1}~~,~~\forall \xi_1,\xi_2\in \Gamma(\cM,\cS)~~.
\een

\subsection{Symmetries of duality and electromagnetic structures}

Let $\Delta=(\cS,D,\omega)$ be a duality structure on $\cM$ and $J$ be a taming
of $(\cS,\omega)$. Let $\Xi=(\cS,D,J,\omega)$ be the corresponding electromagnetic
structure with underlying duality structure $\Xi_0=\Delta$.

\begin{Definition}
An {\em unbased} automorphism $f\in \Aut^\ub(\cS)$ is called:
\begin{enumerate}[1.]
\itemsep 0.0em
\item A {\em symmetry of the duality structure} $\Delta$, if $f$ is
  symplectic with respect to $\omega$ and covariantly constant with
  respect to $D$.
\item A {\em symmetry of the electromagnetic structure} $\Xi$, if $f$
  is complex with respect to $J$ and is a symmetry of the duality
  structure $\Delta$, i.e. if $f$ is unitary and covariantly-constant
  with respect to $D$.
\end{enumerate}
\end{Definition}

\

\noindent Let $\Aut^\ub(\Delta)=\Aut^\ub(\cS,D,\omega)$ and
$\Aut^\ub(\Xi)=\Aut^\ub(\cS,D,J,\omega)$ denote the groups of
symmetries of $\Delta$ and $\Xi$. We have:
\beqa
&& \Aut^\ub(\Xi)=\Aut^\ub(\Xi_0)\cap \Aut(\cS,J)=\Aut^\ub(\cS,D)\cap \Aut^\ub(\cS,J,\omega)\\
&& \Aut^\ub(\Delta)=\Aut^\ub(\cS,\omega)\cap \Aut^\ub(\cS,D)~~.
\eeqa

\subsection{The action of $\Aut^\ub(\Delta)$ on electromagnetic and scalar-electromagnetic structures}

Let $\Delta=(\cS,D,\omega)$ be a fixed duality structure on
$\cM$. Given a taming $J$ of $(\cS,\omega)$ and a symplectic
automorphism $f\in \Aut^\ub(\cS,\omega)$, the endomorphism
$\bAd(f)(J)$ is again a taming of $(\cS,\omega)$. Hence for any
electromagnetic structure $\Xi=(\cS,D,J,\omega)$ having $\Delta$ as
its underlying duality structure, the quadruplet:
\ben
\label{Xif}
\Xi_f\eqdef (\cS,D,\bAd(f)(J),\omega)
\een
is again an electromagnetic structure having $\Delta$ as its
underlying duality structure.  This defines a left action of the group
$\Aut^\ub(\cS,\omega)$ on the set $\ES_{\Delta}(\cM)\simeq
\fJ_+(\cS,\omega)$ of all electromagnetic structures whose underlying
duality structure equals $\Delta$. Let $Q_{J,\omega}$ denote the
Euclidean scalar product induced on $\cS$ by $\omega$ and $J$ (see
\eqref{Qdef}).

\begin{Proposition}
For any $f\in \Aut^\ub(\cS,\omega)$, we have: 
\ben
\label{Qtf}
Q_{\bAd(f)(J),\omega}^{f_0}(\hat{f}(\xi_1), \hat{f}(\xi_2)) =Q_{J,\omega}(\xi_1,\xi_2)~~,~~\forall \xi_1,\xi_2\in \Gamma(\cM,\cS)~~,
\een
where $Q_{\bAd(f)(J),\omega}^{f_0}\eqdef
(Q_{\bAd(f)(J),\omega})^{f_0}:\cS^{f_0}\times_\cM \cS^{f_0}\rightarrow
\R_\cM$ is the $f_0$-pullback of $Q_{\bAd(J),\omega}$.
\end{Proposition}

\begin{proof}
For any $\xi_1,\xi_2\in \Gamma(\cM,\cS)$, we have:
\be
Q_{\bAd(f)(J),\omega}(\mf(\xi_1),\mf(\xi_2))=\omega(\bAd(f)(J)(\mf(\xi_1)),\mf(\xi_2))=\omega(\mf(J\xi_1),\mf(\xi_2))=\omega(J\xi_1,\xi_2)\circ f_0^{-1}=Q_{J,\omega}(\xi_1,\xi_2)\circ f_0^{-1}~~,
\ee
where we used \eqref{bAdxi} and \eqref{fsymp}. Using relation \eqref{mf}, this gives:
\be
Q_{\bAd(f)(J),\omega}^{f_0}(\hat{f}(\xi_1), \hat{f}(\xi_2))=Q_{\bAd(f)(J),\omega}^{f_0}(\mf(\xi_1)^{f_0}, \mf(\xi_2)^{f_0})=Q_{\bAd(f)(J),\omega}(\mf(\xi_1),\mf(\xi_2))\circ f_0=Q_{J,\omega}(\xi_1,\xi_2)~~.
\ee
\end{proof}

\noindent The action of $\Aut^\ub(\cS,\omega)$ on $\ES_\Delta(\cM)$
restricts to an action of the subgroup $\Aut^\ub(\Delta)\subset
\Aut^\ub(\cS,\omega)$ on $\ES_\Delta(\cM)$.  For $f\in
\Aut^\ub(\Delta)$, the fundamental forms $\Theta_\Xi,\Theta_{\Xi_f}\in
\Omega^1(\cM,End(\cS))$ of $\Xi$ and $\Xi_f$ are related as follows:

\begin{Proposition}
For any $f\in \Aut^\ub(\Delta)$, we have: 
\ben
\label{Thetafast}
\Theta_{\Xi_f}=\Ad(f^{-1})^\ast(\Theta_\Xi)~~.
\een
\end{Proposition} 

\begin{proof}
Replacing $f$ with $f^{-1}$ in relation \eqref{fDad} gives $\Ad(f^{-1})^\ast\circ D^\ad=D^\ad\circ
\bAd(f)$. Thus:
\beqa
\Ad(f^{-1})^\ast(\Theta_{\Xi})&=& \Ad(f^{-1})^\ast (D^\ad(J))=D^\ad(\bAd(f)(J))=\Theta_{\Xi_f}~~.
\eeqa
\end{proof}

\noindent Let $\cD=(\Sigma,\Xi)$ be a scalar-electromagnetic
structure, where $\Sigma=(\cM,\cG,\Phi)$ and $\Xi=(\cS,D,J, \omega)$,
with underlying duality structure $\Delta=(\cS,D,\omega)$.  For any
$f\in \Aut^\ub(\Delta)$, consider the scalar-electromagnetic structure:
\ben
\label{cDfDef}
\cD_f\eqdef (\Sigma_f, \Xi_f)~~,
\een
where: 
\ben
\label{SigmafDef}
\Sigma_f\eqdef (\cM,(f_0^{-1})^\ast(\cG),\Phi\circ f_0^{-1})~~.
\een
This defines a left action of $\Aut^\ub(\Delta)$ on the set of
scalar-electromagnetic structures whose scalar manifold is $\cM$ and
whose underlying duality structure is $\Delta$. The fundamental fields
of $\cD$ and $\cD_f$ are related as follows:

\begin{Proposition}
For any $f\in \Aut^\ub(\Delta)$, we have: 
\ben
\label{Psitf}
\Psi_{\cD_f}=\Ad(f)_\ast(\Psi_\cD)~~.
\een
\end{Proposition}

\begin{proof}
Proposition \ref{prop:pullpush} (with $\cS$ replaced by $End(\cS)$ and
$f$ replaced by $\Ad(f)$) implies that the following relation holds
for any $\Theta\in \Omega^1(\cM,End(\cS))$:
\ben
\label{Addual}
[\Ad(f)^\ast(\Theta)]^{\sharp_{f_0^\ast(\cG)}}=(\Ad(f^{-1}))_\ast (\Theta^{\sharp_\cG})~~.
\een
Thus:
\be
\Psi_{\cD_f}=(\Theta_{\Xi_f})^{\sharp_{(f_0^{-1})^\ast(\cG)}}=[\Ad(f^{-1})^\ast(\Theta_\Xi)]^{\sharp_{(f_0^{-1})^\ast(\cG)}}=\Ad(f)_\ast [(\Theta_\Xi)^{\sharp_\cG}]=\Ad(f)_\ast(\Psi_{\cD})~~,
\ee
where in the second equality we used \eqref{Thetafast} and in the
third equality we used \eqref{Addual} with $f$ replaced by
$f^{-1}$.
\end{proof}

\subsection{The action of $\Aut^\ub(\cS,D,\omega)$ on electromagnetic fields}

Let $\cD=(\Sigma,\Xi)$ be a scalar-electromagnetic structure, where $\Sigma=(\cM,\cG,\Phi)$ and $\Xi=(\cS,D,J,\omega)$, 
with underlying duality structure $\Delta=(\cS,D,\omega)$. Let $M$ be a
four-manifold and $\varphi:M\rightarrow \cM$ be a smooth map. The
pull-back of the based isomorphism ${\hat
  f}:\cS\stackrel{\sim}{\rightarrow} \cS^{f_0}$ gives a based
isomorphism of vector bundles ${\hat
  f}^\varphi:\cS^\varphi\stackrel{\sim}{\rightarrow} \cS^{f_0\circ
  \varphi}=(\cS^{f_0})^{\varphi}$. This extends to an isomorphism from
$\wedge T^\ast M\otimes \cS^\varphi$ to $\wedge T^\ast M\otimes
\cS^{f_0\circ \varphi}$ (where the action on $\wedge T^\ast M$ is
trivial) and hence induces an isomorphism of $\cinf$-modules which we
denote by the same symbol:
\ben
\label{linaction0}
{\hat f}^\varphi\equiv \id_{\Omega(\cM)}\otimes_{\cC^\infty(\cM,\R)} {\hat f}^\varphi:\Omega(M,\cS^\varphi) \rightarrow \Omega(M,\cS^{f_0\circ \varphi})~~.
\een 
We have:
\be
{\hat f}^\varphi(\rho\otimes \xi)=\rho\otimes {\hat f}^\varphi(\xi)~~,~~\forall \rho\in \Omega(M)~~\forall \xi\in \Gamma(M,\cS^\varphi)~~.
\ee

\begin{Lemma}
\label{lemma:Dcommutation}
For any $f\in \Aut^\ub(\cS,D)$, we have:
\ben
\label{diffrelphi}
\dd_{D^{f_0\circ \varphi}}\circ {\hat f}^\varphi ={\hat f}^\varphi \circ \dd_{D^\varphi}~~.
\een
\end{Lemma}

\begin{proof}
For any $\rho\in \Omega^k(M)$ and any $\xi\in \Gamma(M,\cS^\varphi)$, we have: 
\ben
\label{s1}
(\dd_{D^{f_0\circ \varphi}}\circ {\hat f}^\varphi)(\rho\otimes \xi)=(\dd\rho)\otimes {\hat f}^\varphi(\xi)+(-1)^k \rho\wedge D^{f_0\circ \varphi}({\hat f}^\varphi(\xi))~~.
\een
Relation \eqref{Dhatf} implies:
\be
D^{f_0\circ \varphi}\circ {\hat f}^\varphi={\hat f}^\varphi\circ D^\varphi~~.
\ee
Hence \eqref{s1} becomes: 
\be
(\dd_{D^{f_0\circ \varphi}}\circ {\hat f}^\varphi)(\rho\otimes \xi)=(\dd\rho)\otimes {\hat f}^\varphi(\xi)+(-1)^k \rho \wedge {\hat f}^\varphi( D^\varphi(\xi))=({\hat f}^\varphi\circ \dd_{D^\varphi})(\rho\otimes \xi)~~.
\ee
\end{proof} 

\begin{Proposition}
\label{prop:emtf}
The following properties hold:
\begin{enumerate}[1.]
\item For any $f\in \Aut^\ub(\cS,\omega)$, we have:
\be
{\hat f}^\varphi(\iOmega^{\Xi,g,\varphi}(M))=\iOmega^{\Xi_f,g,f_0\circ \varphi}(M)
\ee
Thus a 2-form $\cV\in \Omega^2(M,\cS^\varphi)$ satisfies the positive
polarization condition with respect to $g$ and $J$ in the scalar
field background $\varphi$ iff ${\hat f}^\varphi(\cV)$ satisfies the
positive polarization condition with respect to $g$ and 
$\bAd(f)(J)$ in the scalar field background $f_0\circ \varphi$.
\item  For any $f\in \Aut^\ub(\Delta)$, we have:
\be
{\hat f}^\varphi(\iOmega_\cl^{\Xi,g,\varphi}(M))=\iOmega_\cl^{\Xi_f,g,f_0\circ \varphi}(M)~~.
\ee
Thus, when the duality structure $\Delta=(\cS,D,\omega)$ is fixed, a
2-form $\cV\in \Omega^2(M,\cS^\varphi)$ is an electromagnetic field strength
with respect to $g$ and $J$ in the scalar field background $\varphi$
iff ${\hat f}^\varphi(\cV)$ is an electromagnetic field strength with respect
to $g$ and $\bAd(f)(J)$ in the scalar field background $f_0\circ
\varphi$.
\end{enumerate}
\end{Proposition}

\begin{proof}

\

\begin{enumerate}[1.]
\item We have $\bAd(f)(J)^{f_0}=\Ad({\hat f})(J)$ by relation
  \eqref{AdHatAd}. Thus $\bAd(f)(J)^{f_0\circ \varphi} \circ
  \hat{f}^\varphi=\Ad({\hat f})(J)^\varphi \circ
  \hat{f}^\varphi=[\Ad({\hat f})(J)\circ {\hat f}]^\varphi= ({\hat
    f}\circ J)^\varphi={\hat f}^\varphi \circ J^\varphi$, where we
  used \eqref{AdHatAd}. Since $\ast_g$ acts trivially on
  $\cS^\varphi$, we also have $\ast_g \circ {\hat f}^\varphi={\hat
    f}^\varphi \circ \bast_{g}$. Hence a 2-form $\eta\in
  \Omega^2(M,\cS^\varphi)$ satisfies the positive polarization
  condition $\ast_g\eta=-J^\varphi\eta$ iff the 2-form ${\hat
    f}^\varphi(\eta)\in \Omega^2(M,\cS^{f_0\circ \varphi})$ satisfies
  the positive polarization condition $\ast_g\hat{f}^\varphi
  (\eta)=-\bAd(f)(J)^{f_0\circ \varphi} \hat{f}^\varphi(\eta)$.
\item Relation \eqref{diffrelphi} implies that $\eta\in
\Omega(M,\cS^\varphi)$ is $\dd_{D^\varphi}$-closed iff ${\hat
  f}^\varphi (\eta)$ is $\dd_{D^{f_0\circ \varphi}}$-closed. Combining
this with the statement at point 1. gives the conclusion.
\end{enumerate}
\end{proof}

\noindent Let $Q_{J,\omega}$ denote the scalar product induced on
$\cS$ by $J$ and $\omega$ and $(~,~)_{g, J,\omega,\varphi}=(~,~)_{g,Q_{J,\omega}^\varphi}$ denote the
pseudo-Euclidean scalar product induced by $g$ and $Q_{J,\omega}^\varphi$ on
$\wedge T^\ast M\otimes \cS^\varphi$ (See subsection \ref{sec:completetheoryI}).

\begin{Lemma}
For any $f\in \Aut^\ub(\cS,\omega)$ and $\cV_1,\cV_2\in \Omega^2(M,\cS^\varphi)$, we have: 
\ben
\label{prodtf}
({\hat f}^\varphi(\cV_1),{\hat f}^\varphi(\cV_2))_{g, \bAd(f)(J),\omega, f_0\circ \varphi}=(\cV_1,\cV_2)_{g, J,\omega, \varphi}~~.
\een
\end{Lemma}

\begin{proof}
It suffices to consider the case $\cV_i=\rho_i\otimes \xi_i$ with
$\rho_i\in \Omega^2(M)$ and $\xi_i\in \Gamma(M,\cS^\varphi)$ (where
$i=1,2$). Then:
\be
({\hat f}^\varphi(\cV_1),{\hat f}^\varphi(\cV_2))_{g,\bAd(f)(J),\omega, f_0\circ \varphi}=(\rho_1,\rho_2)_g Q_{\bAd(f)(J),\omega}^{f_0\circ\varphi}({\hat f}^\varphi(\xi_1),{\hat f}^\varphi(\xi_2))=
(\rho_1,\rho_2)_g Q_{J,\omega}^{\varphi}(\xi_1,\xi_2)=(\cV_1,\cV_2)_{g, J,\omega, \varphi}~~,
\ee
where $(~,~)_g$ denotes the Euclidean scalar product induced by $g$ on
the exterior bundle $\wedge T^\ast M$ and we used relation
\eqref{Qtf}.
\end{proof}

\noindent Recall the twisted inner contraction 
$\loslash:=\loslash_{g, J,\omega, \varphi}:
\wedge_M^2(\cS^\varphi)\times_M\wedge_M^2(\cS^\varphi)\rightarrow
\otimes^2(T^\ast M)$ of Definition \ref{def:twistedinner}.

\begin{Proposition}
\label{prop:VPsiVtf}
The following properties hold for any $\cV_1,\cV_2\in \Omega^2(M,\cS^\varphi)$: 
\begin{enumerate}[1.]
\itemsep 0.0em
\item For any $f\in \Aut^\ub(\cS,\omega)$, we have:
\ben
\label{VVtf}
\hat{f}^{\varphi}(\cV_1) \loslash_{g, \bAd(f)(J),\omega, f_0\circ \varphi} \hat{f}^{\varphi}(\cV_2)=\cV_1\loslash_{g, J,\omega, \varphi}\cV_2~~.
\een
\item For any $f\in \Aut^\ub(\Delta)$, we have: 
\ben
\label{PsiVtf}
(\Psi_{\cD_f}^{f_{0}\circ\varphi }\hat{f}^{\varphi}(\cV_1),
\hat{f}^{\varphi}(\cV_2))_{g, \bAd(f)(J),\omega, f_0\circ
  \varphi}=(\widehat{\dd f_0})^\varphi ((\Psi_{\cD}^\varphi\cV_1,
\cV_2)_{g, J,\omega, \varphi})~~,
\een
where $(\widehat{\dd f_0})^\varphi:(T\cM)^\varphi\rightarrow
(T\cM)^{f_0\circ \varphi}$ is the $\varphi$-pullback of the map
$\widehat{\dd f_0}:T\cM\rightarrow (T\cM)^{f_0}$ induced by the
differential $\dd f_0:T\cM\rightarrow T\cM$.
\end{enumerate}
\end{Proposition}

\begin{proof}
It suffices to consider the case $\cV_i = \rho_i \otimes \xi_i$, with
$\rho_i\in \Omega^2(M)$ and $\xi_i\in \Gamma(M,\cS^\varphi)$. Then:
\ben
\label{E1}
\hat{f}^{\varphi}(\cV_1) \loslash_{g,\bAd(f)(J),\omega, f_0\circ \varphi} \hat{f}^{\varphi}(\cV_2) = Q^{f_{0}\circ \varphi}_{\bAd(f)(J),\omega}(\hat{f}^{\varphi}(\xi_1),\hat{f}^{\varphi}(\xi_2)) \rho_1\oslash_g \rho_2=
Q^{\varphi}_{J,\omega}(\xi_1,\xi_2) \rho_1\oslash_g \rho_2=\cV_1\loslash_{g,J,\varphi,\omega}\cV_2~~,
\een
where we used \eqref{Qtf}. This shows that \eqref{VVtf} holds. Using
relations \eqref{Psitf} and \eqref{Adpush}, we find:
\be
R_{\hat f}(\Psi_{\cD_f}^{f_{0}})=R_{\hat f}(\Ad(f)_\ast (\Psi_\cD)^{f_0})=(\widehat{\dd f_0}\otimes L_{{\hat f}})(\Psi_\cD)~~.
\ee
Remembering that $(~,~)_{g,J,\omega,\varphi}$ is extended trivially to
$(T\cM)^\varphi$ (see Subsection \ref{sec:completetheoryI}), this
gives:
\beqa
&& (\Psi_{\cD_f}^{f_{0}\circ\varphi} \hat{f}^{\varphi}(\cV_1),
\hat{f}^{\varphi}(\cV_2) )_{g, \bAd(f)(J),\omega, f_0\circ \varphi}
=((\widehat{\dd f_0}\otimes L_{{\hat f}})^\varphi (\Psi_\cD^\varphi)
(\cV_1),\hat{f}^{\varphi}(\cV_2))_{g, \bAd(f)(J),\omega, f_0\circ
  \varphi}=\\ && =(\widehat{\dd f_0})^\varphi
((\hat{f}^{\varphi}(\Psi_\cD^\varphi\cV_1),
\hat{f}^{\varphi}(\cV_2))_{\bAd(f)(J),f_0\circ \varphi}) =(\widehat{\dd
  f_0})^\varphi (( \Psi_\cD^\varphi \cV_1, \cV_2)_{g, J,\varphi,
  \omega})~~,
\eeqa
where in the last line we used relation \eqref{prodtf}. Thus
\eqref{PsiVtf} holds.
\end{proof}

\subsection{Scalar-electromagnetic dualities and symmetries}

Let $\cD=(\Sigma,\Xi)$ be a scalar-electromagnetic structure with
underlying scalar structure $\Sigma=(\cM,\cG,\Phi)$ and underlying
electromagnetic structure $\Xi=(\cS,D,J,\omega)$. Let
$\Delta=(\cS,D,\omega)$ be the underlying duality structure of $\Xi$ and
$\cD_0=(\Sigma,\Delta)$ be the underlying scalar-duality structure of
$\cD$. Fixing a Lorentzian metric $g\in Met_{3,1}(M)$, consider the set:
\ben
\label{ConfcD0}
\Conf^g_{\cD_0}(M)\eqdef \cup_{J'\in \fJ_+(\cS,\omega)}\Conf^g_{(\Sigma,(\cS,D,J',\omega))}(M)~~.
\een

\begin{Definition} 
The {\em scalar-electromagnetic duality group} of $\cD_0$ is the
subgroup of $\Aut^\ub(\Delta)$ defined through:
\be
\Aut(\cD_0)\eqdef \{f\in \Aut^\ub(\Delta)|f_0\in \Aut(\Sigma)\}=\{f\in \Aut^\ub(\Delta)|f_0\in \Iso(\cM,\cG)~\&~\Phi\circ f_0=\Phi)\}~~.
\ee
An element of this group is called a {\em scalar-electromagnetic
  duality}. The {\em duality action} is the action of
$\Aut(\cD_0)$ on the set $\Conf^g_{\cD_0}(M)$ given by:
\be
\label{dualityaction0}
f \diamond (\varphi, \cV)\eqdef (f_0\circ \varphi, {\hat f}^\varphi (\cV))~~\forall f\in \Aut(\cD_0)~~.
\ee
\end{Definition}

\noindent For any $f\in \Aut(\Sigma)$, we have $\Sigma_f=\Sigma$ and $\cD_f=(\Sigma,\Xi_f)$, 
where $\Sigma_f$ and $\cD_f$ were defined in \eqref{SigmafDef} and \eqref{cDfDef}.

\begin{Theorem}
\label{thm:DualityInvariance}
For any $g\in Met_{3,1}(M)$ and any scalar-electromagnetic duality
transformation $f\in \Aut(\cD_0)$, we have:
\ben
f \diamond \Sol^g_{\cD}(M)=\Sol^g_{\cD_f}(M)~~.
\een
\end{Theorem}

\begin{proof}
Follows immediately from Proposition \ref{prop:thetatf}, Proposition
\ref{prop:emtf} and Proposition \ref{prop:VPsiVtf} upon using the form
of the scalar-electromagnetic equations \eqref{sc} and \eqref{em}.
\end{proof}

\begin{Definition}
The {\em scalar-electromagnetic symmetry group} of $\cD$ is the subgroup
$\Aut(\cD)$ of $\Aut^\ub(\Xi)$ defined through:
\be
\Aut(\cD)\eqdef\{f\in \Aut^\ub(\Xi)|f\in \Aut(\Sigma)\}=\{f\in \Aut^\ub(\Xi)|f_0\in \Iso(\cM,\cG)~\&~\Phi\circ f_0=\Phi\}~~.
\ee
An element of this group is called a {\em scalar-electromagnetic
  symmetry}.
\end{Definition}

\noindent We have: 
\be
\Aut(\cD)\eqdef\{f\in \Aut(\cD_0)|\bAd(f)(J)=J\}~~.
\ee
In particular, $\Aut(\cD)$ is a subgroup of $\Aut(\cD_0)$. 

\begin{Corollary}
For all $f\in \Aut(\cD)$, we have:
\be
f\diamond \Sol^g_{\cD}(M)=\Sol^g_{\cD}(M)~~.
\ee
Thus $\Aut(\cD)$ consists of symmetries of the scalar-electromagnetic
equations of motion \eqref{sc} and \eqref{em}, for any fixed Lorentzian
metric $g$.
\end{Corollary}

\begin{proof}
For any $f\in \Aut(\cD)$, we have $D_{f} = \cD$ and the desired statement follows
from Theorem \ref{thm:DualityInvariance}.
\end{proof}

\

\noindent We have short exact sequences:
\beqa
&& 1\rightarrow \Aut(\Delta)\hookrightarrow \Aut(\cD_0)\longrightarrow \Aut^\Delta(\Sigma)\rightarrow 1~~\\
&& 1\rightarrow \Aut(\Xi)\hookrightarrow \Aut(\cD)\longrightarrow \Aut^\Xi(\Sigma)\rightarrow 1~~,
\eeqa
where $\Aut(\Delta)$ and $\Aut(\Xi)$ are the groups of {\em based}
symmetries of $\Delta$ and $\Xi$, i.e. the groups of automorphism of
$\Delta$ and $\Xi$ in the categories $DS(\cM)$ and $ES(\cM)$,
respectively. The groups appearing in the right hand side consist of
those automorphisms of the scalar structure $\Sigma$ which
respectively admit lifts to scalar-electromagnetic dualities of
$\cD_0=(\Sigma,\Delta)$ and scalar-electromagnetic symmetries of
$\cD=(\Sigma,\Xi)$. Fixing a point $p\in \cM$, we can identify
$\Aut(\Delta)$ with the commutant of $\Hol_D^p$ inside the group
$\Aut(\cS_p,\omega_p)\simeq \Sp(2n,\R)$.  In particular, the exact
sequences above show that $\Aut(\cD_0)$ and $\Aut(\cD)$ are Lie
groups.

\begin{Remark}
Consider the case when the duality structure $\Delta$ is trivial. Then
$\Hol_D^p=1$ and we have $\Aut(\Delta)\simeq \Sp(2n,\R)$.  In this
case, we have $\Aut^\Delta (\Sigma)=\Aut(\Sigma)$. Moreover, the first
exact sequence above splits, giving an isomorphism $\Aut(\cD_0)\simeq
\Sp(2n,\R)\times \Aut(\Sigma)$. This recovers the
scalar-electromagnetic duality group traditionally discussed in the
supergravity literature. The taming $J$ can be identified with a
smooth map $\tau:\cM\rightarrow \SH_n$ into the Siegel upper half
space (where $\rk\cS=2n$) and we have isomorphisms of groups:
\beqan
\Aut(\Xi) &\simeq & \{A\in \Sp(2n,\R)|A\bullet\tau=\tau\}\nn\\
\Aut(\cD)&\simeq& \{(\psi,A)\in \Aut(\Sigma)\times \Sp(2n,\R)|\tau\circ \psi=A\bullet \tau\}~~,
\eeqan
where $\bullet$ is the action of $\Sp(2n,\R)$ on $\SH_n$ by matrix
fractional transformations (see Appendix \ref{app:space_tamings}).
The electromagnetic structure $\Xi$ is unitary iff $\tau$ is a
constant map, in which case we have $\Aut(\Xi)\simeq \U(n)$,
$\Aut^\Xi(\Sigma)=\Aut(\Sigma)$ and $\Aut(\cD)\simeq \U(n)\times
\Aut(\Sigma)$. It is clear from the above that the group of
scalar-electromagnetic dualities is considerably more complicated in
the general case when the duality structure $\Delta$ is nontrivial.
\end{Remark}

\section{The Dirac quantization condition}
\label{sec:Dirac}

In this section, we explain how the Dirac quantization condition can
be formulated when the model has a non-trivial duality structure
$\Delta=(\cS,D,\omega)$ defined on the scalar manifold $\cM$. In this
situation, the Dirac lattice is replaced by a {\em Dirac system},
defined as a smooth fiber sub-bundle $\Lambda$ of full lattices inside
$\cS$, which is preserved by the parallel transport of the flat
connection $D$ and which has the property that $\omega$ is
integer-valued when restricted to $\Lambda$. Pairs
$\bDelta=(\Delta,\Lambda)$ are called {\em integral duality
  structures}. The fibers of $\bDelta$ are objects of a category
$\Symp_0$ of {\em integral symplectic spaces}, which is discussed in
Appendix \ref{app:integral}. Integral duality structures correspond to
local systems $T_\bDelta$ valued in the unit groupoid of $\Symp_0$.
Any integral duality structure has a {\em type} $\bt$, which is a
vector of positive integers obeying certain divisibility
properties. Integral duality structures of type $\bt$ defined over
$\cM$ are classified by points of the $\Sp_\bt(2n,\Z)$-character
variety of the fundamental group $\pi_1(\cM)$, where
$\Sp_\bt(2n,\Z)\subset \Sp(2n,\R)$ is the modified Siegel modular
group of type $\bt$. A taming $J$ of $(\cS,\omega)$ defines an {\em
  integral electromagnetic structure} $\bXi=(\Xi,\Lambda)$, where
$\Xi=(\cS,D,J,\omega)$ is the electromagnetic structure defined by $J$
and $\Delta$. The fibers of an integral electromagnetic structure are
{\em integral tamed symplectic spaces}, being objects of a category
$\TamedSymp_0$ which is discussed in Appendix \ref{app:integral}.
As explained in Appendix \ref{app:integral}, there exists a
commutative diagram of categories and functors whose vertical arrows
are equivalences and whose horizontal arrows are forgetful functors 
(this is one of the sides of diagram \eqref{diag:cat}):
\be
\label{diag:cat0}
\scalebox{1.2}{
\xymatrix{
\TamedSymp_0 \ar_{X_h}[d] \ar[r] & \Symp_0\!\!\!\ar^{X_s}[d]\\
\AbVar \ar[r]  & \TorSymp  \\
}}
\ee
Here, $\TorSymp$ denotes the category of symplectic tori and $\AbVar$
is the category of polarized complex Abelian varieties. The vertical
equivalences are given by taking the quotient through the lattice. The
equivalence of categories $X_s$ allows one to encode the data of an
integral duality structure using a bundle $\cX_s(\bDelta)$ of symplectic
tori, endowed with a flat Ehresmann connection whose transport
preserves the symplectic structure of the fibers. The equivalence of
categories $X_h$ allows one to encode the data of an integral
electromagnetic structure using a (smooth) bundle $\cX_s(\bXi)$ of
polarized Abelian varieties, endowed with a flat Ehresmann connection
whose transport preserves the symplectic structure of the torus fibers
but need not preserve their complex structure.

A smooth scalar map $\varphi:M\rightarrow \cM$ induces a pulled-back
integral electromagnetic structure
$\Xi^\varphi=(\cS^\varphi,D^\varphi,J^\varphi,
\omega^\varphi,\Lambda^\varphi)$ defined on $M$. The Dirac
quantization condition requires that the de Rham cohomology class
$[\cV]\in H^2_{\dd_{D^\varphi}}(M,\cS^\varphi)$ of the electromagnetic
field strength $\cV$ of a {\em semiclassical} (as opposed to
classical) Abelian gauge field configuration be {\em integral} in the
sense that it belongs to the image through the universal coefficient
map of the second cohomology of $M$ with coefficients in the
$\Symp_0$-valued local system $T_{\bDelta^\varphi}$ defined by the
pulled-back integral duality structure
$\bDelta^\varphi=(\Delta^\varphi, \Lambda^\varphi)$.  This condition
indicates that a geometric model of semiclassical Abelian field
configurations can be constructed using a certain version of twisted
differential cohomology.

\subsection{Dirac systems and integral duality structures}

Let $\Delta=(\cS,D,\omega)$ be a duality structure of rank $2n$
defined over a manifold $N$. As explained in Appendix
\ref{app:integral}, an {\em integral symplectic space} is a
finite-dimensional vector space over $\R$ containing a full lattice on
which the symplectic pairing takes integer values.  Integral
symplectic spaces form a category denoted by $\Symp_0$. The elementary
divisor theorem implies that an integral symplectic space of dimension
$2n$ is determined up to isomorphism by its {\em type} $\bt$, which
is an element of the set $\Div^n$ of ordered $n$-uples of positive
integers $(t_1,\ldots, t_n)$ satisfying the divisibility conditions
$t_1|t_2|\ldots |t_n$. The integral symplectic space is called {\em principal} 
if its type equals $\bdelta(n)\eqdef (1,1,\ldots, 1)\in \Div^n$. 

\begin{Definition}
A {\em Dirac system} for $\Delta$ is a fiber sub-bundle
$\Lambda\subset \cS$ which satisfies the following conditions:
\begin{enumerate}
\item For any $x\in N$, the triple $(\cS_x,\omega_x,\Lambda_x)$ is an
  integral symplectic space, i.e. $\Lambda_x$ is a full lattice inside
  $\cS_x$ and we have $\omega_x(\Lambda_x,\Lambda_x)\subset \Z$.
\item $\Lambda$ is invariant under the parallel transport $T^\Delta$
  of $D$, in the sense that the following condition is satisfied for
  any path $\gamma\in \cP(N)$:
\ben
\label{LambdaFlatness}
T_\gamma^\Delta(\Lambda_{\gamma(0)})=\Lambda_{\gamma(1)}~~.
\een
\end{enumerate}
For every $x\in N$, the lattice $\Lambda_x\subset \cS_x$ is called the
{\em Dirac lattice defined by $\Lambda$ at the point $x$}. A duality
structure is called {\em semiclassical} if it admits a Dirac system.
\end{Definition}

\begin{Definition}
An {\em integral duality structure} defined on $N$ is a pair
$\bDelta\eqdef (\Delta,\Lambda)$, where $\Delta$ is is a duality
structure defined on $N$ and $\Lambda$ is a Dirac system for
$\Delta$. 
\end{Definition}

\

\noindent Since $T_\gamma^\Delta$ is an isomorphism of symplectic
vector spaces from $(\cS_{\gamma(0)},\omega_{\gamma(0)})$ to
$(\cS_{\gamma(1)},\omega_{\gamma(1)})$ and $N$ is connected, relation
\eqref{LambdaFlatness} implies that the type $\bt$ of the integral
symplectic space $(\cS_x,\omega_x,\Lambda_x)$ (see Appendix
\ref{app:integral}) does not depend on the point $x\in N$. This
quantity is denoted $\bt(\bDelta)$ and called the {\em type} of the
integral duality structure $\bDelta$. The integral duality structure
$\bDelta$ (and the corresponding Dirac system $\Lambda$) is called
{\em principal} if $\bt(\bDelta)=\bdelta(n)$,
i.e. if every fiber $(\cS_x,\omega_x,\Lambda_x)$ is a principal
integral symplectic space.

\begin{Definition}
Let $\bDelta=(\Delta_1,\Lambda_1)$ and
$\bDelta_2=(\Delta_2,\Lambda_2)$ be two integral duality structures
defined on $N$. A {\em morphism of of integral duality structures}
from $\bDelta_1$ to $\bDelta_2$ is a morphism of duality structures
$f:\Delta_1\rightarrow \Delta_2$ such that $f(\Lambda_1)\subset
\Lambda_2$.
\end{Definition}

\noindent With this definition, integral duality structures defined on
$N$ form a category which we denote by $\DS_0(N)$. Let
$\bDelta=(\Delta,\Lambda)$ be an integral duality structure defined on
$N$, where $\Delta=(\cS,D,\omega)$.

\begin{Definition}
The {\em parallel transport functor of $\bDelta$} is the functor
$T_{\bDelta}:\Pi_1(N)\rightarrow \Symp_0^\times$ which associates to
any point $x\in N$ the integral symplectic space $T_\bDelta(x)\eqdef
(\cS_x,\omega_x,\Lambda_x)$ and to any homotopy class $c\in \Pi_1(N)$
with fixed endpoints the isomorphism of integral symplectic spaces
$T_\bDelta(c)=T_\gamma^\Delta:(\cS_{\gamma(0)},\omega_{\gamma(0)},\Lambda_{\gamma(0)})\rightarrow
(\cS_{\gamma(1)},\omega_{\gamma(1)},\Lambda_{\gamma(1)})$, where
$\gamma\in \cP(N)$ is any path representing $c$ and $T_\gamma^\Delta$
is the parallel transport of $D$ along $\gamma$.
\end{Definition}

\

\noindent Clearly $T_\bDelta$ is a well-defined
$\Symp_0^\times$-valued local system on $N$. The correspondence which
associates $T_\bDelta$ to $\bDelta$ gives an equivalence of categories
between $\DS_0(N)$ and the functor category
$[\Pi_1(N),\Symp_0^\times]$. The parallel transport functors
$T_\bDelta$ and $T_\Delta$ of $\bDelta$ and $\Delta$ are related
through:
\be
T_\Delta=\cF\circ T_\bDelta~~,
\ee
where $\cF:\Symp_0^\times \rightarrow \Symp^\times$ is the surjective
functor which forgets the lattice. Hence giving a Dirac system on
$\Delta$ amounts to specifying a lift of $T_\Delta$ along $\cF$. As
explained in the next paragraph, the obstruction to existence of such
a lift can be described using the holonomy of $D$.

\paragraph{Classification of integral duality structures.}

Let $\bDelta=(\cS,D,\omega,\Lambda)$ be an integral duality structure
and consider a point $x\in N$. Then the Dirac lattice
$\Lambda_x\subset \cS_x$ at $x$ is invariant under the tautological
(fundamental) action of the holonomy group $\Hol_D^x$ (see
\eqref{hol}). Hence the latter consists of automorphisms
of the integral symplectic space $(\cS_x,\omega_x,\Lambda_x)$:
\be
\Hol_D^x=T_\bDelta(\pi_1(N,x))\subset \Sp(\cS_x,\omega_x,\Lambda_x)~~.
\ee
Accordingly, the holonomy representation \eqref{holrep} at $x$ gives a group morphism:
\be
\hol_D^x=T_\bDelta|_{\pi_1(N,x)}:\pi_1(N,x)\rightarrow \Sp(\cS_x,\omega_x,\Lambda_x)~~.
\ee
Let $\bt$ be the type of $(\cS,D,\omega)$. By Proposition
\ref{prop:intsympbasis}, any integral symplectic basis of
$(\cS_x,\omega_x,\Lambda_x)$ induces an isomorphism of integral
symplectic spaces $(\cS_x,\omega_x,\Lambda_x)\simeq
(\R^{2n},\omega_n,\Lambda_\bt)$ and an isomorphism of groups
$\Sp(\cS_x,\omega_x,\Lambda_x)\simeq \Sp_\bt(2n,\Z)$, where $\Sp_\bt(2n,\Z)$ is 
modified Siegel modular group of type $\bt$ and $\Lambda_\bt$ is the 
full lattice inside $\R^{2n}$ defined as in \eqref{LambdaDef}. Let ${\tilde
  N}\rightarrow N$ be the universal covering space (viewed as a
principal $\pi_1(N,x)$-bundle defined over $N$). Then $\Lambda$ is
isomorphic with the fiber bundle associated to ${\tilde N}$ through
the corestriction of $\hol_D^x$ to $\Lambda_x$:
\be
\Lambda\simeq {\tilde N}\times_{\hol_D^x}\Lambda_x~~.
\ee
This implies that the set of isomorphism classes of integral duality
structures of type $\bt$ defined on $N$ is in bijection with the
character variety:
\be
C_{\pi_1(N)}(\Sp_\bt(2n,\Z))=\Hom(\pi_1(N),\Sp_\bt(2n,\Z))/\Sp_\bt(2n,\Z)~~.
\ee
In particular, a duality structure $\Delta=(\cS,D,\omega)$ admits a
Dirac system of type $\bt$ iff its holonomy representation $\hol_D^x$
can be conjugated such that its image lies inside the group
$\Sp_\bt(2n,\Z)$.

\paragraph{The bundle of symplectic tori defined by an integral duality structure.}

Let $\bDelta\eqdef (\cS,D,\omega,\Lambda)$ be an integral duality
structure or rank $2n$ and type $\bt\in \Div^n$, defined on $N$.  For
any $x\in N$, the integral symplectic space
$(\cS_x,\omega_x,\Lambda_x)$ defines a symplectic torus
$X_s(\cS_x,\omega_x,\Lambda_x)$ of type $\bt$. When $x$ varies in $N$,
these symplectic tori fit into a fiber bundle $\cX_s(\bDelta)$, endowed
with a complete flat Ehresmann connection\footnote{A complement of the
  vertical distribution inside the tangent bundle of the total space.}
$\cH_\bDelta$induced by $D$. The Ehresmann transport of this
connection proceeds through isomorphisms of symplectic tori, so it
preserves the Abelian group structure and the symplectic structure of
the fibers. In particular, the holonomy group of $\cH_\bDelta$ is
contained in $\Sp_\bt(2n,\Z)$.

\begin{Definition}
The fiber bundle with connection $(\cX_s(\bDelta),\cH_\bDelta)$ is called the {\em bundle of
  symplectic tori} defined by the integral duality structure
$\bDelta$.
\end{Definition}

\subsection{Integral electromagnetic structures}

The formulation of the Dirac quantization condition for our models
involves the notion of integral electromagnetic structure, which we
discuss in this subsection.

\begin{Definition} 
An {\em integral electromagnetic structure} defined on $N$ is a pair
$\bXi=(\Xi,\Lambda)$, where $\Xi=(\Delta,J)$ is an electromagnetic
structure defined on $N$ and $\Lambda$ is a Dirac system for the
underlying duality structure $\Delta=(\cS,D,\omega)$ of $\Xi$.
The type of the integral duality structure $\bDelta=(\cS,D,\omega,
\Lambda)$ is called the {\em type} $\bt(\bXi)$ of $\bXi$:
\be
\bt(\bXi)\eqdef \bt(\bDelta)~~.
\ee
\end{Definition}

\

\noindent Let $\bXi=(\cS,D, J,\omega,\Lambda)$ be an integral
electromagnetic structure of real rank $2n$ and type $\bt$, with
underlying duality structure $\Delta=(\cS,D,\omega)$. Thus
$\bDelta=(\cS,D,\omega,\Lambda)$ is an integral duality structure of
type $\bt$. For every $x\in N$, the fiber
$(\cS_x,J_x,\omega_x,\Lambda_x)$ is an integral tamed symplectic space
(see Appendix \ref{app:integral}), which defines a polarized Abelian
variety $X_h(\cS_x,J_x,\omega_x,\Lambda_x)$ of type $\bt$. 
The underlying complex and symplectic tori of this Abelian variety are given by
$X_c(\cS_x,J_x,\Lambda_x)$ and $X_s(\cS_x,\omega_x,\Lambda_x)$, where $X_c$ is defined 
in Appendix \ref{app:integral}. When
$x$ varies in $N$, these polarized Abelian varieties fit into a smooth
fiber bundle $\cX_h(\bXi)$, whose fiber at $x\in N$ is given by:
\be
\cX_h({\bXi})_x=X_h(\cS_x,J_x,\omega_x,\Lambda_x)~~.  
\ee
Forgetting the complex structure on the fibers of $\cX_h(\bXi)$
produces the bundle of symplectic tori $\cX_s(\bDelta)$. As above, the
connection $D$ induces a complete integrable Ehresmann connection
$\cH_\bXi\eqdef \cH_\bDelta$ on $\bcX_h(\bDelta)$. The Ehresmann
transport of $\cH_\bXi$ is through isomorphisms of {\em symplectic}
tori, so it preserves the group structure and symplectic structure of
the fibers but it need not preserve their complex structure. In
particular, the holonomy group of $\cH_\bXi$ is contained in
$\Sp_\bt(2n,\Z)$.

\begin{Definition}
The fiber bundle with connection $(\cX_h({\bXi}),\cH_\bXi)$ is called the {\em bundle of
  polarized Abelian varieties} defined by the integral electromagnetic
structure $\bXi$.
\end{Definition}

\subsection{Integral electromagnetic fields}

Let $(M,g)$ be a Lorentzian four-manifold and $(\cM,\cG)$ be a
Riemannian manifold. Let $\varphi\in \cC^\infty(M,\cM)$ be a 
smooth map from $M$ to $\cM$. Let $\bXi=(\Xi,\Lambda)$ be an integral
electromagnetic structure defined on $\cM$, with underlying
electromagnetic structure $\Xi=(\cS,D,J,\omega)$ and underlying
duality structure $\Delta=(\cS,D,\omega)$.  Then the system
$\bXi^\varphi=(\Xi^\varphi,\Lambda^\varphi)$ is an electromagnetic
structure on $M$, where $\Lambda^\varphi$ is the $\varphi$-pullback of
the fiber sub-bundle $\Lambda\subset \cS$. We have
$\Xi^\varphi=(\Delta^\varphi,J^\varphi)$, with
$\Delta^\varphi=(\cS^\varphi,D^\varphi,\omega^\varphi)$. Let
$\bDelta^\varphi\eqdef (\Delta^\varphi,\Lambda^\varphi)$ denote the
underlying integral duality structure.

The twisted de Rham theorem implies that the twisted de Rham cohomology
space $H^\bullet_{\dd_D^{\varphi}}(M,\cS^\varphi)$ is naturally
isomorphic (as a graded vector space) with the singular cohomology
space $H^\bullet(M,\cS^\varphi)$ with coefficients in the local system
defined by $(\cS^\varphi,D^\varphi)$ (which in turn can be viewed as
the sheaf cohomology of the sheaf $\fS^\varphi$ of flat local sections
of $(\cS^\varphi,D^\varphi)$). We shall identify these cohomology
spaces in what follows and denote them by $H(N,\Delta^\varphi)$: 
\be
H(M,\Delta^\varphi)\eqdef H^\bullet_{\dd_{D^\varphi}}(M,\cS^\varphi)\equiv H^\bullet(M,\cS^\varphi)\equiv H^\bullet(M,\fS^\varphi)~~.
\ee
Let $H^\bullet(M,\bDelta^\varphi)$ be the twisted singular cohomology
of $M$ with coefficients in the $\Vect_0^\sp$-valued local system
$T_{\bDelta^\varphi}$ determined by the integral duality structure
$\bDelta^\varphi$.  Since $\cS^\varphi=\Lambda^\varphi\otimes_\Z\R$,
the coefficient sequence gives a map
$j_\ast:H^\bullet(M,\bDelta^\varphi)\rightarrow
H^\bullet(M,\Delta^\varphi)$, whose image is a lattice
$H^\bullet_{\Lambda^\varphi}(M,\Delta^\varphi)\subset
H^\bullet(M,\Delta^\varphi)$. We have
$H_{\Lambda^\varphi}^\bullet(M,\Delta^\varphi)\simeq
H^\bullet(M,\bDelta^\varphi)_\tf\eqdef
H^\bullet(M,\bDelta^\varphi)/\Tors(H^\bullet(M,\bDelta^\varphi))$ and
$H^\bullet(M,\Delta^\varphi)\simeq
H^\bullet(M,\bDelta^\varphi)\otimes_\Z \R$.

\begin{Definition}
An electromagnetic field $\cV\in \Omega^2(M,\cS^\varphi)$ is called
{\em $\Lambda^\varphi$-integral} if its $D^\varphi$-twisted cohomology
class $[\cV]\in H^2(M,\Delta^\varphi)$ belongs to
$H_{\Lambda^\varphi}^2(M,\Delta^\varphi)$:
\ben
\label{integrality}
[\cV]\in H_{\Lambda^\varphi}^2(M,\Delta^\varphi)=j_\ast(H^2(M,\bDelta^\varphi))~~.
\een
\end{Definition}

\

\noindent The condition that $\cV$ be integral is called the {\em
  (twisted) Dirac quantization condition} defined by the integral
electromagnetic structure $\bXi$.

\subsection{Integral duality transformations and integral scalar-electromagnetic symmetries}

\begin{Definition}
An {\em integral scalar-duality structure} is a pair $\bD_0\eqdef
(\bcD_0,\Lambda)$, where $\bD=(\Sigma,\Delta)$ is a scalar-duality
structure and $\Lambda$ is a Dirac system for $\Delta$.
\end{Definition}

\begin{Definition}
An {\em integral scalar-electromagnetic structure} is a pair
$\bcD=(\cD,\Lambda)$, where $\cD=(\Sigma,\Xi)$ is a
scalar-electromagnetic structure and $\Lambda$ is a Dirac system
for the underlying duality structure $\Delta$ of $\Xi$.
\end{Definition}

\

\noindent Let $\bcD=(\cD,\Lambda)$ be an integral
scalar-electromagnetic structure with
$\cD=(\cM,\cG,\Phi,\cS,D,J,\omega)$. Let $\Sigma=(\cM,\cG,\Phi)$,
$\Delta=(\cS,D,\omega)$ and $\Xi\eqdef (\cS,D,J,\omega)$ be the
underlying scalar, duality and electromagnetic structures of
$\cD$. Let $\bDelta=(\Delta, \Lambda)$ and $\bXi=(\Xi,\Lambda)$ be the
underlying integral duality structure and integral electromagnetic
structure. Let $\cD_0=(\Sigma, \Delta)$ be the underlying
scalar-duality structure and $\bD_0=(\cD_0,\Lambda)$ be the underlying
integral scalar-duality structure.

\begin{Definition}
The {\em integral duality group} defined by the integral
scalar-duality structure $\bcD_0$ is the following subgroup of the
duality group $\Aut(\cD_0)=\Aut^\ub_\Sigma(\Delta)$:
\be
\Aut(\bcD_0)\eqdef \{f\in \Aut(\cD_0)|f(\Lambda)=\Lambda\}\subset \Aut(\cD_0)~~.
\ee
Elements of this group are called {\em integral duality
  transformations}.
\end{Definition}

\begin{Definition}
The {\em integral scalar-electromagnetic group} defined by the
integral scalar-electromagnetic structure $\bcD$ is the following
subgroup of the scalar-electromagnetic group
$\Aut(\cD)=\Aut^\ub_\Sigma(\Xi)$:
\be
\Aut(\bcD)\eqdef \{f\in \Aut(\cD)|f(\Lambda)=\Lambda \}\subset \Aut(\cD)~~.
\ee
Elements of this group are called {\em integral scalar-electromagnetic
  symmetries}.
\end{Definition}

\noindent Notice that $\Aut(\bcD)$ is a subgroup of $\Aut(\bcD_0)$. 

\begin{Remark}
Consider the case when the duality structure $\Delta=(\cS,D,\omega)$
is trivial. Then the parallel transport of $D$ identifies the data
encoded by the the integral duality structure
$\bDelta=(\cS,D,\omega,\Lambda)$ with an integral symplectic space
$(\cS_0,\omega_0,\Lambda_0)$, which we can take to be the fiber of
$\bDelta$ at some fixed point of $\cM$. In this case, it is easy to see that:
\be
\Aut(\bcD)\simeq \Aut(\Sigma)\times \Aut(\cS_0,\omega_0,\Lambda_0)\simeq \Aut(\Sigma)\times \Sp_\bt(2n,\Z)~~,
\ee
where $\bt$ is the type of $(\cS_0,\omega_0,\Lambda_0)$. This recovers
statements familiar from the supergravity literature, which apply only
to the case of trivial duality structures. Notice that most of the 
supergravity literature considers only the case of principal Dirac lattices. 
\end{Remark}

\subsection{Remarks on the space of semiclassical fields}
\label{sec:semiclassical}

Since $H_{\Lambda^\varphi}^2(M,\cS^\varphi)$ does not capture the
torsion part of the group $H^2(M,\Lambda^\varphi)$, the integrality
condition \eqref{integrality} is weaker than the condition which is
expected to constrain the semiclassical fields of the quantized
Abelian gauge theory in the background metric $g$ and background
scalar field $\varphi$ (when an appropriate quantization can be
expected, such as when the space-time $(M,g)$ is globally
hyperbolic). In this case, the natural model for semiclassical fields
is provided by a twisted version of differential cohomology. Since we
plan to discuss this in a separate paper, we only make some brief
remarks about this aspect.

Consider an integral symplectic space $(\cS_0,\omega_0, \Lambda_0)$
defined on $N$.  One can define a version of differential cohomology
valued in such objects, which satisfies a variant of the axioms of
\cite{Hopkins, SimonsSullivan}; an explicit construction can be given
for example using Cheeger-Simons characters valued in the (affine)
symplectic torus $\cS_0/\Lambda_0$ defined by
$(\cS_0,\omega_0,\Lambda_0)$. This can be promoted to a twisted
version $\check{H}^k(N,T)$, whereby the coefficient object is replaced
by a local system $T:\Pi_1(N)\rightarrow \Symp_0^\times$. Given an
integral electromagnetic structure $\bXi$ whose underlying integral
duality structure $\bDelta$ corresponds to the local system $T$, the
correct model for the set of semiclassical Abelian gauge fields is
provided by a certain subset of $\check{H}^2(N,T)$; the curvature
$\cV=\curv(\alpha)$ of any element $\alpha$ of this subset is a
polarized closed 2-form which satisfies the integrality condition
\eqref{integrality}. When $(M,g)$ is globally hyperbolic, the initial
value problem for the twisted electromagnetic field is well-posed by
the results of \cite{Bar} (because the twisted d'Alembert operator is
normally hyperbolic) and restriction to a Cauchy hypersurface allows
one to describe explicitly the space of semiclassical fields. In that
case, one can use a version of the method of \cite{Szabo} to quantize
the electromagnetic theory in a manner which reproduces this model for
the space of semiclassical fields. It turns out that elements
$\alpha\in \check{H^2}(N,T)$ classify affine symplectic
$T^{2n}$-bundles\footnote{Non-principal fiber bundles with fiber a
  symplectic $2n$-torus and whose structure group reduces to the
  affine symplectic group of such a torus.} with connection. This
allows one to represent $\cV$ as the curvature of a connection on such
a bundle, which gives the geometric interpretation of semiclassical
Abelian gauge fields.

\section{Global solutions and classical locally geometric U-folds}
\label{sec:GeometricU}

In this section we show that, when the duality structure is
non-trivial, then general solutions of twisted ESM theories can be
interpreted as locally geometric U-folds. We also discuss the case
when $(\cS,\omega)$ is symplectically trivial but the scalar-manifold
is not simply-connected and $D$ is a non-trivial flat connection
(i.e. $D$ has non-trivial holonomy). 

\subsection{U-fold interpretation of global solutions}

Let $M$ be a four-manifold and $\cD=(\Sigma,\Xi)$ be a
scalar-electromagnetic structure with underlying scalar structure
$\Sigma = (\cM,\cG,\Phi)$ and underlying electromagnetic structure
$\Xi = (\cS,D,J,\omega)$. Let $(g, \varphi, \cV)\in \Sol_{\cD}(M)$ be
a global solution of the ESM equations defined by $\cD$ on $M$. Let
$(\cU_\alpha)_{\alpha\in I}$ be a good open cover of $\cM$.  Since
each $\cU_\alpha$ is contractible, we can find a flat a symplectic
local frame $\cE_\alpha=(e_1^{(\alpha)}\ldots e_n^{(\alpha)},
f_1^{(\alpha)}\ldots f_n^{(\alpha)})=(u_1^{(\alpha)}\ldots
u_{2n}^{(\alpha)})$ of $\Delta=(\cS,D,\omega)$ defined over
$\cU_\alpha$, which trivializes $\Delta$ over $\cU_\alpha$.  Let
$\cU_{\alpha\beta}\eqdef \cU_\alpha\cap \cU_\beta$. When
$\cU_{\alpha\beta}\neq \emptyset$, we have
$u^{(\beta)}_M=_{\cU_{\alpha\beta}}(f_{\alpha\beta})^N_{~M}
u^{(\alpha)}_N$ for some uniquely-defined matrix $f_{\alpha\beta}\in
\Sp(2n,\R)$, where we use the Einstein summation convention. When
$\cU_{\alpha\beta}=\emptyset$, we set $f_{\alpha\beta}=I_{2n}$.  Then
$(f_{\alpha\beta})_{\alpha,\beta\in I}$ satisfy the cocycle
conditions:
\ben
\label{fcocycle}
f_{\alpha\beta}f_{\beta\delta}=_{\cU_{\alpha\beta\delta}}f_{\alpha\delta}~~,
\een
where $\cU_{\alpha\beta\delta}\eqdef \cU_{\alpha}\cap \cU_{\beta}\cap
\cU_\delta$.  Let $U_\alpha\eqdef \varphi^{-1}(\cU_\alpha)$. Then
$(U_\alpha)_{\alpha\in I}$ is an open cover of $M$ and
$\cE_\alpha^\varphi\eqdef ((e_1^{(\alpha)})^\varphi,\ldots,
(e_n^{(\alpha)})^\varphi, (f_1^{(\alpha)})^\varphi,\ldots,
(f_n^{(\alpha)})^\varphi)=((u_1^{(\alpha)})^\varphi,\ldots,(u_{2n}^{(\alpha)})^\varphi)$
are flat symplectic local frames of $(\cS^\varphi,
D^\varphi,\omega^\varphi)$ defined over $U_\alpha$. Define:
\begin{equation}
(g_{\alpha}, \varphi_{\alpha}, \cV_{\alpha}) \eqdef (g|_{U_\alpha}, \varphi|_{U_\alpha}, \cV|_{U_{\alpha}}) \in \Sol_{\cD}(U_{\alpha})~~.
\end{equation}
When $\cU_{\alpha\beta}\neq \emptyset$, we have $U_{\alpha\beta}\eqdef
U_{\alpha}\cap U_\beta=\varphi^{-1}(\cU_{\alpha\beta})\neq \emptyset$
and:
\ben
\label{ualphabeta}
(u_M^{(\beta)})^\varphi=_{U_{\alpha\beta}}(f_{\alpha\beta})^N_{~M} (u^{(\alpha)}_N)^\varphi
\een
as well as:
\ben
\label{eq:ggluing}
g_{\alpha}|_{U_{\alpha\beta}} = g_{\beta}|_{U_{\alpha\beta}}\, , \qquad \varphi_{\alpha}|_{U_{\alpha\beta}} = \varphi_{\beta}|_{U_{\alpha\beta}}~~,~~
\cV_{\alpha}|_{U_{\alpha\beta}}=\cV_{\beta}|_{U_{\alpha\beta}}~~.
\end{equation}
Expand:
\begin{equation}
\cV_{\alpha} = \cV^M_{\alpha}\otimes (u^{(\alpha)}_M)^\varphi=F^i_\alpha \otimes (e_i^{(\alpha)})^\varphi + G^i_\alpha \otimes (f_i^{(\alpha)})^\varphi~~,~~\mathrm{with}~~\cV^M_{\alpha}\in\Omega^2(U_{\alpha})~~,
\end{equation}
where we use Einstein summation over $M=1\ldots 2n$ and $i=1\ldots n$
and we set $F^i_\alpha\eqdef \cV^{i}_\alpha$ and $G^i_\alpha\eqdef
\cV^{i+n}_\alpha$ for all $i=1\ldots n$. Combining this with
\eqref{ualphabeta} gives:
\be
\cV^M_\alpha=(f_{\alpha\beta})^M_{~N}\cV^N_\beta~~.
\ee
Let ${\hat \cV}_\alpha \eqdef \left[\begin{array}{c} \cV^1_\alpha
    \\\ldots\\ \cV^{2n}_\alpha\end{array}\right]\in
\Omega^2(U_\alpha,\R^{2n})$. The relation above reads:
\ben
\label{cVgluing}
{\hat \cV}_\alpha=_{U_{\alpha\beta}}f_{\alpha\beta}{\hat \cV}_\beta~~,
\een
showing that when going from patch to patch one must perform a
locally-defined duality transformation $f_{\alpha\beta}\in
\Sp(2n,\R)$. Thus global solutions are ``glued'' from local solutions
using duality transformations.

Conversely, given a good open cover $(\cU_{\alpha})_{\alpha\in I}$ of
$M$ and elements $f_{\alpha\beta}\in \Sp(2n,\R)$ satisfying the
cocycle condition \eqref{fcocycle}, any family of local solutions
$(g_\alpha,\varphi_\alpha,\cV_\alpha)\in \Sol_\cD(U_\alpha)$ (where
$U_\alpha\eqdef \varphi^{-1}(\cU_\alpha)$) which satisfies conditions
\eqref{eq:ggluing} corresponds to a global scalar-electromagnetic
solution $(g,\varphi,\cV)$ of the ESM theory of type $\cD$. When the
ESM theory is the bosonic sector of a four-dimensional supergravity,
the elements $f_{\alpha\beta}\in \Sp(2n,\R)$ correspond to
locally-defined U-duality transformations. These observations justify
the following:

\begin{Definition}
A {\em locally geometric classical U-fold of type $\cD$} defined
on $M$ is a global solution $(g,\varphi,\cV)\in \Sol_{\cD}(M)$ of the
equations of motion of the ESM theory defined by $\cD$ on $M$.
\end{Definition}

\subsection{The case when $(\cS,\omega)$ is symplectically trivial}
\label{sec:Strivial}

Recall that a symplectic vector bundle $(\cS,\omega)$ is called {\em
  symplectically trivial} if it admits a globally-defined symplectic
frame. Symplectic triviality of $(\cS,\omega)$ implies topological
triviality of $\cS$, but the converse is {\em not} true. If $P_0$
denotes the principal $\GL(2n,\R)$-bundle of frames of $\cS$ and $P$
denotes the principal $\Sp(2n,\R)$-bundle of symplectic frames of
$(\cS,\omega)$, then $\cS$ is topologically trivial iff $P_0$ is
trivial while $(\cS,\omega)$ is symplectically trivial iff $P$ is
trivial. Notice that $P_0$ is associated to $P$ through the inclusion
morphism $\Sp(2n,\R)\hookrightarrow \GL(2n,\R)$. The symplectic frame
bundle $P$ reduces to a principal $\U(n)$-bundle upon choosing a
taming on $(\cS,\omega)$, two-different choices giving isomorphic
reductions. This also follows from the fact that the groups
$\Sp(2n,\R)$ and $\U(n)$ are homotopy equivalent, the latter being a
maximal compact form of the former. The Chern-classes
$c_k(\cS,\omega)$ of a symplectic vector bundle $(\cS,\omega)$ are
defined as the Chern classes of the unitary vector bundle
$(\cS,\omega,J)$, where $J$ is any taming on $(\cS,\omega)$. A
necessary (but generally not sufficient) condition that
$(\cS,\momega)$ be symplectically trivial is that $c_k(\cS,\omega)=0$
for all $k=1\ldots \dim \cM$.

When $\pi_1(\cM)\neq 1$ and $(\cS,\omega)$ is symplectically trivial,
there exist flat symplectic connections $D$ on $(\cS,\omega)$ which
have non-trivial holonomy. Hence a duality structure
$(\cS,D,\omega)$ defined on a non-simply connected scalar manifold
$\cM$ can be nontrivial even when $(\cS,\omega)$ is symplectically
trivial.

Let $\Xi=(\cS,D,J,\omega)$ be an electromagnetic structure defined on
$\cM$ such that $(\cS,\omega)$ is symplectically trivial and let
$\cE=(e_1\ldots e_n, f_1\ldots f_n)=(u_1\ldots u_{2n})$ be a global
symplectic frame of $(\cS,\omega)$. Let ${\hat
  A}=(A^M_{~N})_{M,N=1\ldots 2n}\in \Omega^1(\cM,\Mat(2n,\R))$ be the
matrix-valued connection 1-form of $D$ with respect to $\cE$, whose
entries are defined through:
\be
Du_M=A^N_{~M}u_N~~.
\ee
The condition that $A$ is flat reads: 
\ben
\label{hatAflat}
\dd {\hat A}+{\hat A}\wedge {\hat A}=0~~.
\een
The matrix ${\hat J}=(J^M_{~N})_{M,N=1\ldots 2n}$ of the taming $J$
with respect to $\cE$ has entries defined through:
\be
Ju_M=J^N_{~M}u_N~~.
\ee
Consider a solution $(g,\varphi,\cV)\in \Sol_\cD(M)$ of the ESM
equations. Then the pulled-back symplectic vector bundle
$(\cS^\varphi, \omega^\varphi)$ is symplectically trivial and $\cE$
induces a global symplectic frame $\cE^\varphi=(e_1^\varphi,\ldots,
e_n^\varphi, f_1^\varphi,\ldots, f_n^\varphi)=(u_1^\varphi,\ldots,
u_{2n}^\varphi)$ of $(\cS^\varphi,\omega^\varphi)$. The matrix-valued
connection 1-form ${\hat A}^\varphi\in \Omega^1(M,\Mat(2n,\R))$ of the
pulled-back connection $D^\varphi$ with respect to $\cE^\varphi$ has
entries:
\be
(A^\varphi)^{M}_{~N}=\varphi^\ast(A^M_{~N})\in \Omega^1(M)~~,
\ee
while the matrix ${\hat J}^\varphi$ of the pulled-back taming
$J^\varphi$ with respect to this frame has entries:
\be
(J^\varphi)^M_{~N}=J^{M}_{~N}\circ \varphi~~.
\ee
Writing:
\begin{equation}
\cV = \cV^{M} \otimes u^{\varphi}_{M}~~\mathrm{with}~~\cV^{M}\in \Omega^2(M)
\end{equation} 
shows that $\cV$ is globally equivalent with the vector-valued 2-form
${\hat \cV}\eqdef \left[\begin{array}{c}\cV^1\\\ldots
    \\ \cV^{2n} \end{array} \right]\in \Omega^2(M,\R^{2n})$. The
polarization condition $\ast \cV=-J^\varphi\cV$ amounts to:
\be
\ast {\hat \cV}=-{\hat J}^\varphi {\hat \cV}
\ee
while the electromagnetic equation takes the form:
\be
\dd {\hat \cV}+{\hat A}\wedge {\hat \cV}=0~~.
\ee
We distinguish the cases:
\begin{enumerate}
\itemsep 0.0em
\item When the duality structure $(\cS,D,\omega)$ is trivial, we can
  find a globally-defined {\em flat} symplectic frame of
  $(\cS,\omega)$, i.e. a global symplectic frame $\cE=(u_1\ldots
  u_{2n})$ such that $D(u_M)=0$ for all $M=1\ldots 2n$. In this case,
  we have ${\hat A}=0$ and the electromagnetic equations reduce to the
  condition that the vector-valued globally-defined one-form ${\hat
    \cV}$ is closed.
\item When the duality structure $(\cS,D,\omega)$ is not trivial,
  ${\hat A}$ is nonzero with respect to any globally-defined
  symplectic frame $\cE$ and hence ${\hat A}^\varphi$ will be nonzero
  if $\varphi$ is generic enough. Thus ${\hat \cV}$ cannot, in
  general, be globally identified with a {\em closed} vector-valued
  2-form defined on $M$.
\end{enumerate}

\paragraph{Relation to the construction of
\cite{AndrianopoliUduality, AndrianopoliFlat}.}

References \cite{AndrianopoliUduality, AndrianopoliFlat} consider the
particular setting of the bosonic sector of {\em extended} ($\cN\geq
2$) supergravity theories, taking the scalar manifold $(\cM,\cG)$ to
be a {\em simply-connected} symmetric space of non-compact type (which
therefore is contractible \footnote{As explained in
  \cite{GeometricUfolds}, the symmetric spaces traditionally
  considered in the extended supergravity literature are contractible,
  being diffeomorphic with $\R^N$ for some positive integer $N>0$.})
of the form $\cM=G/H$. In this case, we have $\Iso(\cM,\cG)=G$. The
assumption that $\cM$ is contractible insures that the principal
$H$-bundle $P_H$ defined over $\cM$ by the natural projection
$\pi:G\rightarrow \cM=G/H$ is trivial, and hence there exists a global
smooth choice of coset representatives, i.e a smooth global section of
$P_H$ (such a section is denoted $L$ in
\cite[eq. (3.2)]{AndrianopoliUduality} and
\cite[eq. (3.2)]{AndrianopoliFlat}).  This assumption also allows
loc. cit. to use globally-defined coordinates and thus to ignore all
topological aspects. For that particular case, any duality structure
is necessarily trivial since $\pi_1(\cM)=1$. This means that the flat
symplectic connection $D$ has trivial holonomy in the situation of
loc. cit., being gauge-equivalent to the trivial connection. Hence
even the very special setting discussed in the present subsection is
more general than that considered in loc. cit. This is because a
symplectically trivial duality structure on a scalar manifold $\cM$
which is {\em not} simply connected need not be trivial, since its
flat connection $D$ may have non-trivial holonomy. When
$\pi_1(\cM)\neq 1$, there generally exists an infinity of
gauge-equivalence classes of non-trivial flat connections, to which
the construction of \cite{AndrianopoliUduality, AndrianopoliFlat}
cannot apply globally unless it is first modified using the results of the
present paper. 

\paragraph{Conditions for symplectic triviality of $(\cS^\varphi,\omega^\varphi)$.}
The pulled-back bundle $(\cS^\varphi,\omega^\varphi)$ may be
symplectically trivial even when $(\cS,\omega)$ is symplectically
non-trivial. In this case, the twisted Abelian gauge theory defined on
$(M,g)$ by the pulled-back electromagnetic structure
$\Xi^\varphi=(\cS^\varphi,D^\varphi,J^\varphi,\omega^\varphi)$ can be
described as a theory of vector-valued 2-forms ${\hat
  \cV}=\left[\begin{array}{c}\cV^1\\\ldots \\ \cV^{2n} \end{array}
  \right]\in \Omega^2(M,\cS^\varphi)$, where $\cV^M$ are the
coefficients of $\cV$ in a globally-defined symplectic frame
$\bu_1\ldots \bu_{2n}$ of $\cS^\varphi$ (which need not be the
pull-back of a symplectic frame of $(\cS,\omega)$):
\be
\cV=\cV^M\otimes \bu_M~~.
\ee
Most of the formulas above apply, except that $A^\varphi$ and
$J^\varphi$ are no longer the pull-back of objects defined on
$\cM$. Every symplectic vector bundle over a (not-necessarily compact)
four-manifold is associated to a $\U(2)$ principal bundle
\cite{Husemoller}. Thus $(\cS^{\varphi},\omega^{\varphi})$ is
symplectically trivial iff its first and second Chern classes vanish.
Since $c_k(\cS^\varphi,\omega^\varphi)=\varphi^\ast(c_k(\cS,\omega))$,
this amounts to the conditions:
\ben
\label{U2cond}
\varphi^\ast(c_{1}(\cS,\omega))=0~~\mathrm{and}~~\varphi^\ast(c_{2}(\cS,\omega))=0~~.
\een
Let us give some sufficient conditions for symplectic triviality of
$(\cS^\varphi,\omega^\varphi)$.  Recall that an electromagnetic
structure $\Xi = (S,D,J, \omega)$ defined on $M$ is unitary iff the
parallel transport of $D$ preserves $J$, which means that
$(\cS,D,J,\omega)$ is a flat Hermitian vector bundle. Accordingly, the
holonomy bundle of $(\cS,D,\omega)$ is a flat principal
$\U(n)$-bundle.  

\begin{Definition}
A unitary electromagnetic structure $\Xi = (S,D,J, \omega)$ defined on
$\cM$ is called \emph{reduced} if the structure group of its holonomy
bundle reduces to $\U(2)$.
\end{Definition}

\begin{Proposition}
Assume that $H^\ast(M,\Z)$ has no torsion. Then the $\varphi$-pull-back of
a reduced unitary electromagnetic structure defined over $\cM$ is
symplectically trivial.
\end{Proposition}

\begin{proof}
Let $\Xi=(\cS,D,J,\omega)$ be a reduced unitary electromagnetic
structure defined on $\cM$. Then
$\Xi^\varphi=(\cS^\varphi,D^\varphi,J^\varphi,\omega^\varphi)$ is
associated to a flat principal $\U(2)$-bundle defined over $M$. By
Chern-Weil theory, this implies that the first two Chern classes of
$(\cS^\varphi,\omega^\varphi)$ are pure torsion, while the higher
Chern classes vanish. Since $H^\ast(M,\Z)$ has no torsion, it follows 
that conditions \eqref{U2cond} are satisfied. 
\end{proof}

\begin{Proposition}
Assume that $M$ is non-compact and that
$(\cS^{\varphi},\omega^{\varphi})$ admits a real Lagrangian
subbundle. Then $(\cS^{\varphi},\omega^{\varphi})$ is symplectically
trivial.
\end{Proposition}

\begin{proof}
Since $(\cS^{\varphi},\omega^{\varphi})$ admits a Lagrangian
subbundle, by \cite[Proposition 3.1.6.]{Vaisman} the symplectic frame
bundle of $(\cS^\varphi,\omega^\varphi)$ admits a reduction to a
$\O(n)$-principal bundle. By \cite[Proposition 4.4.6.]{Vaisman}, the
odd-dimensional Chern-classes of $(\cS^\varphi,\omega^\varphi)$ are
zero and thus $c_1(\cS^\varphi,\omega^\varphi)=0$.  On the other hand,
$c_2(\cS^\varphi,\omega^\varphi)$ vanishes since $H^4(M,\Z)=0$ by
non-compactness of $M$.
\end{proof}

\section{Relation to the literature}
\label{sec:relation}

The fact that flat symplectic vector bundles arise in four-dimensional
{\em extended} supergravity theories (theories with $\cN\geq 2$ local
supersymmetry) is known since the work of \cite{Strominger, Freed,
  Cortes} on (projective) special K\"ahler geometry. In the
situation of references \cite{Strominger, Freed,Cortes}, one deals
with an $\cN=2$ supergravity theory coupled to scalars and to Abelian
gauge fields and the duality structure $\Delta=(\cS,D,\omega)$ relevant to 
loc. cit. is highly constrained by other data due to supersymmetry.

In contrast to the above, the theory considered in the present paper
is a {\em purely bosonic} sigma model coupled to gravity and to
Abelian gauge fields and the duality structure $\Delta$ is arbitrarily
chosen and introduced without any need to couple the model to fermions
and without any requirement of supersymmetry. Such a generalization of
the ordinary Einstein-Scalar-Maxwell theory is allowed because Abelian
gauge theories defined on oriented Lorentzian four-manifolds can be
``twisted'' by a flat symplectic vector bundle --- a fact which has
nothing to do with fermions or supersymmetry. A generally non-trivial
duality structure structure {\em must} in fact be present in the
purely bosonic theory, if one wishes to view $\cN=2$ supergravity as a
particular case of $\cN=0$ theories. The models considered in
this paper provide a wide generalization of
the traditional ESM theory (which corresponds to the case of trivial
duality structures). To our knowledge, this extension of
Einstein-Scalar-Maxwell theory was not studied systematically before.
As we showed in the present paper, our generalized models possess some
of the features and structures which were noticed previously in models
with $\cN=2$ supersymmetry. In particular, the Dirac quantization
condition produces a smooth bundle of polarized Abelian varieties
defined over the scalar manifold, a feature which does not require any
coupling to fermions and hence is unrelated to supersymmetry.

The crucial object allowing for a general and frame-free formulation
of our models is the taming $J$ of $(\cS,\omega)$, which encodes the
gauge couplings and theta angles of the Abelian gauge theory in a
manner that makes no mention of any choice of local flat symplectic
frame. Picking such a local frame $\cE$ of $\cS$ (supported on a
simply-connected open subset $\cU$ of $\cM$) allows one to {\em
  locally} encode the taming using a function $\tau^\cE:\cU\rightarrow
\SH_n$ valued in the Siegel upper half space, whose entries
$\tau^\cE_{ij}$ are the gauge-kinetic functions appearing in
references such as \cite{Andrianopoli, AndrianopoliUduality,
  AndrianopoliFlat} (which discuss only the case of extended
supergravity theories with trivial duality structure). Unlike the
taming $J$, the gauge-kinetic functions depend on the local frame
$\cE$, so they are both frame-dependent and only locally defined. The
reliance of the supergravity literature on the gauge-kinetic functions
forces the traditional formulation to be frame-dependent and thus not
truly geometric. That formulation cannot work globally unless the
duality structure $(\cS,D,\omega)$ is trivial, since a
globally-defined flat symplectic frame only exists in that case. The
relation with the construction of \cite{AndrianopoliUduality,
  AndrianopoliFlat} is discussed briefly in Subsection
\ref{sec:Strivial}.

Two Lie groups play an important role in the general theory, namely
the holonomy group $\Hol_D\subset \Sp(2n,\R)$ of the flat symplectic
connection $D$ and the group of duality transformations
$\Aut(\cD_0)$. The nature of these groups depends markedly on whether
or not the duality structure is trivial, showing one difference
between the generalized theories constructed in this paper and those
traditionally considered in the supergravity literature:
\begin{enumerate}[(a)]
\itemsep 0.0em
\item In the traditional case of a trivial duality structure
  $\Delta_0$ of rank $2n$, we have $\Hol_D=1$ and the classical
  duality group is $\Aut(\cD_0)\simeq \Aut(\Sigma)\times
  \Aut(\Delta_0)$.  We have $\Aut(\Delta_0)\simeq
  \Aut(\cS_0,\omega_0)\simeq \Sp(2n,\R)$, where $(\cS_0,\omega_0)$ is
  the typical fiber of $(\cS,\omega)$. When passing to the
  semiclassical theory, the Dirac system $\Lambda$ can be
  identified with a lattice $\Lambda_0\subset \cS_0$, which is the
  Dirac lattice of the traditional model. Accordingly, the integral
  duality structure $\bDelta=(\Delta_0,\Lambda)$ of the semiclassical
  theory can be identified with the integral symplectic vector space
  $(\cS_0,\omega_0,\Lambda_0)$, whose type we denote by $\bt$. The
  semiclassical duality group is the group of integral duality
  transformations $\Aut(\bcD_0)\simeq \Aut(\Sigma)\times
  \Aut(\bDelta)$, where $\Aut(\bDelta)\simeq
  \Aut(\cS_0,\omega_0,\Lambda_0)\simeq \Sp_\bt(2n,\Z)$ is the group
  consisting of those symplectomorphisms of $(\cS_0,\omega_0)$ which
  stabilize the Dirac lattice $\Lambda_0$. In the principal case
  $\bt=\bdelta(n)$, this subgroup is isomorphic with $\Sp(2n,\Z)$,
  while in general it is a modified Siegel modular group.
\item The situation is much more complicated when the duality
  structure $\Delta$ is non-trivial. In this case, we have $\Hol_D\neq
  1$ and the classical duality group $\Aut(\cD_0)$ {\em does not}
  decompose as a product of $\Aut(\Sigma)$ and
  $\Aut(\Delta)$. Moreover, $\Aut(\Delta)$ need {\em not} be a full
  symplectic group. The structure of the groups $\Aut(\cD_0)$ and
  $\Aut(\cD)$ depends on the holonomy representation
  $\hol_D:\pi_1(\cM)\rightarrow \Sp(2n,\R)$. In the semiclassical
  theory, the holonomy group $\Hol_D$ for an integral duality
  structure $\bDelta$ of type $\bt$ must be a subgroup of
  $\Sp_\bt(2n,\Z)$. On the other hand, the group $\Aut(\bcD_0)$ of
  semiclassical duality transformations can be quite different from
  $\Aut(\Sigma)\times \Sp_\bt(2n,\Z)$.
\end{enumerate}
It should be clear from these remarks that our models behave
quite differently from those considered traditionally in the
supergravity literature. These theories may seem exotic even in the
simple case when $(\cS,\omega)$ is symplectically trivial but the flat
symplectic connection $D$ has non-trivial holonomy, as illustrated in
Subsection \ref{sec:Strivial}.

\section{A simple example}

In this section, we illustrate the salient features of electromagnetic theories 'twisted' 
by a nontrivial duality structures in a very simple example. 

Consider the space-time $M=\R^3\times S^1$ and let $x=(x^1,x^2,x^3)$ be the Cartesian coordinates on $\R^3$ and $\theta$ be an 
angular coordinate on $\rS^1$. Also consider the scalar structure $\Sigma=(M,\cG,\Phi)$, where $M=\rS^1$ is the unit circle, 
$\cG$ is the metric on $\rS^1$ given by: 
\be
\dd s^2_\cG\eqdef \rho^2 \dd \psi^2
\ee
and $\Phi:\rS^1\rightarrow \R$ is a fixed smooth function. 
Up to isomorphisms, duality structures $\Delta=(\cS, D,\omega)$ of rank $2n$ defined on $\rS^1$ correspond to the 
symplectic character variety: 
\be
C_{\pi_1(\rS^1)}(\Sp(2n,\R))=\Hom(\Z,\Sp(2n,\R))/\Sp(2n,\R)\simeq \Sp(2n,\R)/_{\Ad} \Sp(2n,\R)~~,
\ee
which coincides with the set of conjugacy classes of the group $\Sp(2n,\R)$, where the element of $\Sp(2n,\R)$ 
associated to $\Delta$ is given by the holonomy of the flat symplectic connection $D$ around the circle. 

Let $g$ be a Lorentzian metric on $M$ which satisfies the vacuum Einstein equation. 
Then a Scalar-Maxwell theory on $(M,g)$ with scalar structure $\Sigma$ is
an Einstein-Scalar-Maxwell theory considered in the background
approximation for $(M,g)$, i.e. when back-reaction is neglected in the
Einstein equation. Such a theory is obtained by picking a taming $J$ of
$(\cS,\omega)$, and hence an electromagnetic structure
$\Xi=(\cS,D,\omega,J)$. This theory contains a single scalar field
described by a smooth map $\varphi:M\rightarrow \rS^1$ and an
electromagnetic field strength $\cV\in \Omega^2(M,\cS^\varphi)$, which
satisfies the positive polarization condition $\star_{g,J^\varphi}
\cV=\cV$ and the `twisted' Maxwell equation $\dd_{D^\varphi} \cV=0$.
In the background approximation for the scalar field $\varphi$
(i.e. when $\varphi$ is fixed and one neglects back-reaction in the
scalar equation), the only equation of motion is the electromagnetic
equation $\dd_{D^\varphi}\cV=0$, which gives a `twisted
electromagnetic theory' defined on $(M,g)$.

A particularly simple class of examples is obtained as follows.
Consider the Lorentzian metric\footnote{Notice that $(M,g)$ is a
solution of the {\em vacuum} Einstein equation.} $g\in \Met_{3,1}(M)$
whose squared line element is given by:
\be
\dd s_g^2\eqdef \dd s_0^2+R^2\dd \theta^2~~,
\ee
where $\dd s_0^2$ is the Minkowski metric on $\R^3$. Moreover, take
$\varphi$ to be independent of $x$, i.e. set $\varphi=\phi\circ
\pi_2$, where $\phi:\rS^1\rightarrow \rS^1$ is a smooth map and
$\pi_2:M\rightarrow \rS^1$ is the canonical projection of
$M=\R^3\times \rS^1$ onto the second factor. For any such smooth map, the 
pull-back connection is `independent of $x$' in the sense that it has the form: 
\be
D^\varphi=(D^\phi)^{\pi_2}~~,
\ee
where $D^\phi$ is the $\phi$-pullback of $D$ to a flat symplectic connection 
on the symplectic vector bundle $(\cS^\phi,\omega^\phi)$ defined over the circle. 
Moreover, the pull-back taming is `independent of $x$' in the sense that we have: 
\be
J^\varphi=(J^\phi)^\pi~~,
\ee
where $J^\phi$ is the $\phi$-pullback of $J$ to a taming of $(S^\phi,\omega^\phi)$.
It is clear that the holonomy of $D^\phi$ along $\rS^1$ (starting at a reference point 
$0\in \rS^1$) is given by: 
\be
\Hol_{D^\phi}(0)=\Hol_D(\phi(0))^{\iota(\phi)}~~,
\ee
where $\Hol_D(\phi(0))$ is the holonomy of $D$ starting at $\phi(0)$
and $\iota(\phi)$ is the index of the map $\phi$. In particular, the
flat connection $D^\phi$ is non-trivial (and hence the pull-back
duality structure $(\cS^\varphi, D^\varphi, \omega^\varphi)$ is
non-trivial iff that the duality structure $(\cS,D,\omega)$ is
non-trivial and $\phi$ has non-trivial index. As a simple example, one can take 
$\phi$ to be the $n$-fold isogeny $\varphi_m$ for some $m\in \Z$: 
\be
\phi_m(\theta)=m\theta~~\forall \theta\in \R/(2\pi \Z)~~.
\ee
Since $\iota(\phi_m)=m$, this gives: 
\be
\Hol_{D^\phi}(0)=\Hol_D(\phi(0))^m~~.
\ee 
When the duality structure $(\cS,D,\omega)$ is non-trivial and the map $\phi$ has non-zero index, 
the electromagnetic theory defined on $(M,g)$ by the modified Maxwell equations $\dd_{D^\varphi} \cV=0$ 
differs {\em globally} from ordinary $n$-field electromagnetism, even though it is locally indistinguishable 
from the latter.  

\section{Conclusions and further directions}
\label{sec:conclusions}

We gave a the global mathematical formulation, which in particular is
coordinate-free and frame-free, of four-dimensional bosonic sigma
models coupled to gravity and to Abelian gauge fields for the general
case when the duality structure of the Abelian gauge theory is
described by a non-trivial flat symplectic vector bundle
$(\cS,D,\omega)$ defined over the scalar manifold $\cM$. We showed
that such a model is naturally described using a taming $J$ of
$(\cS,\omega)$, which encodes the inverse gauge couplings and theta
angles of the Abelian gauge theory in a globally well-defined and
frame-free manner. We discussed the groups of duality transformations
and electro-magnetic symmetries of such models, which involve lifting
isometries of the scalar manifold to the bundle $\cS$.  We also gave
the mathematical formulation of Dirac quantization for such theories,
which involves certain discrete local systems and leads to a smooth
bundle of polarized Abelian varieties defined over $\cM$. We also
showed that global classical solutions of such models are classical
locally geometric U-folds.

Our generalized formulation of the Abelian gauge sector involves a
twist by the duality structure and treats that sector as a twisted
version of a theory of self-dual forms. As we will show in a different
publication, this allows one to approach the quantization of that
sector by adapting methods used in the theory of self-dual form fields
\cite{Witten,BM,Szabo, FMS}. Among other aspects, this involves a
version of differential cohomology twisted by local coefficients.

The coupled system consisting of equations \eqref{eins}, \eqref{sc}
and \eqref{em} can be approached from various points of view. For
example, it would be interesting to study the initial value problem
for this system. One could also use it to construct classical
U-fold solutions and to study their general properties. One advantage
of our approach it that it allows one to address various questions
regarding this class of U-folds using the well-developed tools of 
global and geometric analysis.

We comment briefly on the construction of solutions with non-trivial
duality structure. As pointed out in \cite{GeometricUfolds} (where
some simple examples were given), it is easy to construct examples of
non-simply connected scalar manifolds $(\cM,\cG)$ by considering
quotients $\cM_0/\Gamma$ of simply-connected scalar manifolds
$(\cM_0,\cG_0)$ through discrete subgroups $\Gamma$ of the isometry
group $\Iso(\cM_0,\cG_0)$. Non-trivial duality structures on such
quotients arise from symplectic representations of $\Gamma$. When
starting from a simply-connected scalar manifold $\cM_0$ endowed with
a scalar potential $\Phi_0\in \cC^\infty(\cM_0,\R)$, one can take
$\Gamma$ to be a subgroup of $\Aut(\cM_0,\cG_0,\Phi_0)$. This insures
that $\Phi_0$ descends to a scalar potential $\Phi$ on $\cM$. As
explained in loc. cit., the most interesting examples of U-fold
solutions are those for which the space-time also has non-trivial
fundamental group. Such space-times can also be produced through
quotient constructions, being generic among Lorentzian four-manifolds
which satisfy the usual causality conditions. When $(M,g)$ is
globally-hyperbolic with a Cauchy hypersurface $C$, we have
$\pi_1(M)=\pi_1(C)$, so it suffices to consider situations when
$C$ is not simply connected. It is clear from these observations
that classical U-folds of the type described in this paper are
abundant.

As we will show in separate paper, the models constructed
here can be extended further to a class of theories which afford a
general, globally-geometric description of classical locally-geometric
U-folds in four dimensions. The present paper is part of a larger program
aimed at understanding the global geometric structure of supergravity
theories. As pointed out in \cite{Lipschitz,lip}, this involves
clarifying not only the global structure of the bosonic
sector (for which the work presented here constitutes one step) but
also addressing a number of questions in spin geometry, which have
remained largely unexplored until recently.

\acknowledgments{The authors thank Tomas Ortin for discussions and
  correspondence.  The work of C. I. L. is supported by grant
  IBS-R003-S1. The work of C.S.S. is supported by the ERC Starting
  Grant 259133 – Observable String.}

\appendix

\section{Tamings and positive polarizations of symplectic vector spaces}
\label{app:space_tamings}

We recall some well-known facts regarding tamings of symplectic 
vector spaces and positive polarizations of their complexifications (see
\cite{Vaisman, Berndt}).

\subsection{Period matrices of complex vector spaces}
\label{app:periods}

Let $V$ be an $\R$-vector space of real dimension $2n$ and
$J\in \End_\R(V)$ be a complex structure on $V$ (thus
$J^2=-\id_S$). This makes $V$ into a vector space over $\C$ upon using the 
multiplication of scalars defined through: 
\be
(\alpha+\i \beta )\xi\eqdef \alpha\xi+\beta J\xi~~,~~\forall \alpha,\beta\in \R~~\forall \xi\in V~~.
\ee
Consider a basis $\cE=(e_1\ldots e_{2n})$ of $V$ over
$\R$ and a basis $\cW=(w_1\ldots w_n)$ of $(V,J)$ over $\C$; thus
$(w_1,\ldots w_n, J w_1,\ldots, Jw_n)$ is a basis of $V$ over
$\R$. 

\begin{Definition}
The {\em fundamental matrix} of $(V,J)$ with respect to $\cE$ and
$\cW$ is the matrix $\Pi:=\Pi_\cE^\cW \in \Mat(n,2n,\C)$ defined through:
\ben
\label{PiDef}
e_\alpha=\sum_{k=1}^n \Pi_{k,\alpha} w_k=\sum_{k=1}^n [(\Re \Pi_{k,\alpha})w_k+(\Im \Pi_{k,\alpha}) J w_k]~~(\alpha=1\ldots 2n)~~.
\een
\end{Definition}

\noindent Notice that $\Pi$ has rank $n$ as a complex matrix, i.e. its
columns are linearly independent over $\C$.  It is easy to see that
this amounts to the condition that the determinant of square
matrix $\left[\begin{array}{c}\Pi \\ {\bar \Pi} \end{array}\right]\in
\Mat(2n,2n,\C)$ is non-zero. Let: 
\be
\Lambda \eqdef \{m_1e_1+\ldots +m_{2n}e_n|m_1,\ldots, m_{2n}\in \Z\}
\ee
be the full lattice spanned by $\cE$ in $V$. The $\Pi_\cE^\cW$ is
the period matrix of the complex torus $X_c(V,J,\Lambda)$
determined by the integral complex vector space $(V,J,\Lambda)$,
computed with respect to the integral basis $\cE$ of $\Lambda$ and the
complex basis $\cW$ of $(V,J)$ (see \cite{BL}).

\subsection{Symplectic spaces}
\label{app:sympspaces}

By definition, a {\em symplectic space} is a finite-dimensional
symplectic vector space $(V,\omega)$ defined over $\R$. Given two
symplectic spaces $(V_1,\omega_1)$ and $(V_2,\omega_2)$ over $\R$, a
{\em symplectic morphism} from $(V_1,\omega_1)$ to $(V_2,\omega_2)$ is
an $\R$-linear map $f:V_1\rightarrow V_2$ such that
$\omega_2(f(\xi_1),f(\xi_2))=\omega_1(\xi_1,\xi_2)$. Such a map is
necessarily injective ($\ker f=0$) but it need not be
surjective. Symplectic vector spaces over $\R$ and symplectic
morphisms form a category denoted $\Symp$. Given a symplectic 
space $(V,\omega)$, let $\Sp(V,\omega)$ denote its group of 
automorphisms in the category $\Symp$:
\be
\Sp(V,\omega)\eqdef\{f\in \Aut_\R(V,\omega)|\omega\circ (f\otimes f)=\omega\}~~;
\ee
the elements of this group are called {\em symplectomorphisms of $(V,\omega)$}.
We set $\Sp(2n,\R)\eqdef \Sp(\R^{2n},\Omega_n)$. Notice that $\dim \Sp(2n,\R)=n(2n+1)$. 

A linear transformation $\varphi\in \Aut_\R(V)$ is symplectic with
respect to $\omega$ iff its matrix $M:=M_\cE(\varphi)$ in the
symplectic basis $\cE$ satisfies the condition:
\ben
\label{sympmatrix}
M^T\Omega_n M=\Omega_n~~,
\een
whose solution set inside $\GL(2n,\R)$ equals the symplectic group
$\Sp(2n,\R)$. Recall that the entries of
$M=(M_{\alpha,\beta})_{\alpha,\beta=1\ldots 2n}$, are defined through
the relations:
\ben
\label{sympbasis}
\varphi(e_i)=\sum_{j=1}^n M_{ji}e_j+ \sum_{j=1}^n M_{n+j
  i}f_j\, , ~~ \varphi(f_i)=\sum_{j=1}^n M_{j n+i} e_j + \sum_{j=1}^n
M_{n+j,n+i} f_j~~,
\een
which take the form $\varphi(\cE)=M^T\cE$ upon viewing $\cE$ and
$\varphi(\cE)\eqdef (\varphi(e_1)\ldots \varphi(e_n),
\varphi(f_1),\ldots \varphi(f_n))$ as a column vectors. We have:
\be
M_{\cE}(\varphi_1\circ \varphi_2)=M_\cE(\varphi_1)M_\cE(\varphi_2)=M_{\cE}(\varphi_2)M_{\varphi_2(\cE)}(\varphi_1)~~\forall
\varphi_1,\varphi_2\in \Sp(V,\omega)~~.
\ee
and: 
\be
M_\cE(\varphi^{-1})=M_{\varphi^{-1}(\cE)}(\varphi)^{-1}~~.
\ee
Writing:
\ben
\label{M}
M=\left[\begin{array}{cc} A & B\\ C & D\end{array}\right]~~,
\een
where $A,B,C,D$ are real square matrices of size $n$, relations
\eqref{sympbasis} become:
\ben
\label{matrixsymp}
\varphi(e_i)=\sum_{j=1}^n A_{ji}e_j+ \sum_{j=1}^n C_{ji}f_j~~,~~\varphi(f_i)=\sum_{j=1}^n B_{ji} e_j+ \sum_{j=1}^n D_{ji} f_j~~.
\een
Equations \eqref{sympmatrix} amount to the conditions: 
\be
A^T C=C^T A~~,~~B^T D=D^T B~~,~~A^T D-C^T B=I_n~~,
\ee
which are equivalent with: 
\be
A B^T =B A^T~~,~~C D^T=D C^T~~,~~AD^T-BC^T=I_n~~.
\ee
Given $M=\left[\begin{array}{cc} A & B\\ C & D\end{array}\right]\in \Sp(2n,\R)$, each of the following matrices is again an
element of $\Sp(2n,\R)$:
\ben
\label{Mder}
~~M^T=\left[\begin{array}{cc} A^T &
    C^T\\ B^T & D^T\end{array}\right]~~,~~
M^{-1}=\left[\begin{array}{cc} D^T & -B^T\\ -C^T &
    A^T\end{array}\right]~~,~~M^{-T}=\left[\begin{array}{cc} D &
    -C\\ -B & A\end{array}\right]~~.
\een

\subsection{Tamed symplectic spaces}

\paragraph{The category of tamed symplectic spaces.}

Let $(V,\omega)$ be a symplectic space of dimension $\dim_\R
V=2n$. A complex structure $J\in \End_\R(V)$ ($J^2=-\id_V$) is
called {\em compatible with $\omega$ } if
$\omega(Jv_1,Jv_2)=\omega(v_1,v_2)$ for all $v_1,v_2\in V$. Given such
a complex structure, the $\R$-valued bilinear pairing $Q:=Q_{J,\omega}:V\times
V\rightarrow \R$ defined through: 
\ben
\label{gomega}
Q(v_1,v_2)\eqdef \omega(Jv_1,v_2)=-\omega(v_1,Jv_2)~~,~~\forall v_1,v_2\in V
\een
is symmetric, non-degenerate and $J$-compatible, the latter meaning
that it satisfies $Q(Jv_1,Jv_2)=Q(v_1,v_2)$ for all $v_1,v_2$ in $V$.

\begin{Definition}
We say that $J$ is a {\em taming}\footnote{Also known as a positive
  $\omega$-compatible complex structure.} of $\omega$ if
$Q_{J,\omega}$ is positive-definite (and hence a Euclidean scalar
product on $V$). In this case, the triple $(V,J,\omega)$ is called a
{\em tamed symplectic space}.
\end{Definition}

\begin{Definition}
Let $(V_1,J_1,\omega_1)$ and $(V_2,J_2,\omega_2)$ be two tamed symplectic spaces. 
An $\R$-linear map $f:V_1\rightarrow V_2$ is called a {\em morphism of tamed symplectic 
spaces} if the following conditions hold: 
\begin{enumerate}[1.]
\itemsep 0.0em
\item $f$ is a morphism of symplectic spaces, i.e.:
\be
\omega_2(f(\xi_1),f(\xi_2))=\omega_1(\xi_1,\xi_2)~~,~~\forall \xi_1,\xi_2\in V_1
\ee
\item $f$ is a morphism of $\C$-vector spaces, i.e.:
\be
J_2\circ f=f\circ J_1~~.
\ee
\end{enumerate}
\end{Definition}

\noindent With this definition, tamed symplectic spaces form a category denoted
$\TamedSymp$. 

\begin{Remark}
References \cite{Vaisman, Berndt} use a different convention for the definition of $Q$ and $h$, 
which is related to ours by replacing $\omega$ with $-\omega$. Our convention for $Q$ and $h$ 
agrees with that of \cite{BL}. 
\end{Remark}

\paragraph{Relation to the category of Hermitian spaces.}

A Hermitian vector space is a triple $(V,J,h)$, where $V$ is a finite-dimensional real vector space, $J\in \End_\R(V)$ is a complex structure on $V$ and $h:V\times V\rightarrow \C$ is a Hermitian pairing with respect to $J$. Then a morphism of Hermitian spaces $f:(V_1,J_1,h_1)\rightarrow (V_2,J_2,h_2)$ is an $\R$-linear map $f:V_1\rightarrow V_2$ such that $f\circ J_1=J_2\circ
f$ and such that $h_2(f(\xi_1),f(\xi_2))=h_1(\xi_1,\xi_2)$ for all
$\xi_1,\xi_2\in V_1$. Notice that such a morphism is necessarily
injective ($\ker f=0$). Let $\Herm$ denote the category of Hermitian
spaces and morphisms of such. Consider the functors
$F:\TamedSymp\rightarrow \Herm$ and $G:\Herm\rightarrow \TamedSymp$
defined as follows:
\begin{itemize}
\item Definition of $F$. Given a tamed symplectic space $(V,J,\omega)$, the map $h_{J,\omega}:V\times
V\rightarrow \C$ given by:
\ben
\label{hdef}
h_{J,\omega}(v_1,v_2)\eqdef Q_{J,\omega}(v_1,v_2)+\i\omega(v_1,v_2)=\omega(Jv_1,v_2)+\i\omega(v_1,v_2)~~,~~\forall v_1,v_2\in V
\een
is a Hermitian scalar product on the complex vector space $(V,J)$
(complex-linear in the first variable) such that $\Re
h_{J,\omega}=Q_{J,\omega}$ and $\Im h_{J,\omega}=\omega$. If
$f:(V_1,\omega_1,J_1)\rightarrow (V_2,J_2,\omega_2)$ is a morphism of
tamed symplectic spaces, then $f$ is also a morphism of Hermitian
vector spaces from $(V_1,J_1,h_{J_1,\omega_1})$ to
$(V_2,J_2,h_{J_2,\omega_2})$. By definition, the functor $F$ takes a
tamed symplectic space $(V,J,\omega)$ into the Hermitian space
$F(V,J,\omega)=(V,J,h_{J,\omega})$ and the a morphism of tamed
symplectic spaces $f:(V_1,J_1,\omega_1)\rightarrow (V_2,J_2,\omega_2)$
into the morphism of Hermitian spaces
$F(f)=f:(V_1,J_1,h_{J_1,\omega_1})\rightarrow
(V_2,J_2,h_{J_2,\omega_2})$.
\item Definition of $G$. Given a Hermitian space $(V,J,h)$, the
  Hermitian pairing $h$ determines a symplectic pairing
  $\omega_h\eqdef \Im h$ and a Euclidean scalar product $Q_h\eqdef \Re
  h$ on $V$ which satisfy \eqref{gomega}. A morphism of Hermitian
  spaces $f:(V_1,J_1,h_1)\rightarrow (V_2,J_2,h_2)$ is also a morphism
  of tamed symplectic spaces $f:(V_1,J_1,\omega_{h_1})\rightarrow
  (V_2,J_2,\omega_{h_2})$.  The functor $F$ sends a Hermitian space
  $(V,J,h)$ into the tamed symplectic space $G(V,J,h)=(V,J,\omega_h)$
  and a morphism $f$ of Hermitian spaces into the morphism $G(f)=f$ of
  tamed symplectic spaces.
\end{itemize}

\noindent The proof of the following statement is immediate: 

\begin{Proposition}
The functors $F$ and $G$ are mutually quasi-inverse, thus giving and
equivalence of categories between $\TamedSymp$ and $\Herm$.
\end{Proposition}

\begin{Remark}
Notice that $h_{J,\omega}$ is completely determined by the
$J$-compatible Euclidean scalar product $Q_{J,\omega}$, since
$\omega=\Im h_{J,\omega}$ can be recovered as $\omega(v_1,v_2)=Q_{J,\omega}(v_1,Jv_2)$.
\end{Remark}

\noindent Consider tamed symplectic space $(V,J,\omega)$ with
associated Hermitian space $F(V,J,\omega)=(V,J,h_{J,\omega})$.  Define
$\Aut(V,J,\omega)$ as the automorphism group of $(V,J,\omega)$ in the
category $\TamedSymp$:
\be
\Aut(V,J,\omega)=\{f\in \Sp(V,\omega)|f\circ J=J\circ f\}~~.
\ee
Let $\U(V,J,h_{J,\omega})$ be the unitary group of
$(V,J,h_{J,\omega})$, which is its automorphism group in the category
$\Herm$. Then the equivalence of categories discussed above gives
the equality: 
\be
\Aut(V,J,\omega)=\U(V,J,h_{J,\omega})~~.
\ee

\paragraph{The space of tamings of a symplectic space.}

Let $(V,\omega)$ be a symplectic space and $\fJ_+(V,\omega)$ denote
the set of all tamings of $(V,\omega)$. Let $J\in \fJ_+(V,\omega)$ be
a taming and $Q\eqdef Q_{J,\omega}$, $h\eqdef h_{J,\omega}$.  A
symplectomorphism of $(V,\omega)$ commutes with $J$ iff it preserves
$Q$, in which case it also preserves $h$. Thus:
\be
\Aut(V,J,\omega)=\U(V,J,h)=\Sp(V,\omega)\cap \O(V,Q)=\{A\in \Sp(V,\omega)|A\circ J=J\circ A\}~~.
\ee 
In particular, $\Aut(V,J,\omega)$ is a maximal compact subgroup of $\Sp(V,\omega)$. 
The adjoint action $\Ad:\Sp(V,\omega)\rightarrow
\Aut_\R(\End_\R(V))$ of $\Sp(V,\omega)$:
\be
\Ad(\varphi)(A)\eqdef \varphi\circ A\circ \varphi^{-1}~~\forall \varphi\in \Sp(V,\omega)~~\forall A\in \End_\R(V)\, ,
\ee 
preserves the set of tamings of $\fJ_+(V,\omega)$, on which it induces
a smooth action which we denote by $\Ad_0:\Sp(V,\omega)\rightarrow \Diff(\fJ_+(V,\omega))$: 
\be
\Ad_0(\varphi)\eqdef \Ad(\varphi)|_{\fJ_+(V,\omega)}~~.
\ee
It is known \cite{Vaisman} that $\Ad_0$ is transitive, with isotropy
subgroup at $J\in \fJ_+(V,\omega)$ given by $\Aut(V,J,\omega)$. This
gives a homogeneous space presentation:
\ben
\label{fJhom}
\fJ_+(V,\omega)\simeq \Sp(2n,\R)/\U(n)~~.
\een
This shows that $\fJ_+(V,\omega)$ is a non-compact irreducible
Hermitian symmetric space of type III in Cartan's classification; in
particular, $\fJ_+(V,\omega)$ is a simply-connected K\"{a}hler manifold of complex
dimension:
\be
\dim_\C \fJ_+(V,\omega)=\frac{1}{2}n(n+1)~~.
\ee
Notice that $\fJ_+(V,\omega)$ is diffeomorphic with
$\R^{n(n+1)}$ and hence contractible. 

\paragraph{Characterization of tamings using a basis.}

Let $V$ be a vector space over $\R$ of dimension $2n$. Let $\omega$ be
a symplectic pairing on $V$ and $J$ be a complex structure on $V$. Let
$\cE=(e_1\ldots e_{2n})$ be a real basis of $V$ and $\cW=(w_1\ldots
w_n)$ be a complex basis of $(V,J)$. Let $\Pi\eqdef \Pi_\cE^\cW$ be
the fundamental matrix of $(V,J)$ with respect to the bases $\cE$ and
$\cW$. Let $A\in \Mat(2n,\R)$ be the matrix of $\omega$ with respect
to the real basis $\cE$:
\be
A_{\alpha\beta}\eqdef \omega(e_\alpha,e_\beta)~~\forall \alpha,\beta=1\ldots 2n~~.
\ee
The following version of Riemann bilinear relations gives the condition for $J$ to be a taming of $(V,\omega)$: 

\begin{Proposition}
\label{prop:Riemann}
The following statements hold: 
\begin{enumerate}[I.]
\itemsep 0.0em
\item $J$ is compatible with $\omega$ iff the following relation holds: 
\ben
\label{Riemann1}
\Pi A^{-1} \Pi^T=0
\een
\item $J$ is a taming of $(V,\omega)$ iff \eqref{Riemann1} holds and the real matrix:
\ben
\label{Riemann2}
\i \Pi A^{-1} \Pi^\dagger
\een
is strictly positive definite. 
\end{enumerate}
\end{Proposition}

\begin{proof}
The proof is identical to that of \cite[Theorem 4.2.1.]{BL}.
\end{proof}

\paragraph{The modular matrix of a tamed symplectic space.}

\noindent Recall that the Siegel upper half space $\SH_n$ is defined through:
\be
\SH_n\eqdef \{\tau\in \Mat_s(n,\C)|\Im \tau~\mathrm{is~strictly~positive~definite}\}~~,
\ee
where $\Mat_s(n,\C)$ is the set of all square and symmetric
complex-valued matrices of size $n$. When endowed with the natural
complex structure induced from the affine space $\Mat_s(n,\C)$, the
space $\SH_n$ is a complex manifold of complex dimension
$\frac{n(n+1)}{2}$ which is biholomorphic with the simply-connected
bounded complex domain:
\be
\{Z\in\Mat_s(n,\C)|I-{\bar Z}Z^T>0\}~~.
\ee

A symplectic basis of
$(V,\omega)$ is a basis of the form $\cE=(e_1\ldots e_n, f_1\ldots
f_n)$, whose elements satisfy the conditions:
\be
\omega(e_i,e_j)=\omega(f_i,f_j)=0~~,~~\omega(e_i,f_j)=-\omega(f_i,e_j)=\delta_{ij}~~,
\ee
which state that matrix of $\omega$ in such a basis equals the
antisymmetric matrix \eqref{Omegan}. 

\begin{Proposition}
\label{prop:PiE}
Let $\cE=(e_1\ldots e_n, f_1\ldots f_n)$ be a symplectic basis of
$(V,\omega)$. Then the vectors $\cE_2\eqdef (f_1,\ldots, f_n)$ form a
basis of the complex vector space $(V,J)$ over $\C$ and the
fundamental matrix of $(V,J)$ with respect to $\cE$ and $\cE_2$ has
the form:
\ben
\label{PiE}
\Pi_\cE^{\cE_2}=[\btau^\cE, I_n]^{T}\, ,
\een
where $\btau^\cE\in \Mat(n,n,\C)$ is a complex-valued square matrix of
size $n$. Moreover, $J$ is a taming of $(V,\omega)$ iff $\btau_\cE$
belongs to $\SH_n$. In this case, the matrix of the Hermitian form $h$
with respect to the basis $\cE_2$ equals $(\Im\btau^\cE)^{-1}$.
\end{Proposition}

\begin{proof}
Let $W$ denote the span of $f_1,\ldots, f_n$ over $\R$, which is a
Lagrangian subspace of $(V,\omega)$.  For any $x\in W\cap J(W)$, we
have $Jx\in W$ and hence $Q(x,x)=\omega(x,Jx)=0$, which implies $x=0$.
Thus $W\cap J(W)=0$, which implies that $\cE_2$ is a basis of $(V,J)$
over $\cE$. It is clear that the fundamental matrix $\Pi_\cE^{\cE_2}$
has the form \eqref{PiE} for some complex-valued matrix
$\btau:=\btau^\cE\in \Mat(n,n,\C)$. The fact that $J$ is a taming of
$(V,\omega)$ iff $\btau$ belongs to $\SH_n$ follows from Proposition
\ref{prop:Riemann}. Relation \eqref{hdef} gives:
\be
h(f_i,f_j)=\omega(f_i,Jf_j)
\ee
Using this and \eqref{PiDef}, we compute:
\be
\omega(e_i,f_j)=(\Im \btau)_{ik}\omega(f_k,Jf_j)=h(f_j,f_k)(\Im\btau)_{ki}~~.
\ee
Since $\omega(e_i,f_j)=\delta_{ij}$, this implies that the matrix of
$h$ in the basis $\cE_2$ equals $(\Im \btau)^{-1}$. 
\end{proof}

\begin{Definition} 
Let $(\cS,J,\omega)$ be a tamed symplectic space and $\cE$ be a
symplectic basis of $(V,\omega)$. The fundamental matrix
$\Pi_\cE^{\cE_2}$ is called the {\em canonical fundamental matrix} of
$(\cS,J,\omega)$ with respect to $\cE$, while $\btau^\cE\in \SH_n$ is
called the {\em modular matrix} of $(\cS,J,\omega)$ with respect to
$\cE$. 
\end{Definition}

\subsection{The modular map of a symplectic space relative to a symplectic basis}

\noindent Let $(V,\omega)$ be a $2n$-dimensional symplectic
space and $\cE=(e_1\ldots e_n, f_1\ldots f_n)$ be a symplectic basis
of $(V,\omega)$. 

\begin{Definition}
The {\em modular map} of $(V,\omega)$ relative to the symplectic basis
$\cE$ is the map $\btau^\cE:\fJ_+(V,\omega)\rightarrow \SH_n$ which
associates to a taming $J$ the modular matrix $\btau_\cE(J)$ of the
tamed symplectic space $(V,J,\omega)$.
\end{Definition}

\noindent Proposition \ref{prop:Riemann} implies that $\btau^\cE$ is a
bijection from $\fJ_+(V,\omega)$ to the Siegel upper half space
$\SH_n$. Using \eqref{fJhom}, this gives the homogeneous space presentation:
\be
\SH^n\simeq \Sp(2n,\R)/\U(n)~~.
\ee 
The inverse of the modular map is given explicitly as follows:

\begin{Proposition}
Let $\btau=\btau_R+i\btau_I \in \SH_n$ be a point of the Siegel upper
half space, where $\btau_R,\btau_I\in \Mat_s(n,n,\R)$ are the real and
imaginary parts of $\btau$:
\be
\btau_R\eqdef \Re\btau~~,~~\btau_I\eqdef \Im \btau~~.
\ee
Then the matrix of the taming $J(\btau)\eqdef \btau_\cE^{-1}(\btau)$ in the
symplectic basis $\cE$ is given by:
\ben
\label{hatJ}
{\hat J}(\btau)=\left[\begin{array}{cc}  \btau_I^{-1}\btau_R~~ &~~\btau_I^{-1}\\-\btau_I-\btau_R\btau_I^{-1}\btau_R
    &~~-\btau_R\btau_I^{-1}\end{array}\right]~~.
\een
\end{Proposition}

\begin{proof}
We have:
\ben
\label{Jmatrix}
{\hat J}=\left[\begin{array}{cc} K & P \\ Q & R\end{array}\right]~~,
\een
where:
\beqa
Je_i &=& K_{ji} e_j+Q_{ji} f_j\nn\\
Jf_i & = &P_{ji} e_j+R_{ji} f_j~~,
\eeqa
and we use Einstein summation over $j=1\ldots n$. Since the canonical
fundamental matrix in the basis $\cE$ has the form $\Pi=[\btau,I_n]$,
we have:
\be
e_i=(\btau_R)_{ij}f_j+(\btau_I)_{ij} J f_j~~,
\ee
which gives: 
\be
J f_i=(\btau_I^{-1})_{ij} e_j -(\btau_I^{-1}\btau_R)_{ij} f_j
\ee
and: 
\be
Je_i=(\btau_R)_{ij} J f_j-(\btau_I)_{ij} f_j=(\btau_R \btau_I^{-1})_{ij} e_j-(\btau_I +\btau_R \btau_I^{-1}\btau_R)_{ij} f_j~~.
\ee
Thus $K=\btau_I^{-1}\btau_R$, $Q=-\btau_I -\btau_R \btau_I^{-1}\btau_R$,
$P=\btau_I^{-1}$ and $R=-\btau_R\btau_I^{-1}$.
\end{proof}

\begin{Remark}
The matrix ${\hat J}$ satisfies the relations: 
\be
{\hat J}^2=-I_{2n}~~,~~{\hat J}^t\Omega_n{\hat J}=\Omega_n~~,
\ee
which are equivalent with: 
\be
{\hat J}^2=-I_{2n}~~,~~{\hat J}^t=\Omega_n{\hat J}\Omega_n~~,
\ee
where we used the relation $\Omega_n^2=-I_{2n}$. The proposition
implies that any matrix ${\hat J}\in \Mat(2n,2n,\R)$ satisfying these
relations has the form \eqref{hatJ} for some uniquely-determined 
complex matrix $\btau=\btau_R+\i \btau_I$ belonging to $\SH_n$. Conversely, 
the matrix \eqref{hatJ} satisfies these relations for any $\btau\in \SH_n$. 
\end{Remark}

\paragraph{Equivariance properties of the modular map.}

Let $\mu:\Sp(2n,\R)\rightarrow \Diff(\SH_n)$ denote the modular action
of $\Sp(2n,\R)$ on the Siegel upper half space, i.e. the action
through matrix fractional transformations:
\ben
\label{fracmatrix}
\mu(M)(\btau):=M\bullet \btau\eqdef (A\btau+B)(C\btau+D)^{-1}~~
\forall M=\left[\begin{array}{cc} A & B\\ C & D\end{array}\right]\in \Sp(2n,\R)~~\forall \btau\in \SH_n~~.
\een
The group $\Sp(V,\omega)$ acts freely and transitively on the space
$\SB(V,\omega)$ of symplectic bases of $(V,\omega)$ as follows:
\be
\varphi(\cE)=(\varphi(e_1),\ldots \varphi(e_n),\varphi(f_1),\varphi(f_n))~~,~~\forall \varphi\in \Sp(V,\omega)~~,~~\forall \cE=(e_1,\ldots, e_n, f_1,\ldots, f_n)\in \SB(V,\omega)~~.
\ee 

\begin{Proposition}
For any $\varphi\in \Sp(V,\omega)$, we have:
\beqan
\label{tauequiv}
&& \btau^{\varphi(\cE)}(J)= M_\cE(\varphi)^T \bullet \btau^{\cE}(J)\nn\\
&& \btau^\cE(\Ad_0(\varphi)(J)) = M_\cE(\varphi^{-1})^T \bullet \btau^\cE(J)~~,
\eeqan
where $M_\cE(\varphi)\in \Sp(2n,\R)$ is the matrix of $\varphi$ in the
symplectic basis $\cE$. In particular, the modular map is invariant under
the diagonal action:
\be
\btau^{\varphi(\cE)}(\Ad_0(\varphi)(J))= \btau^\cE(J)~~\forall \varphi\in \Sp(V,\omega)~~.
\ee
\end{Proposition}

\begin{proof}
Let $\cE'\eqdef \varphi(\cE)$ and $\btau\eqdef \btau_\cE(J)$,
$\btau'\eqdef \btau_{\cE'}(J)$.  Setting $M\eqdef M_\cE(\varphi)$, we
have $\cE'=M^T\cE$, where $M$ has the form \eqref{M}. Using relations
\eqref{Mder}, this gives $e'=A^T e+C^T f=(A^T\btau +C^T)f$. On the
other hand, we have $e'=\btau' f'=\btau' (B^T e+D^T f)=\btau'(B^T\btau +
D^T)f$. Comparing these two expressions gives:
\be
\btau'=(A^T\btau +C^T)(B^T\btau + D^T)^{-1}=M^T \bullet \btau~~,
\ee
which shows that the first relation in \eqref{tauequiv} holds.  The
matrix of $\Ad_0(\varphi)(J)$ in the symplectic basis $\cE$ coincides
with the matrix of $J$ in the symplectic basis $\varphi(\cE)$. This
implies
$\btau^{\varphi(\cE)}(\Ad_0(\varphi)(J))=\btau^\cE(J)$. Replacing $\cE$
by $\varphi^{-1}(\cE)$ and using the first relation in
\eqref{tauequiv}, this gives:
\be
\btau^{\cE}(\Ad_0(\varphi)(J))=\btau^{\varphi^{-1}(\cE)}(J)=M_\cE(\varphi^{-1})^T \bullet \btau^{\cE}(J)~~,
\ee
which shows that the second relation in \eqref{tauequiv} holds. 
\end{proof}

\paragraph{The Lagrangian Grassmannian.}

Let $(V,J,\omega)$ be a tamed symplectic space. 
Then one can show \cite{Vaisman, Berndt} that a subspace $L\subset V$ is
Lagrangian with respect to $\omega$ iff $V=L\oplus J(L)$ and
$h(L\times L)\subset \R$. Let $\cL(V,\omega)$ be the Lagrangian
Grassmannian of $(V,\omega)$ (i.e. the set of all
Lagrangian subspaces).  Then $\cL(V,\omega)$ is a real analytic
manifold of dimension $\frac{n(n+1)}{2}$, being an embedded
submanifold of the ordinary Grassmannian $\Gr_n(V)$ of $n$-planes of
$V$. 

\subsection{Positive polarizations of $(V_\C,\omega_\C)$}
\label{app:pospol}

Let $V_\C$ denote the complexification of $V$ and $\bar{~}$ denote the
complex conjugation of $V_\C$. Any $x\in V_\C$ decomposes as:
\be
x=\Re x+\i~\Im x~~,
\ee
where $\Re x\eqdef \frac{1}{2}(x+\bar{x})$ and $\Im x\eqdef
\frac{1}{2\i}(x-\bar{x})$ belong to $V$.  Let $\omega_\C:V_\C\times
V_\C\rightarrow \C$ denote the complexification of $\omega$, which is
a complex symplectic form on $V_\C$. Consider the non-degenerate
Hermitian form $q:V_\C\times V_\C\rightarrow \C$ (complex-linear in
the first variable) defined through:
\be
q(x,y)\eqdef \i \omega_\C(x,\bar{y})~~,~~\forall x,y\in V_\C~~.
\ee
A complex Lagrangian subspace $L\subset V_\C$ of the complex
symplectic space $(V_\C,\omega_\C)$ is called {\em positive} if the
restriction of $q$ is a Hermitian scalar product on $L$, i.e. if:
\be
\i\omega_\C(x,\bar{x})>0~~,~~\forall x\in L\setminus\{0\}~~.
\ee
A positive Lagrangian subspace of $(V_\C,\omega_\C)$ is also called a
{\em positive polarization} of $(V_\C,\omega_\C)$. Let
$\cL_+(V_\C,\omega_\C)$ denote the set of all positive
polarizations, which is an open real submanifold of the 
complex Lagrangian Grassmannian $\cL(V_\C,\omega_\C)$.
For any positive polarization $L$, the space $\bar{L}$
is also a positive polarization and we have a direct sum decomposition
$V_\C=L\oplus \bar{L}$, which implies $V=\Re(L)=\Re(\bar{L})=\Im
(L)=\Im (\bar{L})$. Since $\dim_\R L=\dim_\R V=2n$, this implies that
the $\R$-linear maps $\Re|_L:L\rightarrow V$ and $\Im |_L:L\rightarrow
V$ are bijective. The action $\rho:\Sp(V,\omega)\rightarrow
\Diff(\cL(V_\C,\omega_\C))$ given by:
\ben
\label{rhodef}
\rho(\varphi)(L)\eqdef \varphi^\C(L)~~\forall L\in \cL(V_\C,\omega_\C)~~\forall \varphi\in \Sp(V,\omega)
\een
preserves $\cL_+(V_\C,\omega_\C)$, on which it restricts to an action
which we denote by $\rho_0$. One can show that $\rho_0$ is
transitive, with isotropy group at $L\in \cL_+(V_\C,\omega_\C)$ given
by $\U(L,q)$, thus giving a homogeneous space presentation:
\be
\cL_+(V_\C,\omega_\C)\simeq \Sp(2n,\R)/\U(n)~~.
\ee 

\subsection{Equivalence of tamings and positive polarizations}

Given a taming $J\in \fJ_+(V,\omega)$, let $H:V_\C\times
V_\C\rightarrow \C$ denote the Hermitian scalar product (complex-linear
in the first variable) induced by $Q$ on $V_\C$:
\be
H(v_1+\i w_1,v_2+\i w_2)=Q_\C(x,\bar{y})=Q(v_1,v_2)+Q(w_1,w_2)+\i[Q(w_1,v_2)-Q(v_1,w_2)]~~\forall v_1,v_2,w_1,w_2\in V~~.
\ee
We have: 
\be
H(x,y)=\omega_\C(J^\C(x),\bar{y})=-\omega_\C(x,J^\C(\bar{y}))~~\mathrm{and}~~q(x,y)=-\i H(J^\C(x), y)~~\forall x,y\in V_\C\, ,
\ee
where $J^\C\in \End_\C(V_\C)$ denote the complexification of $J$, which
has eigenvalues $\pm \i$. Then $J^\C$ is $\omega_\C$-compatible and
commutes with the complex conjugation of $V_\C$. Let:
\be
L_J^\pm\eqdef \ker(J^\C\mp \i~\id_{V_\C})\subset V_\C
\ee 
denote the eigenspaces of $J^\C$ corresponding to the eigenvalues $\pm
\i$, which are complementary complex Lagrangian subspaces of
$(V_\C,\omega_\C)$. We have $V_\C=L^+_J\oplus L^-_J$,
$\sigma(L_J^\pm)=L_J^\mp$ and the operators $P^\pm_J\eqdef
\frac{1}{2}(\id_{V_\C}\mp \i J^\C)\in \End_\C(V_\C)$ are complementary
$H$-Hermitian projectors on $L^\pm_J$, hence $L_J^+$ and $L_J^-$
are $H$-orthogonal to each other. In particular, we have
$L_J^\pm=\ker P_J^\mp$. The space:
\be
L_J\eqdef L_J^+=\ker(J^\C- \i~\id_{V_\C})
\ee
is called {\em the complex Lagrangian subspace of $(V_\C, \omega_\C)$
  determined by the taming $J$}.  With this notation, we have
$L_J^-=\overline{L_J}$ and an $H$-orthogonal decomposition
$V_\C=L_J\oplus \overline{L_J}$. Moreover, the map:
\ben
\label{kappaJ}
\kappa_J\eqdef \sqrt{2}P^+_J|_V=\frac{1}{\sqrt{2}}(\id_V- \i J):V\stackrel{\sim}{\rightarrow} L_J
\een
is an isomorphism of Hermitian vector spaces between $(V,J,h)$ and $(L_J,H)$, whose 
inverse is given by: 
\ben
\label{Re}
\kappa_J^{-1}=\sqrt{2}\Re|_{L_J}:L_J\stackrel{\sim}{\rightarrow} V~~.
\een
We have: 
\be
L_J=\{x\in V_\C|\Im x=-J(\Re x)\}=\{x\in V_\C|\Re x=J(\Im x)\}=\{v-\i Jv|v\in V\}~~.
\ee
and $q|_{L_J\times L_J}=H|_{L_J\times L_J}$, which gives: 
\be
h(v_1,v_2)=q(\kappa_J(v_1), \kappa_J(v_2))~~\forall v_1, v_2 \in V~~.
\ee
In particular, $\kappa_J$ is an isomorphism from $(V,J,h)$ to
$(L_J,q)$ and hence $L_J$ is a positive polarization of $(V_\C,\omega_\C)$:
\be
L_J\in \cL_+(V_\C,\omega_\C)~~\forall J\in \fJ_+(V,\omega)~. 
\ee
In fact, giving a taming of $(V,\omega)$ is {\em equivalent} to giving
a positive polarization of $(V_\C,\omega_\C)$:

\begin{Proposition}{\cite{Vaisman}} 
The map $\lambda:\fJ_+(V,\omega)\rightarrow \cL_+(V_\C,\omega_\C)$ given by
$\lambda(J)=L_J$ is bijective and its inverse
$\lambda^{-1}:\cL_+(V_\C,\omega_\C)\rightarrow \fJ_+(V,\omega)$ is
given as follows. Given $L\in \cL_+(V_\C,\omega_\C)$, the
corresponding taming $J_L\eqdef \lambda^{-1}(L)$ is the complex
structure of $V$ obtained by transporting the complex structure of $L$
through the $\R$-linear bijection $ \Re|_L:L\rightarrow
V$. Explicitly, we have $J_L\circ \Re|_L=\Re|_L\circ (\i\cdot)$, i.e.:
\be
J_L(x+\bar{x})=\i (x-\bar{x})~~\forall x\in L~~,
\ee
where $\i\cdot:L\rightarrow L$ denotes the operator of multiplication
by $\i$ acting on $L$.
\end{Proposition}

\

\noindent For any $\varphi\in \Sp(V,\omega)$, we have $\ker
(\Ad(\varphi^\C)(J^\C)- \i~\id_{V_\C})=\varphi^\C(\ker(J^\C-
\i~\id_{V_\C}))$.  This reads:
\ben
\label{lambda_equiv}
\lambda\circ \Ad_0(\varphi)=\rho(\varphi)\circ \lambda~~\forall \varphi\in \Sp(V,h)~~,
\een
showing that $\lambda$ is $\Sp(V,h)$-equivariant. 

\subsection{The modular matrix of a positive polarization}
\label{app:Siegel}

Given a symplectic basis $\cE=(e_1\ldots
e_n, f_1\ldots f_n)$ of $(V,\omega)$ and a matrix $\tau\in \SH_n$, the
complex subspace $L^\cE(\tau)$ of $V_\C$ spanned by the following vectors
(which are linearly-independent over $\C$):
\ben
\label{uk}
u_k^\cE(\tau)\eqdef e_k+\sum_{l=1}^n\tau_{kl}f_l\in V_\C~~(k=1\ldots n)
\een
is a positive complex Lagrangian subspace of $(V_\C,\omega_\C)$. 

\begin{Remark}
If one works with the conventions of references \cite{Vaisman,
  Berndt}, then one has to interchange $e_i$ and $f_i$ in
relation \eqref{uk}. This is because $\cE=(e_i,f_i)$ is a symplectic basis for 
$\omega$ iff $\cE'=(f_i,e_i)$ is a symplectic basis for $-\omega$. 
\end{Remark}

\begin{Proposition}{\cite{Vaisman}} For any symplectic basis $\cE$ of
$(V,\omega)$, the map $L^\cE:\SH_n\rightarrow \cL_+(V_\C,\omega_\C)$
given by $L^\cE(\tau)=\oplus_{k=1}^n \C u_k^\cE(\tau)$ is a bijection
from $\SH_n$ to $\cL_+(V_\C,\omega_\C)$.
\end{Proposition}

\noindent Thus any positive polarization $L\in \cL_+(V_\C,\omega_\C)$
and any symplectic basis $\cE$ of $(V,\omega)$ determine a unique
matrix $\tau^\cE(L)\in \SH_n$ with the property that the vectors
\eqref{uk} form a basis of $L$. This is called the {\em modular matrix
  of $L$} with respect to the symplectic basis $\cE$.

\begin{Proposition} 
We have:
\ben
\label{modmatrix}
\tau^\cE\circ \lambda=-\overline{\btau^\cE}~~.
\een
\end{Proposition}

\begin{proof}
Fix $J\in \fJ_+(V,\omega)$ and $L\eqdef \lambda(J)$. Let $\tau\eqdef
\tau^\cE(L)=(\tau^\cE\circ \lambda)(J)$.  Let $u_k\eqdef
u_k^\cE(\tau)$. Since $u_k$ belong to $L$, we have $J(\Re u_k)=\Re(\i
u_k)$ which gives:
\be
\Re u_k=-J(\Re (\i u_k))~~.
\ee
Substituting $u_k=e_k+ \tau_{kl}f_l$ in both terms of this relation gives:
\be
e_k+(\Re \tau)_{kl} f_l= (\Im \tau)_{kl}J(f_l)\Longrightarrow e_k=-(\Re \tau)_{kl} f_l+(\Im \tau)_{kl} J(f_l)~~,
\ee
which shows that $\tau=-\overline{\btau^\cE(J)}$. 
\end{proof}

\noindent Combining everything, we have a commutative diagram of bijections:
\ben
\label{diag:fJ}
\scalebox{1.2}{
\xymatrix{
\fJ_+(V,\omega) \ar[d]_{\lambda} \ar[dr]^{~~-\overline{\btau^\cE}} \\
\cL_+(V_\C,\omega_\C)  \ar[r]^{~~~~\tau^\cE} & \SH_n\\
}}
\een
where all maps are $\Sp(V,\omega)$-equivariant for the actions on their
argument. Moreover, $\tau^\cE$ and $\btau^\cE$ are separately
$\Sp(V,\omega)$-equivariant for the action on $\cE$.

\begin{Remark}
Let $\theta\eqdef \Re \tau$ and $\gamma\eqdef \Im \tau$. Then
$\tau=\theta+\i \gamma$ and $\btau\eqdef -\bar{\tau}=-\theta+\i
\gamma$. Using \eqref{hatJ}, we have:
\ben
\label{hatJtg}
{\hat J}(\btau)=\left[\begin{array}{cc}  -\gamma^{-1}\theta~~ &~~\gamma^{-1}\\-\gamma-\theta\gamma^{-1}\theta
    &~~+\theta\gamma^{-1}\end{array}\right]~~.
\een
\end{Remark}

\section{Tamings and positive polarizations of symplectic vector bundles}
\label{app:bundle_tamings}

\subsection{Tamings}

Let $(\cS,\omega)$ be a real symplectic vector bundle over a connected
manifold $N$, where $\rk \cS=2n$. Let $P$ denote the
principal $\Sp(2n,\R)$-bundle of symplectic frames of
$(\cS,\omega)$. Then $\cS$ is associated to $P$ through the
fundamental (=tautological) representation of $\Sp(2n,\R)$ in
$\R^{2n}$. Let $\bfJ_+(\cS,\omega)$ be the {\em bundle of tamings} of
$(\cS,\omega)$, which we define to be the fiber bundle whose fiber at
$p\in N$ is the space $\fJ_+(\cS_p,\omega_p)$ of tamings of the
symplectic vector space $(\cS_p,\omega_p)$. The discussion of Appendix
\ref{app:space_tamings} shows that $\bfJ_+(\cS,\omega)$ is a bundle of
homogeneous spaces (with fibers isomorphic to $\Sp(2n,\R)/\U(n)$)
which is associated to $P$ through the representation
$\Ad_0$. Finally, let $\SH(\cS,\omega)$ denote the {\em Siegel bundle}
of $(\cS,\omega)$, which we define as the fiber bundle with typical
fiber $\SH_n$ which is associated to $P$ through the modular
representation $\mu$ of $\Sp(2n,\R)$.

\begin{Definition}
A {\em taming} of $(\cS,\omega)$ is a smooth section of the fiber
bundle $\bfJ_+(\cS,\omega)$, i.e. an almost complex structure $J$ on
$\cS$ such that $J_p$ is a taming of $(\cS_p,\omega_p)$ for every $p\in
N$.
\end{Definition}

\

\noindent A {\em tamed symplectic vector bundle} is a triple
$(\cS,J,\omega)$, where $(\cS,\omega)$ is a symplectic vector bundle
for which $J$ is a taming. The pair $(J,\omega)$ defines a Hermitian
scalar product $h$ on $\cS$ through relation \eqref{hdef}, which in
turn determines $\omega$ uniquely once $J$ is given. Thus $(\cS, J,h)$
is a Hermitian vector bundle and this correspondence defines an
equivalence of categories between the category of tamed symplectic
vector bundles on $N$ and the category of Hermitian vector bundles on
$N$.

Giving a taming $J$ of $(\cS,\omega)$ amounts to giving a reduction of
structure group from $\Sp(2n,\R)$ to $\U(n)$. Since the homogeneous
space $\Sp(2n,\R)/\U(n)\simeq \SH_n$ is contractible, such a reduction
always exists. In fact, the space $\fJ_+(\cS,\omega)$ of tamings of
$(\cS,\omega)$ is non-empty and contractible \cite{DuffSalamon}.

\paragraph{Local description.} 
Suppose that we are given a local symplectic frame $\cE\in
\Gamma(U,P)$ of $(\cS,\omega)$ defined above an open subset $U\subset
N$. Then the restriction $P|_U$ is trivial and so is the restriction
$\SH(\cS,\omega)|_U\simeq U\times \SH_n$. The discussion of Subsection
\ref{app:Siegel} shows that specifying a taming $J\in
\Gamma(M,\bfJ_+|_U)$ of $(\cS|_U,\omega|_U)$ amounts to specifying a
smooth section $\tau\circ (\cE\times_M J)$ of $\SH(\cS,\omega)|_U$,
which can be identified with a smooth map $\tau^\cE(J)\in
\cC^\infty(U,\SH_n)$. Let $(\cU_\alpha)_{\alpha\in I}$ be an open
cover of $N$ which supports a symplectic trivialization atlas of
$(S,\omega)$, whose transition functions $U_\alpha$ to $U_\beta$ we
denote by $g_{\alpha\beta}:U_\alpha\cap U_\beta\rightarrow
\Sp(2n,\R)$. Then giving a global taming of $(\cS,\omega)$ amounts to
specifying maps $\tau_\alpha\in \cC^\infty(U_\alpha,\SH_n)$ such that
$\tau_\beta=g_{\alpha\beta}\bullet \tau_\alpha$.

\subsection{Positive polarizations}

Let $\cS_\C$ denote the complexification of $\cS$ and $\omega_\C$
denote the complexification of $\omega$. Then $\cS_\C$ is associated
to $P$ through the complexification of the fundamental representation
of $\Sp(2n,\R)$. Let $\bcL_+(\cS_\C,\omega_\C)$ be the {\em bundle of
  positive polarizations} of $(\cS_\C,\omega_\C)$, which we define to
be the fiber bundle whose fiber at $p\in N$ is the set of
$\cL_+(\cS_{\C, p},\omega_{\C, p})$ of positive complex Lagrangian subspaces
of $(\cS_{\C, p},\omega_{\C, p})$; this is a bundle of homogeneous spaces
(with fibers isomorphic with $\Sp(2n,\R)/\U(n)$) which is associated
to $P$ through the representation $\rho$ defined in \eqref{rhodef}.
The equivariant diffeomorphism $\lambda$ of Subsection
\ref{app:pospol} globalizes to a natural isomorphism of fiber bundles
$\lambda:\bfJ_+(\cS,\omega)\stackrel{\sim}{\rightarrow}
\bcL_+(\cS,\omega)$, while the period maps globalize to natural bundle
maps $\btau:P\times_N \bfJ_+(\cS,\omega)\rightarrow \SH(\cS,\omega)$
and $\tau=-\bar{\btau}\circ
(\id\times_M\lambda):P\times_N\bcL_+(\cS,\omega)\rightarrow
\SH(\cS,\omega)$.

\begin{Definition}
A {\em positive polarization} of
$(\cS_\C,\omega_\C)$ is a smooth section of the fiber bundle
$\bcL_+(\cS_\C,\omega_\C)$, i.e a $\C$-linear sub-bundle $L\subset
\cS_\C$ whose fiber $L_p$ at each $p\in N$ is a positive complex
Lagrangian subspace of $(\cS_{\C, p},\omega_{\C, p})$.
\end{Definition}

\subsection{Correspondence between tamings and positive polarizations}

Let $\fJ_+(\cS,\omega)=\Gamma(N,\bfJ_+(\cS,\omega))$ and
$\cL_+(\cS,\omega)=\Gamma(N,\bcL_+(\cS,\omega))$ denote respectively
the sets of tamings and positive polarizations. The bundle isomorphism
$\lambda$ induces a bijection between these sets. This takes a taming
$J\in \fJ_+(\cS,\omega)$ of $(\cS,\omega)$ into the positive
polarization $L_J=\ker(J_\C-\i~\id_{\cS_\C})$ of
$(\cS_\C,\omega_\C)$. The inverse of $\lambda$ takes a positive
polarization $L$ of $(\cS_\C,\omega_\C)$ into the unique complex
structure $J_L$ of $\cS$ which makes the $\R$-linear isomorphism
$\Re|_L:L\stackrel{\sim}{\rightarrow}\cS$ into a complex-linear
map. Thus specifying a taming of $(\cS,\omega)$ amounts to specifying
a positive polarization of $(\cS_\C,\omega_\C)$. 

\subsection{Polarized two-forms valued in a tamed symplectic vector bundle}

\begin{Definition}
Let $J\in \fJ_+(\cS,\omega)$ be a taming of
$(\cS,\omega)$ and $L\eqdef \lambda(J)\in \cL_+(\cS_\C,\omega_\C)$ be the
corresponding positive polarization.  The {\em J-projection} of an
$\cS$-valued 2-form $\eta\in \Omega^2(N,\cS)$ is the unique $L$-valued
2-form $\omega_J\in \Omega^2(N,L)$ such that $\Re\eta_J=\eta$, namely:
\be
\eta_J=2 P_J^+(\eta)=\sqrt{2}\kappa_J(\eta)=\eta-\i J^\C\eta~~.
\ee
\end{Definition}

\noindent Assume that $N$ is four-dimensional and endowed with a
Lorentzian metric $g$. Then the Hodge operator
$\ast_g:\Omega^2(N)\rightarrow \Omega^2(N)$ defined by $g$ induces an
operator $\ast_g^\C\eqdef \ast_g\otimes
\id_{\cS_\C}:\Omega^2(N,\cS_\C)\rightarrow \Omega^2(N,\cS_\C)$, which
is the complexification of the operator $\ast_g\eqdef \ast_g\otimes
\id_{\cS}$.  The complexification of the operator $\star\eqdef
\ast_g\otimes J=\ast_g\circ J$ is given by $\star^\C=\ast_g\otimes
J^\C=\ast_g^\C\circ J^\C$.  The operator $\ast_g^\C$ squares to
minus the identity, so it has eigenvalues $\pm \i$. On the other hand,
$\star_g^\C$ squares to the identity, so it has eigenvalues $\pm 1$.
Let:
\beqa
\Omega^{2\pm}_g(N,L) &\eqdef&  \{\rho\in \Omega^2(N,L)|\star_g^\C \rho=\pm \rho\} =\{\rho\in \Omega^2(N,L)|\ast_g^\C \rho=\mp \i\rho\}\nn\\
\Omega^{2\pm}_{g,J}(N,\cS) &\eqdef& \{\eta\in \Omega^2(N,\cS)|\star_g \eta= \pm \eta\}=\{\eta\in \Omega^2(N,\cS)|\ast_g \eta=\mp J \eta\}
\eeqa
Since the isomorphism $\kappa_J:\cS\stackrel{\sim}{\rightarrow} L$
maps $J$ into multiplication by $\i$, we have:
\be
\kappa_J(\Omega^{2\pm}_{g,J}(N,\cS))=\Omega^{2\pm}_g(N,L)
\ee
In particular, $\kappa_J$ identifies the space
$\Omega^{2+}_{g,J}(N,\cS)$ of positively-polarized 2-forms valued in
$\cS$ with the space $\Omega^{2+}(N,L)$. Thus $\eta\in
\Omega^2(N,\cS)$ is positively polarized with respect to $g$ and $J$
iff $\eta_J$ belongs to $\Omega^{2+}(N,L)$.

\paragraph{Local form of the $J$-polarization condition.}

Let $\cE=(e_1\ldots e_n, f_1\ldots f_n)$ be a symplectic local frame
of $\cS$ defined on an open subset $U\subset N$. Given $\eta\in
\Omega^2(N,\cS)$, expand:
\ben
\label{cVexp}
\eta=_U \alpha_k \otimes e_k + \beta_k \otimes f_k
\een
with $\alpha_k,\beta_k\in \Omega^2(U)$. Set $\btau\eqdef
\btau^\cE(L)\in \cC^\infty(U,\SH_n)$ and define ${\hat \alpha}, {\hat
  \beta}\in \Mat(n,1, \Omega^2(U))$ and ${\hat \eta}\in \Mat(2n,1,
\Omega^2(U))$ through:
\be
{\hat \alpha}\eqdef
\left[\begin{array}{c}\alpha_1\\\ldots\\ \alpha_n\end{array}\right]~,~{\hat \beta}\eqdef
\left[\begin{array}{c}\beta_1\\\ldots\\ \beta_n\end{array}\right]~~,~~{\hat \eta}\eqdef 
\left[\begin{array}{c}{\hat \alpha} \\ {\hat \beta} \end{array}\right]
\ee
The $J$-projecton of $\eta$ expands as:
\ben
\label{etaexp}
\eta_J=_U \alpha_k^+ \otimes e_k + \beta_k^+\otimes f_k
\een 
with $\alpha_k^+,\beta_k^+\in \Omega^2_\C(U)$, where $\Re
\alpha_k^+=\alpha_k$ and $\Re \beta_k^+=\beta_k$. Define ${\hat
  \alpha}^+, {\hat \beta}^+\in \Mat(n,1, \Omega^2_\C(U))$ and ${\hat
  \eta}_J\in \Mat(2n,1, \Omega^2_\C(U))$ through:
\be
{\hat \alpha}^+\eqdef
\left[\begin{array}{c}\alpha_1^+\\\ldots\\ \alpha_n^+\end{array}\right]~,~{\hat \beta}^+\eqdef
\left[\begin{array}{c}\beta_1^+\\\ldots\\ \beta_n^+\end{array}\right]~~,~~{\hat \eta}_J\eqdef 
\left[\begin{array}{c}{\hat \alpha}^+ \\ {\hat \beta}^+ \end{array}\right]
\ee
Finally, write:
\be
\tau=\theta+i\gamma~~,
\ee
where $\theta,\gamma\in \cC^\infty(U,\Mat_s(n,\R))$ and $\gamma$ is strictly positive-definite. 

\begin{Proposition}
\label{prop:polrels}
The following relation holds: 
\ben
\label{abmatrix}
{\hat \beta}^+=\tau {\hat \alpha}^+~~.
\een
Moreover, the following statements are equivalent: 
\begin{enumerate}[1.]
\itemsep 0.0em
\item The 2-form $\eta|_U\in \Omega^2(U,\cS)$ is positively polarized with respect to $g$ and $J$.
\item The following relation holds: 
\ben
\label{etapol}
{\hat \beta}+\i \ast_g {\hat \beta}=\tau ({\hat \alpha}+\i \ast_g {\hat \alpha})~~.
\een
\item The following relation holds: 
\ben
\label{etapol2}
{\hat \beta}=\theta {\hat \alpha} - \gamma \ast_g {\hat \alpha}~~.
\een
\item The following relation holds: 
\ben
\label{etapol3}
\ast_g {\hat \beta}=\gamma {\hat \alpha}+\theta \ast_g {\hat \alpha}~~.
\een
\item The following relation holds: 
\ben
\label{etapol3}
\ast_g {\hat \eta} =\left[\begin{array}{cc} \gamma^{-1}\theta~~&-\gamma^{-1}\\\gamma+\theta\gamma^{-1}\theta 
    &-\theta \gamma^{-1}\end{array}\right]{\hat \eta}~~.
\een

\end{enumerate}
\end{Proposition}

\begin{proof}
Recall that the local sections $u_k\eqdef e_k+\sum_{l=1}^n
\tau_{kl}f_l\in \Gamma(U,\cS_\C)$ (where $k=1\ldots n$) form a frame
of $L|_U$. The fact that $\eta_J$ belongs to $\Omega^2(M,L)$ amounts
to the condition $\eta_J=\beta_k^+\otimes u_k$, which in turn is
equivalent with the relations:
\ben
\label{alphaplus}
\beta_k^+=\sum_{l=1}^n \tau_{kl} \alpha_l^+~~,
\een
which are equivalent with \eqref{abmatrix}. The positive polarization
condition $\ast_g^\C \eta_J=-\i \eta_J$ amounts to the relations
$\ast_g^\C \alpha^+_k=-\i \alpha^+_k$ and $\ast_g^\C \beta^+_k=-\i
\beta^+_k$, which say that $\alpha_k^+$ and $\beta_k^+$ have the form:
\be
\alpha_k^+=\alpha_k+\i\ast_g \alpha_k~~,~~\beta_k^+=\beta_k+\i\ast_g \beta_k~~,
\ee
i.e.:
\ben
\label{abplus}
{\hat \alpha}^+={\hat \alpha}+\i \ast_g {\hat \alpha}~~,~~{\hat \beta}^+={\hat \beta}+\i \ast_g {\hat \beta}~~.
\een
Combining \eqref{abmatrix} and \eqref{abplus} shows that the positive
polarization condition for $\eta$ amounts to \eqref{etapol}.
Separating real and imaginary parts, relation \eqref{etapol} amounts
to conditions \eqref{etapol2} and \eqref{etapol3}, which are mutually
equivalent since $\ast_g$ squares to minus the identity. These
conditions amount to the matrix relations:
\be
{\hat \eta}=\left[\begin{array}{cc} 1 & ~~0 \\ \theta
    &~-\gamma\end{array}\right]\left[\begin{array}{c} {\hat \alpha}
    \\ \ast_g {\hat \alpha}\end{array}\right]~~,~~ \ast_g {\hat
  \eta}=\left[\begin{array}{cc} 0 & 1 \\ \gamma
    &\theta \end{array}\right]\left[\begin{array}{c} {\hat
      \alpha}\\ \ast_g {\hat \alpha}\end{array}\right]~~.
\ee 
which in turn give \eqref{etapol3} upon using the relation: 
\be
\left[\begin{array}{cc} 1 &~~0\\ \theta &- \gamma \end{array}\right]^{-1}=
\left[\begin{array}{cc}  1 & ~~0 \\ \gamma^{-1}\theta & ~-\gamma^{-1} \end{array}\right]~~.
\ee
Notice that \eqref{etapol3} also follows directly from the
polarization condition $\ast_g {\hat \eta}=-{\hat J}{\hat \eta}$ using
\eqref{hatJtg} (recall that $\btau=-{\overline \tau}$, which gives
$\btau_R=-\theta$ and $\btau_I=\gamma$).
\end{proof}

\subsection{Interpretation of the fundamental form in the complexified formalism}

Let $(\cS^{\C}, D^{\C}, J^{\C}, \omega^{\C})$ denote the
complexification of $(\cS, D, J, \omega)$, and let $L^{\pm}_{J}\subset
\cS^{\C}$ denote the complex Lagrangian subbundles of
$(\cS^{\C},\omega^{\C})$ defined by:
\begin{equation}
L^{\pm}_{J} = \ker(J^{\C} \mp \i \,\id_{\cS^{\C}}) = \ker P^{\mp}_{J}\subset \cS^{\C}\, ,
\end{equation}
where $P^{\pm}_{J}\colon \cS^{\C}\to L^{\pm}_{J}$ are the corresponding projectors. We have:
\begin{equation}
\cS^{\C} = L^{+}_{J} \oplus L^{-}_{J}\, .
\end{equation}
Direct calculation gives\footnote{For simplicity in the notation
  we will denote $\id_{\Lambda^{k}T^{\ast}N}\otimes P^{\pm}_{J}$
  simply by $P^{\pm}_{J}$.}:
\begin{equation}
\label{eq:DCPCDC}
D^{\C}\circ P^{\pm}_{J} = P^{\pm}_{J}\circ D^{\C} \mp \frac{\i}{2} \Theta\, .
\end{equation}
Consider the induced connections:
\begin{equation}
D^{\C}_{\pm} \eqdef P^{\pm}_{J}\circ D^{\C}|_{L^{\pm}_{J}} \colon \Omega^{0}(L^{\pm}_{J})\to \Omega^{1}(L^{\pm}_{J})\, .
\end{equation}

\begin{Lemma}
\label{lemma:DCTheta}
The following equation holds:
\begin{equation}
\label{eq:DCDCproyection}
D^{\C}(\eta) = D^{\C}_{\pm}(\eta) \mp \frac{i}{2} \Theta(\eta)\, ,
\end{equation}
for every $\eta\in \Gamma(L^{\pm}_{J})$. In addition, $D^{\C}$
preserves $L^{\pm}_{J}$ if and only if $\Theta \circ P^{\pm}_{J} = 0$.
In particular, $D^{\C}$ preserves both $L^{+}_{J}$ and $L^{-}_{J}$ iff
$\Theta = 0$.
\end{Lemma}

\begin{proof}
Equation \eqref{eq:DCDCproyection} follows directly from equation
\eqref{eq:DCPCDC}. If $D^{\C}$ preserves $L^{\pm}_{J}$ then we have:
\begin{equation}
\Theta\circ P^{\pm}_{J} = D^{\C}\circ J\circ P^{\pm}_{J} -J\circ D^{\C}\circ P^{\pm}_{J} = 0\, ,
\end{equation} 
since $J\circ D^{\C}\circ P^{\pm}_{J} = \pm i\, D^{\C}\circ
P^{\pm}_{J}\colon \Omega^{0}(L^{\pm}_{J})\to \Omega^{1}(L^{\pm}_{J})$
and $ D^{\C}\circ J\circ P^{\pm}_{J} = \pm i D^{\C}\circ
P^{\pm}_{J}\colon \Omega^{0}(L^{\pm}_{J})\to
\Omega^{1}(L^{\pm}_{J})$. On the other hand, if $\Theta \circ
P^{\pm}_{J} = 0$ then $D^{\C}\circ P^{\pm}_{J} = D^{\C}_{\pm}$ and we
conclude.
\end{proof}

\begin{Remark}

\end{Remark}

\noindent
Although $D^{\C}$ is a flat connection on $\cS^{\C}$, $D^{\C}_{\pm}$
need not be a flat connection on $L^{\pm}_{J}$.
\begin{Proposition}
Let $\mathcal{R}_{D^{\C}_{\pm}}\colon \Omega^{0}(L^{\pm}_{J})\to
\Omega^{2}(L^{\pm}_{J})$ be the curvature endomorphism of
$D^{\C}_{\pm}$. The following formula holds:
\begin{equation}
\label{eq:curvatureTheta}
\mathcal{R}_{D^{\C}_{\pm}} = \pm \frac{\i}{2} P^{\pm}_{J}\circ \dd_{D^{\C}}\circ \Theta\, .
\end{equation}
In particular, if $D^{\C}$ preserves $L^{\pm}_{J}$ then:
\begin{equation}
\label{eq:RDpm}
\mathcal{R}_{D^{\C}_{\pm}} = 0\, ,
\end{equation}
and thus $(L^{\pm}_{J},  D^{\C}_{\pm})$ is a flat Lagrangian subbundle of $\cS^{\C}$.
\end{Proposition}

\begin{proof}
Composing $P^{\pm}_{J}\circ \dd_{D^{\C}}$ with both sides of equation
\eqref{eq:DCDCproyection} and using the fact that $D^{\C}$ is a flat
connection gives \eqref{eq:curvatureTheta} follows. Equation
\eqref{eq:RDpm} follows from lemma \ref{lemma:DCTheta}.
\end{proof}

\section{Some twisted constructions}
\label{app:twistedconstructions}

\subsection{Twisted cohomology}

Let $(\cS,D)$ be a flat vector bundle of rank $r$ defined over a
manifold $N$ of dimension $d$, where $D:\Omega^0(N)\rightarrow
\Omega^1(N,\cS)$ is the flat connection. The exterior bundle $\wedge_N
\eqdef \wedge T^\ast N=\oplus_{k=0}^d \wedge^k T^\ast N$ is a bundle
of $\Z$-graded unital associative and supercommutative algebras when
endowed with the fiberwise multiplication given by the wedge
product. Let:
\be
\alpha \eqdef \oplus_{k=0}^d(-1)^k \id_{\wedge^k T^\ast N}
\ee
be the {\em main automorphism} of this bundle of algebras, namely the
unique unital automorphism which satisfies $\alpha|_{T^\ast
  N}=-\id_{T^\ast N}$. Let $\dd_{D}:\Omega(N,\cS)\rightarrow
\Omega(N,\cS)$ be the de Rham differential twisted by the flat
connection $D$, where $\Omega(N,\cS)=\Gamma(\wedge_N,\cS)$.  For any
$\rho\in \Omega(N)$ and $\xi\in \Gamma(N,\cS)$, we have:
\ben
\label{dDdef}
\dd_{D}(\rho\otimes \xi)=(\dd \rho)\otimes \xi+\alpha(\rho)\wedge (D \xi)~~.
\een

\begin{Definition}
For any $k=0\ldots d$, the $k$-th {\em $(\cS,D)$-twisted de Rham
  cohomology space of $N$} is the $\R$-vector space defined through:
\be
H^k_{\dd_D}(N,\cS)\eqdef \ker(\dd_D:\Omega^k(N,\cS)\rightarrow \Omega^{k+1}(N,\cS))/\im (\dd_D:\Omega^{k-1}(N,\cS)\rightarrow \Omega^k(N,\cS))~~.
\ee
The {\em total twisted de Rham cohomology space of $N$} is the
$\Z$-graded vector space defined through $H_{\dd_D}(N,\cS)\eqdef
\oplus_{k=0}^d H_{\dd_D}(N,\cS)$.
\end{Definition}

\noindent Let $\fS$ be the sheaf of smooth flat sections of
$(\cS,D)$, namely the sheaf of vector spaces defined through:
\be
\fS(U)\eqdef \{s\in \Gamma(U,\cS)|Ds=0\}~~
\ee
for any open subset $U\subset N$, with the obvious restriction
maps. This is a locally-constant sheaf of rank $2n$, whose stalk is
isomorphic with the typical fiber $\cS_0$ of $\cS$. Let $\Pi_1(N)$ be
the first homotopy groupoid of $N$.  The {\em local system defined by
  $(\cS,D)$} is the functor $T:\Pi_1(N)\rightarrow \Gp^\times$ which
sends a point $n\in N$ to the group $T(x)\eqdef
\Aut_\R(\cS_x,\omega_x)$ and the homotopy class (with fixed endpoints)
of a piecewise-smooth path $c$ from $x_1$ to $x_2$ in $N$ into the
parallel transport $T(c):\cS_{x_1}\rightarrow \cS_{x_2}$ defined by
$D$ along $c$. For any $x\in N$, we have
$\pi_1(N,x)=\Hom_{\Pi_1(N)}(x,x)$ and the holonomy representation
$\rho_x=\hol_D(x)$ of $D$ at $x$ is given by
$\rho_x=T|_{\pi_1(N,x)}:\pi_1(N,x)\rightarrow \Aut_\R(\cS_x)$. There
exist natural isomorphisms of graded vector spaces:
\be
H^k_{\dd_D}(N,\cS)\simeq H^k(\fS)\simeq H^k(N,T)~~\forall k=0\ldots d~~,
\ee
where $H^k(\fS)$ denote sheaf cohomology of $\fS$ and $H^k(N,T)$
denotes singular cohomology with coefficients in the local system $T$
(see \cite{Whitehead}). Using these isomorphisms, we identify the
three vector spaces above and denote them by $H^k(N,\cS)$, which we
call the $k$-th {\em twisted cohomology space defined by $(\cS,D)$.}
On the other hand, the de Rham isomorphism identifies the de Rham
cohomology $H_\dd(N)$ with the singular cohomology $H(N,\R)$, which we
denote by $H(N)$.

Let ${\tilde N}$ be the universal covering space of $N$ and
$\pi:{\tilde N}\rightarrow N$ be the projection. The complex
$(\Omega({\tilde N})\otimes \cS_x,\dd\otimes \id_{\cS_x})$ is
equivariant under the action of $\pi_1(N,x)$ given by tensoring the
pullback through deck transformations with $\rho_x$. Then $H^k(N,\cS)$
isomorphic with the $k$-th cohomology space of the invariant
sub-complex $(\Omega({\tilde N})\otimes \cS_x)^{\pi_1(N,x)}$.

\begin{Remark}
The exists a similar statement for singular cohomology valued in a
local system $T:\Pi_1(N)\rightarrow \Ab$, giving an isomorphism:
\be
H^k(N,T)\simeq E^k({\tilde N},T_x)~~,
\ee
where $E^k({\tilde N},T_x)$ denotes the $\pi_1(N,x)$-equivariant
singular cohomology of ${\tilde N}$ valued in the group $T_x\eqdef
T(x)$. Here $\pi_1(N,x)$ acts on $T_x$ through the holonomy
representation $\rho_x\eqdef T|_{\pi_1(N,x)}:\pi_1(N,x)\rightarrow
\Aut_\Gp(T_x)$ of $T$ at $x$, which is defined through
$\rho_x(c)=T(c)$ for all $c\in \pi_1(N,x)$. See \cite[Chap VI, page
  281]{Whitehead}.
\end{Remark}

\subsection{Twisted wedge product and cup product}
\label{app:twistedwedge}

Let $(\cS,\omega)$ be a symplectic vector bundle defined over a
manifold $N$. Let $\wedge_N(\cS)\eqdef \wedge_N \otimes \cS=\wedge
T^\ast N\otimes \cS$.

\begin{Definition} 
The {\em $(\cS,\omega)$-twisted wedge product} is the
$\wedge_N$-valued fiberwise bilinear pairing
$\bwedge_\omega:\wedge_N(\cS)\times_N\wedge_N(\cS)\rightarrow
\wedge_N$ which is determined uniquely by the condition (we write
$\bwedge_\omega$ in infix notation):
\ben
\label{bwedge}
(\rho_1\otimes \xi_1) \bwedge_{\omega} (\rho_2\otimes \xi_2)=\alpha(\rho_2)\omega(\xi_1,\xi_2) \rho_1\wedge \rho_2
\een
for all $\rho_1,\rho_2\in \Gamma(N,\wedge_N)=\Omega(N)$ and all $\xi_1,\xi_2\in \Gamma(N,\cS)$.
\end{Definition}

\noindent The sign prefactor in the right hand side of \eqref{bwedge}
corresponds to viewing $\cS$ as the graded vector bundle $\cS[-1]$
concentrated in degree one. Accordingly, the bundle
$\wedge_N(\cS[-1])\eqdef \wedge_N \otimes \cS[-1]$ is graded
with homogeneous components:
\be
\wedge_N^k (\cS[-1])=\wedge^{k-1} T^\ast N\otimes \cS~~
\ee
and the map $\bwedge_\omega:\wedge_N(\cS[-1])\otimes
(\wedge_N\cS[-1])\rightarrow \wedge T^\ast N$ is homogeneous of degree
$-2$. The space of $\cS$-valued differential forms
$\Omega(N,\cS)=\Gamma(N,\wedge_N(\cS))$ admits the rank grading, with
homogeneous components:
\be
\Omega^k(N,\cS)=\Gamma(N,\wedge_N^k(\cS))\, ,
\ee 
as well as the grading induced from $\wedge_N(\cS[-1])$:
\be
\Omega^k(N,\cS[-1])\eqdef \Gamma(N,\wedge_N^k(\cS[-1]))=\Omega^{k-1}(N,\cS)~~.
\ee
We have $\Omega(N,\cS[-1])=\Omega(N,\cS)[-1]$. Let
$\deg\eta=\rk\eta+1$ denote the degree of a pure rank $\cS$-valued
form with respect to the second grading. The pairing $\bwedge_\omega$
is graded-commutative with respect to the grading on
$\Omega(N,\cS)[-1]$, i.e. the following relation holds for all pure rank
bundle-valued forms $\eta_1,\eta_2\in \Omega(N,\cS)$:
\be
\eta_1\bwedge_\omega \eta_2=(-1)^{\deg\eta_1\deg \eta_2}\eta_2\bwedge_\omega\eta_1~~.
\ee
Let:
\be
\balpha\eqdef \alpha\otimes (-\id_\cS)=\oplus_{k=0}^d (-1)^{k+1} \id_{\wedge^k_N(\cS)}\in \End(N,\wedge_N(\cS))\, ,
\ee
be the endomorphism of $\wedge_N(\cS)$ given by:
\be
\balpha (\eta)=(-1)^{\deg \eta} \eta 
\ee
for any pure rank bundle-valued form $\eta\in \Omega(N,\cS)$. 

Let $(\cS,D,\omega)$ be a flat symplectic vector bundle over
$N$. Using \eqref{dDdef}, \eqref{bwedge} and compatibility of $D$
with $\omega$, one finds that $\bwedge_\omega$ induces a morphism of
complexes from $\Omega(N,\cS)[-1]\otimes \Omega(N,\cS)[-1]$ to
$\Omega(N)$, i.e. the following identity holds for any
$\eta_1,\eta_2\in \Omega(N,\cS)$:
\ben
\label{dwedgegen}
\dd (\eta_1\bwedge_\omega\eta_2)=(\dd_{D}\eta_1)\bwedge_\omega \eta_2+\balpha(\eta_1)\bwedge_\omega (\dd_{D}\eta_2)~~.
\een 
Thus $\bwedge_\omega$ descents to an $\R$-bilinear map from the twisted 
cohomology $H(N,\cS)$ to the ordinary cohomology $H(N)$.

\begin{Definition} 
Let $(\cS,D,\omega)$ be a flat symplectic vector bundle over $N$.  The
{\em $(\cS,D,\omega)$-twisted cup product} is the $\R$-bilinear
graded-symmetric pairing:
\be
\boldsymbol{\cup}_\omega:H(N,\cS)[-1]\times H(N,\cS)[-1]\rightarrow H(N)
\ee
induced by $\bwedge_\omega$, which is homogeneous of degree $-2$.
\end{Definition}

\begin{Remark}
The twisted cup product can also be defined using local systems valued
in the groupoid $\Vect^\sp$ of symplectic finite-dimensional
$\R$-vector spaces and admits a version for local systems valued in
the groupoid $\Vect_0^\sp$ of integral finite-dimensional symplectic
vector spaces over $\R$ which was mentioned in Subsection
\ref{sec:semiclassical}. When $N$ is compact and oriented, a choice of
orientation gives an isomorphism $H^d(N)\simeq \R$ by pairing with the
fundamental class $[N]\in H(N)$. In that case, the twisted cup
product induces the {\em twisted Kronecker pairing}
$\cK:H(N,\cS)[-1]\times H(N,\cS)[-1] \rightarrow \R$, which is defined
through:
\be
\cK(\alpha,\beta)=(\alpha\boldsymbol{\cup}_\omega\beta)(N)~~.
\ee
This $\R$-bilinear pairing is graded-symmetric and homogeneous of
degree $-d-2$.  
\end{Remark}

\subsection{Twisted Hodge operator, twisted codifferential and twisted d'Alembert operator}

Let $(\cS,D,J,\omega)$ be an electromagnetic structure on $N$ and $g$
be a pseudo-Riemannian metric of signature $(p,q)$ on $N$. Let
$Q(\cdot,\cdot) = \omega(J\cdot,\cdot)$ denote the Euclidean scalar product induced by $J$ and
$\omega$ on $\cS$, which acts on sections as:
\be
Q(\xi_1,\xi_2)=\omega(J\xi_1,\xi_2)=-\omega(\xi_1,J\xi_2)~~\forall \xi_1, \xi_2\in \Gamma(N,\cS)~~.
\ee
Let $(~,~)_g$ be the pseudo-Euclidean scalar product induced by $g$
on the vector bundle $\wedge_N=\wedge T^\ast N$, with respect to which
the rank decomposition of $\wedge T^\ast N$ is orthogonal and the
normalized volume form $\nu$ of $(N,g)$ satisfies
$(\nu,\nu)_{g}=(-1)^q$.  Together with $Q$, this induces
a pseudo-Euclidean scalar product $(~,~)$ on the vector bundle
$\wedge_N(\cS)$, which is uniquely determined by the condition:
\be
(\rho_1\otimes \xi_1,\rho_2\otimes \xi_2)=\delta_{k_1,k_2}(-1)^{k_1}Q(\xi_1,\xi_2) (\rho_1, \rho_2)_{g}~~.
\ee 
for all $\rho_1\in \Omega^{k_1}(N)$, $\rho_2\in \Omega^{k_2}(N)$ and
$\xi_1,\xi_2\in \Gamma(N,\cS)$. The Hodge operator $\ast:=\ast_{g}$ of
$(N,g)$ induces the endomorphism $\ast\eqdef \ast \otimes \id_{\cS}$
of the vector bundle $\wedge_N(\cS)$. Recall that the $J$-twisted
Hodge operator $\star \in \End(N,\wedge_N(\cS))$ is defined through:
\be
\star\eqdef\ast\otimes J=\ast \circ J=J\circ \ast~~.
\ee
The relation $\rho_1\wedge (\ast \rho_2)=(\rho_1,\rho_2)_{g} \nu$
implies:
\ben
\label{starwedgegen}
\eta_1 \bwedge_\omega (\star \eta_2)= \eta_2\bwedge_\omega (\star \eta_1)=(\eta_1,\eta_2) \nu~~\forall \eta_1,\eta_2\in \Omega(N,\cS)~~,
\een
while the relation $\ast^2=(-1)^q\alpha^{d-1}$ implies:
\ben
\label{starsquaregen}
\star^2=(-1)^{p-1}\balpha^{d-1}~~.
\een 
Let $\langle~,~\rangle$ denote the $\R$-valued symmetric
non-degenerate pairing defined on $\Omega_c(N,\cS)$ through:
\be
\langle \eta_1,\eta_2\rangle\eqdef \int_N (\eta_1,\eta_2)\nu=\int_N \eta_1 \bwedge_\omega (\star \eta_2)~~.
\ee
Relations \eqref{starsquaregen} implies that the formal adjoint of
$\dd_{D}$ with respect to this pairing (which we shall call the {\em
  twisted codifferential}) is given by:
\be
\updelta_D = \star \dd_{D}\star~~.
\ee

\begin{Definition}
The {\em twisted d'Alembert operator} defined by $(\cS,D,\omega)$ is
the second order differential operator
$\Box_{D}:\Omega(N,\cS)\rightarrow \Omega(N,\cS)$ defined through:
\be
\Box_{D}\eqdef
\dd_{D}\updelta_{D}+\updelta_D \dd_{D}~~.
\ee
\end{Definition}

\

\noindent When $d=4$ and $g$ is Lorentzian, we have: 
\ben
\label{HodgeSquare}
\ast^2=-\alpha~~\mathrm{and}~~\star^2=\balpha~~.
\een
In this case, the twisted d'Alembert operator is normally hyperbolic. 

\section{Unbased automorphisms of vector bundles} 
\label{app:unbased}

Let $\cM$ be a manifold. In this appendix, we discuss certain
constructions involving unbased automorphisms of vector bundles $\cS$
defined over $\cM$, such as the twisted push-forward of $\cS$-valued
vector fields and the twisted pull-back of $\cS$-valued forms. We also
prove some properties of these constructions which are used in Section
\ref{sec:duality}.

\subsection{Unbased morphisms of vector  bundles over $\cM$}

Let $\cS,\cS'$ be vector bundles over $\cM$ with projections
$\sigma:\cS\rightarrow \cM$ and $\sigma':\cS'\rightarrow \cM$. Recall
that an unbased morphism of vector bundles from $\cS$ to $\cS'$ is a
smooth map $f:\cS\rightarrow \cS'$ such that there exists a smooth map
$f_0:\cM\rightarrow \cM$ for which the following diagram commutes:
\be
\scalebox{1.2}{
\xymatrix{
\cS \ar[d]_{\sigma} \ar[r]^{f} & ~\cS' \ar[d]^{\sigma'} \\
\cM  \ar[r]^{f_0} & \cM \\
}}
\ee
and such that the restriction $f_p\eqdef f|_{\cS_p}:\cS_p\rightarrow
\cS'_{f_0(p)}$ is linear for all $p\in \cM$. Since $\sigma$ is
surjective, the map $f_0$ is uniquely determined by $f$, being called
its {\em projection} to $\cM$, while $f$ is called a {\em lift} of
$f_0$.  We say that $f$ is {\em based} if $f_0=\id_\cM$. Vector
bundles over $\cM$ and unbased morphisms between such form a category
denoted $\VB^\ub(\cM)$, which contains the category $\VB(\cM)$ of
vector bundles and based morphisms as a non-full subcategory.  Let
$\mathrm{B}_\cM$ be the category with a single object $\cM$ and
morphisms given by smooth maps from $\cM$ to $\cM$. The functor
$\VB^\ub(\cM)\rightarrow \mathrm{B}_\cM$ which takes $\cS$ into $\cM$
and $f$ into $f_0$ has fiber at $\cM$ given by $\VB(\cM)$. Hence the
map $\pi_\cS:\Aut^\ub(\cS)\rightarrow \Diff(\cM)$ given by
$\pi_\cS(f)\eqdef f_0$ is a morphism of groups. Let
$\Hom^\ub(\cS,\cS')$ and $\Isom^\ub(\cS,\cS')$ denote the sets of
unbased morphisms and isomorphisms from $\cS$ to $\cS'$. When
$\cS'=\cS$, we set $\End^\ub(\cS)\eqdef\Hom^\ub(\cS,\cS)$ and
$\Aut^\ub(\cS)\eqdef \Isom^\ub(\cS,\cS)$. The functor to $B_\cM$ gives
a decomposition $\Hom^\ub(\cS,\cS')=\sqcup_{\psi\in
  \cC^\infty(\cM,\cM)}\Hom^\ub_\psi(\cS,\cS')$, where:
\be
\Hom^\ub_\psi(\cS,\cS')\eqdef \{f\in \Hom^\ub(\cS,\cS')|f_0=\psi\}~~,~~\forall \psi\in \cC^\infty(\cM,\cM)~~.
\ee
We have a similar decomposition $\Isom^\ub(\cS,\cS')=\sqcup_{\psi\in
  \Diff(\cM)}\Isom_\psi(\cS,\cS')$, with
$\Isom^\ub_{\id_\cM}(\cS,\cS')=\Isom(\cS,\cS')$.

\subsection{The action of unbased automorphisms on ordinary sections}
\label{subsec:ordsecaction}

Recall that a {\em (smooth) section of $\cS$ along a map} $\psi\in
\cC^\infty(\cM,\cM)$ is a smooth map $\zeta:\cM\rightarrow \cS$ such
that $\sigma\circ \zeta=\psi$, which amounts to the condition
$\zeta_p\in \cS_{\psi(p)}$ for all $p\in \cM$.  Let
$\Gamma_\psi(\cM,\cS)$ denote the $\cC^\infty(\cM,\R)$-module of
sections of $\cS$ along $\psi$, where outer multiplication with
$\alpha\in \cC^\infty(\cM,\R)$ is defined through:
\be
(\alpha\zeta)_p\eqdef (\alpha\circ \psi)(p)\zeta_p\in \cS_{\psi(p)}~~.
\ee
Ordinary sections of $\cS$ are sections along the identity map
$\id_\cM$, thus $\Gamma(\cM,\cS)=\Gamma_{\id_\cM}(\cM,\cS)$.
Composition with $\psi$ from the right gives a morphism of
$\cC^\infty(\cM,\R)$-modules:
\be
\Gamma(\cM,\cS)\ni \xi\rightarrow \xi_\psi\eqdef \xi\circ \psi\in \Gamma_\psi(\cM,\cS)~~,
\ee
where: 
\be
\xi_\psi(p)=\xi(\psi(p))\in \cS_{\psi(p)}~~,~~\forall p\in \cM~~.
\ee
When $\psi\in \Diff(\cM)$, this is an isomorphism of modules whose
inverse is given by:
\be
\Gamma_\psi(\cM,\cS)\ni \zeta \rightarrow ~^\psi\zeta\eqdef \zeta\circ \psi^{-1}\in \Gamma(\cM,\cS)~~,
\ee
i.e.:
\be
({}^\psi\zeta)(p)=\zeta_{\psi^{-1}(p)}\in \cS_p~~,~~\forall p\in \cM~~.
\ee
We have $~^\psi(\xi_\psi)=\xi$ and $({}^\psi \zeta)_\psi=\zeta$. 

An unbased morphism $f\in \Hom^\ub(\cS,\cS')$ induces a map
$\Gamma(\cM,\cS)\rightarrow \Gamma_{f_0}(M,\cS')$ (which we again
denote by $f$), defined through:
\be
f(\xi)\eqdef f\circ \xi~~,~~\forall \xi\in \Gamma(\cM,\cS)~~,
\ee
i.e.:
\be
f(\xi)_p \eqdef f_p(\xi_p)\in \cS'_{f_0(p)}~~,~~\forall \xi\in \Gamma(\cM,\cS)~~\forall p\in \cM~~.
\ee
This map takes ordinary sections of $\cS$ into sections of $\cS$ along
the map $f_0$. For any $\alpha\in \cC^\infty(\cM,\R)$, we have:
\be
f((\alpha\circ f_0)\xi)=\alpha f(\xi)~~.
\ee 
For many purposes, it is convenient to work with with another map
induced by $f$, which takes ordinary sections of $\cS$ into ordinary
sections.

\begin{Definition}
The {\em action of an unbased automorphism $f\in \Aut^\ub(\cS)$ on
  ordinary sections of $\cS$} is the map $\mf\in
\Aut_\R(\Gamma(\cM,\cS))$ defined through:
\ben
\label{mfdef}
\mf(\xi)\eqdef ~^{f_0}(f(\xi))=f\circ \xi\circ f_0^{-1}~~.
\een
\end{Definition}

\

\noindent For any $p\in \cM$, we have:
\ben
\label{mfp}
\mf(\xi)_p=f_{f_0^{-1}(p)}(\xi_{f_0^{-1}(p)})\in \cS_p~~,
\een
which also reads: 
\ben
\label{mfp2}
\mf(\xi)_{f_0(p)}=f_p(\xi_p)\in \cS_{f_0(p)}~~.
\een

\begin{Proposition}
The map $\Aut^\ub(\cS)\ni f\rightarrow \mf\in
\Aut_\R(\Gamma(\cM,\cS))$ is a morphism of groups. Moreover, for any
$f\in \Aut^\ub(\cS)$, we have:
\begin{enumerate}[1.]
\itemsep 0.0em
\item The following relation holds for any $\xi\in \Gamma(\cM,\cS)$
  and any $\alpha\in \cC^\infty(\cM,\R)$:
\ben
\label{mfmod}
\mf(\alpha\xi)=(\alpha\circ f_0^{-1}) \mf(\xi)~~.
\een
In particular, $\mf$ is {\em not} an automorphism of the
$\cC^\infty(\cM,\R)$-module $\Gamma(\cM,\R)$~.
\item For any $\xi\in \Gamma(\cM,\cS)$ and any $p\in \cM$, we have:
\ben
\label{mfInv}
\mf^{-1}(\xi)_p=(f^{-1})(\xi)_{f_0(p)}=(f_p)^{-1}(\xi_{f_0(p)})~~.
\een
\end{enumerate}
\end{Proposition}

\begin{proof}
The fact that the map which takes $f$ into $\mf$ is a morphism of
groups follows from \eqref{mfdef}. Relation \eqref{mfmod} follows by
direct computation using \eqref{mfp}. Relation \eqref{mfInv} follows
from \eqref{mfp} using the relation:
\be
(f^{-1})_p=(f_{f_0^{-1}(p)})^{-1}\in \Hom(\cS_p,\cS_{f_0^{-1}(p)})~~.
\ee
\end{proof}

\subsection{The unbased morphism from the pull-back}

Let $\cS$ be a vector bundle over $\cM$.  For any $\psi\in
\cC^\infty(\cM,\cM)$, the pulled-back bundle $\cS^\psi$ has total
space given by the fiber product $\cS^\psi=\cM\times_{\psi,\sigma}
\cS=\{(p,s)\in \cM\times \cS|\psi(p)=\sigma(s)\}=\{(p,s)|p\in \cM ~\&~
s\in \cS_{\psi(p)}\}$ and projection $\sigma^\psi$ given by
$\sigma^\psi(p,s)=p$. Let $\Phi_\cS(\psi):\cS^\psi\rightarrow \cS$ be
the projection of the fiber product on the second factor:
\be
\Phi_\cS(\psi)(p,s)\eqdef s~~.
\ee
Then $\sigma\circ \Phi_\cS(\psi)=\psi\circ \sigma^\psi$, thus
$\Phi_\cS(\psi)$ is an unbased morphism $\Phi_\cS(\psi)\in
\Hom^\ub_\psi(\cS^\psi, \cS)$ whose projection is given by
$\Phi_\cS(\psi)_0=\psi$ and whose restriction to the fibers is a
linear isomorphism. We have a commutative diagram:
\be
\scalebox{1.2}{
\xymatrix{
\cS^\psi \ar[d]_{\sigma^\psi} \ar[r]^{~~\Phi_\cS(\psi)} & ~~\cS \ar[d]^{\sigma} \\
\cM  \ar[r]^{\psi } & \cM \\
}}
\ee
Given any section $\xi\in \Gamma(\cM,\cS)$, the pull-back section
$\xi^\psi\in \Gamma(\cM,\cS^\psi)$ is uniquely determined by the
condition:
\be
\Phi_\cS(\psi)\circ \xi^\psi=\xi_\psi\in \Gamma_\psi(\cM,\cS)~~,
\ee
i.e.:
\ben
\label{Phipsi}
\xi_\psi=\Phi_\cS(\psi)(\xi^\psi)~~.
\een
Suppose that $\psi\in \Diff(\cM)$. Then $\Phi_\cS(\psi)\in
\Isom^\ub_\psi(\cS^\psi,\cS)$ and \eqref{Phipsi} gives:
\ben
\label{psiast}
\xi^\psi=\Phi_\cS(\psi)^{-1}(\xi_\psi)~~,~~\forall \xi \in \Gamma(\cM,\cS)
\een
i.e.:
\ben
\label{Phizeta}
~^\psi\zeta=[\Phi_\cS(\psi)^{-1}(\zeta)]^{\psi^{-1}}~~,~~\forall \zeta\in \Gamma_\psi(\cM,\cS)~~.
\een

\noindent Given $\psi_1,\psi_2\in \Diff(\cM)$, we identify
$(\cS^{\psi_1})^{\psi_2}$ with $\cS^{\psi_1\psi_2}$.  We also identify
$\cS^{\id_\cM}$ with $\cS$. The proof of the following statement is
immediate:

\begin{Proposition}
\label{prop:PhiComp}
We have $\Phi_\cS(\id_\cM)=\id_{\cS}$. Moreover, the following relations hold:
\begin{enumerate}[1.]
\itemsep 0.0em
\item For any $\psi_1,\psi_2\in \Diff(\cM)$, we have: 
\be
\Phi_\cS(\psi_1\psi_2)=\Phi_\cS(\psi_1)\circ \Phi_{\cS^{\psi_1}}(\psi_2)~~.
\ee
\item For any $\psi,\theta\in \Diff(\cM)$ and any $h\in \Hom(\cS,\cS^\psi)$, we have: 
\be
h^{\theta}=\Phi_{\cS^\psi}(\theta)^{-1}\circ h\circ \Phi_\cS(\theta)~~.
\ee
\end{enumerate}
\end{Proposition}

\subsection{The based isomorphism induced by an unbased automorphism}
\label{subsec:basedinduced}

\begin{Definition}
Let $f\in \Aut^\ub(\cS)$ be an unbased automorphism of $\cS$. 
The {\em based isomorphism induced by $f$} is the based
isomorphism of vector bundles defined through:
\ben
\label{hatf}
{\hat f}\eqdef \Phi_{\cS}(f_0)^{-1}\circ f\in \Isom(\cS,\cS^{f_0})~~.
\een
\end{Definition}

\

\noindent Thus ${\hat f}$ is the unique based isomorphism of vector bundles 
which makes the following diagram commute in the category $\VB^\ub(\cM)$:
\ben
\label{diagram:hatf}
\scalebox{1.2}{
\xymatrixcolsep{5pc}\xymatrix{
\cS \ar[d]_{{\hat f}} \ar[dr]^{~~f} & \\
\cS^{f_0}  \ar[r]^{\Phi_\cS(f_0)} & \cS \\
}}
\een 

\begin{Proposition}
For any $f\in \Aut^\ub(\cS)$ and any $\xi\in \Gamma(\cM,\cS)$, we have:
\beqan
&& \mf(\xi)=[{\hat f}(\xi)]^{f_0^{-1}}~~\label{mf}\\
&& \mf^{-1}(\xi)=f^{-1}(\xi_{f_0})={\hat f}^{-1}(\xi^{f_0})~~.\label{mfinv}
\eeqan 
\end{Proposition}

\begin{proof}
For any section $\xi\in \Gamma(\cM,\cS)$, relation \eqref{hatf} gives
$f(\xi)=\Phi_\cS(f_0)({\hat f}(\xi))\in \Gamma_{f_0}(\cM,\cS)$, where
${\hat f}(\xi)\in \Gamma(\cM,\cS^{f_0})$. Using \eqref{mfdef} and
\eqref{Phizeta}, this gives:
\be
\mf(\xi)=~^{f_0}(f(\xi))=[\Phi_\cS(f_0)^{-1}(f(\xi))]^{f_0^{-1}}=[{\hat f}(\xi)]^{f_0^{-1}}~~,
\ee
which shows that \eqref{mf} holds. Relation \eqref{mf} gives
$\xi={\hat f}^{-1}(\mf(\xi)^{f_0})$. Replacing $\xi$ with
$\mf^{-1}(\xi)$ in this relation, we obtain:
\be
\mf^{-1}(\xi)={\hat f}^{-1}(\xi^{f_0})~~.
\ee
On the other hand, \eqref{hatf} implies $f^{-1}={\hat f}^{-1}\circ
\Phi_\cS(f_0)^{-1}$. Combining this with \eqref{Phipsi} gives:
\be
f^{-1}(\xi_{f_0})=f^{-1}(\Phi_\cS(f_0)(\xi^{f_0}))={\hat f}^{-1}(\xi^{f_0})~~.
\ee
Combining the two relations above gives \eqref{mfinv}.
\end{proof}

\noindent Given $f_1,f_2\in \Aut^\ub(\cS)$, Proposition
\eqref{prop:PhiComp} implies:
\ben
\label{hatfcomp}
\widehat{f_1f_2}={\hat f}_1^{f_{20}}\circ {\hat f_2}~~\mathrm{and}~~\widehat{\id_\cM}=\id_\cM~~~.
\een

\subsection{Relation to twisted automorphisms}

Let $\cS$ be a vector bundle over $\cM$. 

\begin{Definition}
A {\em twisted automorphism} of $\cS$ is a pair $(\psi,h)$, where
$\psi\in \Diff(\cM)$ and $h \in \Isom(\cS,\cS^\psi)$ is a {\em based}
isomorphism from $\cS$ to $\cS^\psi$.
\end{Definition}

\noindent The set:
\be
\Aut^\tw(\cS)\eqdef \sqcup_{\psi \in \Diff(\cM)}\Isom(\cS,\cS^\psi)=\{(\psi,h)|\psi\in \Diff(\cM)~\&~h\in \Isom(\cS,\cS^\psi)\}
\ee
of all twisted automorphisms of $\cS$ becomes a group when endowed with the
multiplication defined through:
\be
(\psi_1, f_1)(\psi_2,f_2)\eqdef (\psi_1\circ \psi_2, f_1^{\psi_2}\circ f_2)~~,
\ee
the unit element being given by $(\id_\cM,\id_\cS)$. The inverse of
$(\psi,f)\in \Aut^\tw(\cS)$ is given by:
\be
(\psi,f)^{-1}=(\psi^{-1}, (f^{-1})^{\psi^{-1}})~~.
\ee
Consider the map: 
\be
\mu_\cS:\Aut^\ub(\cS) \rightarrow \Aut^\tw(\cS)
\ee
which associates to an unbased automorphism $f$ of $\cS$ the pair
$\mu_\cS(f)=(f_0, {\hat f})\in \Aut^\tw(\cS)$, where ${\hat f}\in
\Isom(\cS,\cS^{f_0})$ is the based isomorphism induced by $f$. 

\begin{Proposition}
The map $\mu_\cS:\Aut^\ub(\cS)\rightarrow \Aut^\tw(\cS)$ is an
isomorphism of groups.
\end{Proposition}

\begin{proof}
The fact that $\mu_\cS$ is a morphism of groups follows from
\eqref{hatfcomp}. The fact that $\mu_\cS$ is bijective is clear from
the definition of ${\hat f}$ as the unique based isomorphism which
makes diagram \eqref{diagram:hatf} commute.
\end{proof}

\noindent In particular, an unbased automorphism $f$ of $\cS$ can be
identified with the twisted automorphism $(f_0, {\hat f})$.
Let $\Diff^\cS(\cM)$ be the subgroup of $\Diff(\cM)$ consisting of
those diffeomorphisms $\psi$ of $\cM$ for which there exists a {\em
  based} isomorphism between $\cS^\psi$ and $\cS$: 
\be
\Diff^\cS(\cM)\eqdef\{\psi\in \Diff(\cM)|\Isom(\cS^\psi,\cS)\neq \emptyset\}~~.
\ee
We have a short exact sequence: 
\ben
\label{eseq}
1\longrightarrow \Aut(\cS)\hookrightarrow \Aut^\tw(\cS)\longrightarrow \Diff^\cS(\cM)\longrightarrow 1~~,
\een
where the first map takes $h\in \Aut(\cS)$ into the twisted
automorphism $(\id_\cM,h)$ while the second map takes a twisted
automorphism $(\psi,h)$ into the diffeomorphism $\psi\in \Diff(\cM)$.
Through the isomorphism $\mu_\cS$, this induces a short exact sequence:
\be
1\longrightarrow \Aut(\cS)\hookrightarrow \Aut^\ub(\cS)\longrightarrow \Diff^\cS(\cM)\rightarrow 1~~.
\ee

\subsection{The action of unbased automorphisms on ordinary sections of the bundle of endomorphisms}

Let $\cS,\cS'$ be vector bundles over $\cM$ and $f\in
\Isom^\ub(\cS,\cS')$ be an unbased isomorphism from $\cS$ to $\cS'$.

\begin{Definition}
The {\em adjoint} of $f$ is the unbased isomorphism $\Ad(f)\in
\Isom^\ub(End(\cS), End(\cS'))$ whose fiber
$\Ad(f)_p:\End(\cS_p)\rightarrow \End(\cS'_{f_0(p)})$ at any $p\in
\cM$ is given by:
\be
\Ad(f)_p(t)=f_p \circ t \circ (f_p)^{-1}\in \End(\cS'_{f_0(p)})~~\forall t\in \End(\cS_p)~~.
\ee 
\end{Definition}

\

\noindent We have $\Ad(f)_0=f_0$. When $\cS'=\cS$, the map
$\Ad:\Aut^\ub(\cS)\rightarrow \Aut^\ub(End(\cS))$ is a morphism of
groups which satisfies $\pi_{End(\cS)}\circ \Ad=\pi_\cS$.  For any
$\psi\in \Diff(\cS)$, we have:
\ben
\label{PhiEnd}
\Phi_{End(\cS)}(\psi)=\Ad(\Phi_\cS(\psi))~~,
\een
where we identify $\End(\cS)^\psi\equiv \End(\cS^\psi)$. 

\begin{Proposition}
For any $f\in \Aut^\ub(\cS)$, we have:
\ben
\label{hatwidehat}
\widehat{\Ad(f)}=\Ad({\hat f})\in \Isom(End(\cS), End(\cS)^{f_0})\equiv  \Isom(End(\cS), End(\cS^{f_0}))~~,
\een
where we identify $\End(\cS)^{f_0}$ and $\End(\cS^{f_0})$.
\end{Proposition}

\begin{proof}
Relations \eqref{hatf} and \eqref{PhiEnd} give:
\be
\widehat{\Ad(f)}=\Phi_{End(\cS)}(f_0)^{-1}\circ \Ad(f)=\Ad(\Phi_\cS(f_0))^{-1}\circ \Ad(f)=\Ad(\Phi_\cS(f_0)^{-1}\circ f)=\Ad({\hat f})~~.
\ee
\end{proof}

\noindent Given a based endomorphism $T\in \End(\cS)=\Gamma(\cM,End(\cS))$, we
can apply $\Ad(f)$ to $T$ to obtain a section $\Ad(f)(T)\eqdef
\Ad(f)\circ T\in \Gamma_{f_0}(\cM,End(\cS))$ of $End(\cS)$ along the
map $f_0$. This gives an $\R$-linear map
$\Ad(f):\Gamma(\cM,End(\cS))\rightarrow \Gamma_{f_0}(\cM,\End(\cS))$.  For
many purposes, it is more convenient to consider instead the action
$\bAd(f):\Gamma(\cM,End(\cS))\rightarrow \Gamma(\cM,End(\cS))$ of the
unbased automorphism $\Ad(f)$ on ordinary sections of $End(\cS)$,
which is defined as in Subsection \ref{subsec:ordsecaction}.

\begin{Definition}
The {\em adjoint action of $f\in \Aut^\ub(\cS)$ on ordinary sections of
  $End(\cS)$} is the $\R$-linear automorphism of the vector space
$\Gamma(\cM,End(\cS))=\End(\cS)$ defined through:
\ben
\label{bAdDef}
\bAd(f)(T)\eqdef ~^{f_0}(\Ad(f)(T))\in \Gamma(\cM,End(\cS))~~,~~\forall T\in \Gamma(\cM,End(\cS))~~.
\een
\end{Definition}

\

\noindent We have:
\ben
\label{bAdf}
\bAd(f)(T)_p=f_{f_0^{-1}(p)} \circ T_{f_0^{-1}(p)} \circ (f_{f_0^{-1}(p)})^{-1}\in \End(\cS_{p})~~\forall p\in\cM~~,
\een
i.e.: 
\ben
\label{bAdf2}
\bAd(f)(T)_{f_0(p)}=f_p\circ T_p \circ (f_p)^{-1}\in \End(\cS_{f_0(p)})~~.
\een

\begin{Proposition}
\label{prop:bAd}
For any $f\in \Aut^\ub(\cS)$, the map $\bAd(f):\End(\cS)\rightarrow \End(\cS)$ is a unital endomorphism of
the $\R$-algebra $(\End(\cS),\circ)$ of vector bundle endomorphisms of
$\cS$. The map:
\be
\bAd:\Aut^\ub(\cS)\rightarrow \Aut_\Alg(\End(\cS))
\ee
is a morphism of groups and we have:
\ben
\label{AdHatAd}
\bAd(f)(T)=[\Ad({\hat f})(T)]^{f_0^{-1}}~~\forall T\in \Gamma(\cM,End(\cS))
\een
Moreover, for any $T\in \Gamma(\cM,End(\cS))$ and any $\xi\in \Gamma(\cM,\cS)$, we have: 
\beqan
&& \bAd(f)(T)(\xi)=\mf(T(\mf^{-1}(\xi)))~~\forall \xi\in \Gamma(\cM,\cS)~~\label{bAdxi}\\
&&\bAd(f)^{-1}(T)(\xi)=\bAd(f^{-1})(T)(\xi)=\mf^{-1}(T(\mf(\xi)))~~\label{bAdxiinv}~~.
\eeqan
\end{Proposition}

\begin{proof}
The fact that $\bAd(f)$ is an $\R$-algebra automorphism follows from
\eqref{bAdf2} by direct computation, as does the fact that $\bAd$ is a
morphism of groups. Relation \eqref{AdHatAd} follows from \eqref{mf}
and \eqref{hatwidehat}. Relation \eqref{bAdxi} follows from
\eqref{bAdf2}, \eqref{mfp2} and \eqref{mfInv}. Relation
\eqref{bAdxiinv} follows from \eqref{bAdxi} by replacing $f$ with
$f^{-1}$ and using the fact that the map $f\rightarrow \mf$ and the map $\bAd$
are morphisms of groups.
\end{proof}

\begin{Remark} 
Replacing $f$ by $f^{-1}$ in \eqref{bAdf} and then replacing $p$ by $f_0(p)$ gives: 
\ben
\label{bAdfinv}
\bAd(f)^{-1}(T)_p=\bAd(f^{-1})(T)_p=f_p^{-1} \circ T_{f_0(p)} \circ f_p\in \End(\cS_p)~~\forall p\in\cM~~.
\een
\end{Remark}

\subsection{Conventions for the ordinary push-forward of vector fields 
and ordinary differential pull-back of forms}
\label{subsec:pushpull}

For any $\psi\in \Diff(\cM)$, the differential $\dd\psi$ is an unbased
automorphism of the tangent bundle $T\cM$, whose projection to $\cM$
equals $\psi$. This induces a based isomorphism of vector bundles
$\widehat{\dd\psi}\eqdef \Phi_{T\cM}(\psi)^{-1}\circ
\dd\psi:T\cM\rightarrow (T\cM)^\psi$. Applying $\dd\psi$ to a vector
field $X\in \cX(\cM)=\Gamma(\cM,T\cM)$ gives a section
$(\dd\psi)(X)\in \Gamma_\psi(\cM,T\cM)$ of $T\cM$ above the map $\psi$
(a vector field on $\cM$ along $\psi$). On the other hand, the action
of $\dd\psi$ on ordinary sections (see Subsection
\ref{subsec:ordsecaction}) induces an $\R$-linear automorphism of
$\cX(\cM)$ which we denote by $\mathbf{\dd\psi}$. This associates to
every ordinary vector field $X\in \cX(\cM)$ another ordinary vector field:
\ben
\label{psiastdef}
\psi_\ast(X)\eqdef \boldsymbol{\dd\psi}(X)=~^\psi[\dd\psi(X)]=[\widehat{\dd\psi}(X)]^{\psi^{-1}}\in \cX(M)~~,
\een
called the {\em ordinary push-forward} of $X$ through $\psi$. We have:
\be
\psi_\ast(X)_p=(\dd_{\psi^{-1}(p)}\psi)(X_{\psi^{-1}(p)})\in T_p\cM~~,~~\forall p\in \cM~~.
\ee
which also reads: 
\be
\psi_\ast(X)_{\psi(p)}=(\dd_p \psi)(X_p)~~,~~\forall p\in \cM
\ee
as well as:
\be
\psi_\ast(\alpha X)=(\alpha\circ \psi^{-1}) X~~,~~\forall \alpha\in \cC^\infty(\cM,\R)~~.
\ee
Moreover, the map $\Diff(\cM)\ni \psi\rightarrow \psi_\ast\in \Aut_\R(\cX(\cM))$ is 
a morphism of groups. 

Let $\psi^\ast:\Omega(\cM)\rightarrow \Omega(\cM)$ denote the
ordinary differential pullback of forms, defined in the convention: 
\be
\psi^\ast(\rho)_p(u_1,\ldots, u_k)\eqdef \rho_{\psi(p)}((\dd_p\psi)(u_1), \ldots, (\dd_p \psi)(u_k))~~,
\ee
for all $\rho\in\Omega^k(\cM)$ and all $u_1,\ldots, u_k\in T_p^\ast \cM$. 
Then: 
\be
\psi^\ast(\alpha\rho)=(\alpha\circ \psi)\psi^\ast(\rho)~~,~~\forall \alpha\in \cC^\infty(\cM,\R)
\ee
and the map $\Diff(\cM)\ni \psi\rightarrow \psi^\ast\in
\Aut_\R(\Omega(\cM))^\mathrm{op}$ is a morphism of groups. The
ordinary push-forward and differential pull-back are related through:
\ben
\label{pullpush0}
\psi^\ast(\rho)(X)=\rho(\psi_\ast(X))\circ \psi~~,~~\forall \rho\in \Omega^1(\cM)~~\forall X\in \cX(\cM)~~.
\een

\begin{Remark}
Other conventions are sometimes used for the ordinary push-forward of
vector fields and for the ordinary differential pullback of forms. The
present paper uses the conventions explained above.
\end{Remark}

\subsection{Twisted push-forward of bundle-valued vector fields}

\noindent Let $\cX(\cM,\cS)\eqdef \Gamma(\cM,T\cM\otimes
\cS)=\cX(M)\otimes_{\cC^\infty(\cM,\R)}\Gamma(\cM,\cS)$ denote the
$\cC^\infty(\cM,\R)$-module of $\cS$-valued vector fields defined on
$\cM$. We extend the ordinary push-forward of vector fields along
$\psi$ to the $\R$-linear map $\psi_\ast:\cX(\cM,\cS)\rightarrow
\cX(\cM,\cS^{\psi^{-1}})$ determined uniquely by the condition:
\be
\psi_\ast(X\otimes \xi)=\psi_\ast(X)\otimes \xi^{\psi^{-1}}~~,~~\forall X\in \cX(\cM)~~\forall \xi\in \Gamma(\cM,\cS)~~.
\ee
With this definition, we have: 
\be
\psi_\ast(\alpha Z)=(\alpha\circ \psi^{-1}) Z~~,~~\forall \alpha\in \cC^\infty(\cM,\R)~~\forall Z\in \cX(\cM,\cS)
\ee
and the map $\Diff(\cM)\ni \psi\rightarrow \psi_\ast\in
\Aut_\R(\cX(\cM,\cS))$ is a morphism of groups.  Let $f\in
\Aut^\ub(\cS)$ be an unbased automorphism of $\cS$.

\begin{Definition}
The {\em $\cS$-twisted push-forward along $f$} is the $\R$-linear map
$f_\ast:\cX(\cM,\cS)\rightarrow \cX(\cM,\cS)$ defined through:
\be
f_\ast\eqdef (f_0)_\ast\circ (\id_{\cX(\cM)}\otimes {\hat f})~~,
\ee
where ${\hat f}\in \Isom(\cS,\cS^{f_0})$ is the based isomorphism of
vector bundles induced by $f$ (see Subsection
\ref{subsec:basedinduced}).
\end{Definition}

\

\noindent Notice that $f_\ast$ is uniquely-determined by $\R$-linearity 
and by the property:
\ben
\label{fastdec}
f_\ast(X\otimes \xi)=(f_0)_\ast(X)\otimes {\hat f}(\xi)^{f_0^{-1}}=(f_0)_\ast(X)\otimes \mf(\xi)~~,~~\forall X\in \cX(\cM)~~\forall \xi\in \Gamma(\cM,\cS)~~,
\een
where in the last line we used relation \eqref{mf}. The map
$\Aut^\ub(\cS)\ni f \rightarrow f_\ast\in \Aut_\R(\cX(\cM,\cS))$ is a
morphism of groups.

\

\noindent For any $f\in \Aut^\ub(\cS)$, let $L_{{\hat f}}:End(\cS)\rightarrow
Hom(\cS,\cS^{f_0})$ and $R_{\hat f}:End(\cS^{f_0})\rightarrow
Hom(\cS,\cS^{f_0})$ be the based morphisms of vector bundles whose action
on sections are given by:
\beqa
& & L_{{\hat f}}(A)={\hat f}\circ A\in \Gamma(\cM,Hom(\cS,\cS^{f_0}))~~,~~\forall A\in \Gamma(\cM,End(\cS))~~\nn\\
& & R_{{\hat f}}(A)=B \circ {\hat f}\in \Gamma(\cM,Hom(\cS,\cS^{f_0}))~~,~~\forall B\in \Gamma(\cM,End(\cS^{f_0}))~~,
\eeqa
where ${\hat f}\in \Isom(\cS,\cS^{f_0})$ is the based isomorphism
induced by $f$. For any vector bundle $\cT$ over $\cM$, we extend 
$R_{\hat f}$ to a based morphism of vector bundles
$R_{\hat f}\eqdef \id_\cT\otimes R_{\hat f}:\cT \otimes
End(\cS^{f_0})\rightarrow \cT \otimes Hom(\cS,\cS^{f_0})$.
For any $\Psi\in \cX(\cM,End(\cS))=\Gamma(\cM,T\cM\otimes End(\cS))$, 
we have $\Ad(f)_\ast(\Psi)^{f_0}\in \Gamma(\cM,(T\cM)^{f_0}\otimes End(\cS^{f_0}))$ and 
$R_{\hat f}(\Ad(f)_\ast(\Psi)^{f_0})\in \Gamma(\cM,(T\cM)^{f_0}\otimes \Hom(\cS,\cS^{f_0}))$.

\begin{Proposition}
For any $f\in \Aut^\ub(\cS)$ and any $\Psi\in \cX(\cM,End(\cS))$, we have:
\ben
\label{Adpush}
R_{\hat f}(\Ad(f)_\ast (\Psi)^{f_0})=(\widehat{\dd f_0}\otimes L_{{\hat f}}) (\Psi)\in \Gamma(\cM, (T\cM)^{f_0}\otimes Hom(\cS,\cS^{f_0}))~~,
\een
where we identify $End(\cS)^{f_0}$ with $End(\cS^{f_0})$~. 
\end{Proposition}

\begin{proof}
It suffices to consider the case $\Psi=X\otimes A$ with $X\in
\cX(\cM)$ and $A\in \End(\cS)$. Then:
\beqa
R_{\hat f}(\Ad(f)_\ast (\Psi)^{f_0})&=& R_{\hat
  f}([(f_0)_\ast(X)]^{f_0} \otimes [\bAd(f)(A)]^{f_0})=R_{\hat f}
(\widehat{\dd f_0}(X)\otimes \Ad({\hat f})(A))\\ &=& \widehat{\dd
  f_0}(X)\otimes R_{\hat f}(\Ad({\hat f})(A))= \widehat{\dd
  f_0}(X)\otimes L_{{\hat f}}(A)= (\widehat{\dd f_0}\otimes L_{{\hat
    f}})(\Psi)~~,
\eeqa
where in the first equality we used \eqref{fastdec} while in the
second equality we used relations \eqref{psiastdef} and
\eqref{AdHatAd}.
\end{proof}

\subsection{Twisted differential pull-back of bundle-valued forms}

Let $\psi\in \Diff(\cM)$ and $\cS$ be a vector bundle over $\cM$. The
differential pull-back of $\cS$-valued forms is the $\R$-linear map
$\psi^\ast:\Omega(\cM,\cS)\rightarrow \Omega(\cM,\cS^\psi)$ defined as
follows. For any $\omega\in \Omega^k(\cM,\cS)$ and any $X_1\ldots
X_k\in \cX(\cM)$, we have:
\be
\psi^\ast(\omega)(X_1,\ldots, X_k)=\omega^\psi(\psi_\ast(X_1),\ldots, \psi_\ast(X_k))\in \Omega^k(\cM,\cS^\psi)~~,
\ee
where $\psi_\ast:\cX(\cM)\rightarrow \cX(\cM)$ is the ordinary
push-forward of vector fields (in the conventions of Subsection
\ref{subsec:pushpull}) and $\omega^\psi\in \Omega(\cM,\cS^\psi)$ is
the topological pull-back of $\omega$. Thus:
\be
\psi^\ast(\omega)_p(u_1,\ldots, u_k)=\omega_{\psi(p)}((\dd_p\psi)(u_1),\ldots, (\dd_p\psi)(u_k))
\ee
for all $p\in \cM$ and all $u_1,\ldots, u_k\in T_p\cM$. Notice that
$\psi^\ast$ is uniquely determined by $\R$-linearity and by the
property:
\ben
\label{psiastS}
\psi^\ast(\rho\otimes \xi)=\psi^\ast(\rho)\otimes \xi^\psi~~,~~\forall  \rho\in \Omega(\cM)~~\forall \xi\in \Gamma(\cM,\cS)~~,
\een
where $\psi^\ast(\rho)$ is the ordinary differential pull-back of the
form $\rho\in \Omega(\cM)$ (in the conventions of Subsection
\ref{subsec:pushpull}). We have:
\be
\psi^\ast(\alpha\omega)=(\alpha\circ \psi) \psi^\ast(\omega)~~,~~\forall \alpha\in \cC^\infty(\cM,\R)~~\forall \omega\in \Omega(\cM,\cS)~~
\ee
and the map $\Diff(\cM)\ni \psi\rightarrow \psi^\ast\in
\Aut_\R(\Omega(\cM,\cS))^{\mathrm{op}}$ is a morphism of groups.

\

\noindent Let $f\in \Aut^\ub(\cS)$ be an unbased automorphism of $\cS$. 

\begin{Definition}
The {\em $\cS$-twisted differential pull-back} through $f$ is the
isomorphism of graded $\R$-vector spaces:
\be
f^\ast\eqdef (\id_{\Omega(\cM)}\otimes {\hat
  f}^{-1})\circ f_0^\ast:\Omega(\cM,\cS)\rightarrow \Omega(\cM,\cS)~~,
\ee
where ${\hat f}\in \Isom(\cS,\cS^{f_0})$ is the based isomorphism
induced by $f$.
\end{Definition}

\

\noindent We shall occasionally denote the isomorphism of
$\cC^\infty(\cM,\R)$-modules $\id_{\Omega(\cM)}\otimes {\hat f}^{-1}$
simply by ${\hat f}^{-1}$.

\begin{Proposition}
\label{prop:twistedpullback}
The $\cS$-twisted differential pull-back $f^\ast$ is uniquely
determined by $\R$-linearity and by the property:
\ben
\label{fast}
f^\ast(\rho\otimes \xi)=f_0^\ast(\rho)\otimes \mf^{-1}(\xi)~~\forall \rho\in \Omega(\cM)~~\forall \xi\in \Gamma(\cM,\cS)~~.
\een
Moreover, the map $\Omega(\cM,\cS)\ni f\rightarrow f^\ast\in
\Aut_\R(\Omega(\cM,\cS))^{\mathrm{op}}$ is a morphism of groups and we have:
\beqa
&& f^\ast(\xi)=\mf^{-1}(\xi)~~,~~\forall \xi\in \Gamma(\cM,\cS)\\
&& f^\ast(\alpha\omega)=(\alpha\circ f_0)f^\ast(\omega)~~\forall \alpha\in \cC^\infty(\cM,\R)~~,~~\forall \omega\in \Omega(\cM,\cS)\\
&& f^\ast(\omega_1\wedge \omega_2)=f^\ast(\omega_1)\wedge f^\ast(\omega_2)~~,~~\forall \omega_1, \omega_2\in \Omega(\cM,\cS)~~.
\eeqa
\end{Proposition}

\begin{proof}
For any $\rho\in \Omega(\cM)$ and $\xi\in \Gamma(\cM,\cS)$, we have:
\be
f^\ast(\rho\otimes \xi)=f_0^\ast(\rho)\otimes {\hat f}^{-1}(\xi^{f_0})=f_0^\ast(\rho)\otimes \mf^{-1}(\xi)~~,
\ee
where we used relations \eqref{psiastS} and \eqref{mfinv}. It is clear
that $\R$-linearity and this property determine $f^\ast$. The
remaining statements follow immediately from this and from the
properties of $\mf^{-1}$ and of the ordinary differential pull-back of forms.
\end{proof}

\begin{Proposition}
For any $X\in \cX(\cM)$ and any $\omega\in \Omega(\cM,\cS)$, we have: 
\ben
\label{pullpush}
f^\ast(\omega)(X)=\mf^{-1}(\omega(f_{0\ast}(X)))\in \Gamma(\cM,\cS)~~.
\een
\end{Proposition}

\begin{proof}
It suffices to consider the case $\omega=\rho\otimes \xi$, with
$\rho\in \Omega(\cM)$ and $\xi\in \Gamma(\cM,\cS)$.  Then: 
\be
f^\ast(\omega)(X)=f_0^\ast(\rho)(X)\mf^{-1}(\xi)=\rho(f_{0\ast}(X))^{f_0}{\hat f}^{-1}(\xi^{f_0})={\hat f}^{-1}(\omega(f_{0\ast}(X))^{f_0})=\mf^{-1}(\omega(f_{0\ast}(X)))~~,
\ee
where we used \eqref{pullpush0} and \eqref{mfinv} and the fact that
${\hat f}^{-1}$ is a based isomorphism.
\end{proof}

\begin{Proposition}
For all $\Theta\in \Omega(\cM,End(\cS))$, $\xi\in \Gamma(\cM,\cS)$ and
$T\in \Gamma(\cM,End(\cS))$, we have:
\beqan
&& \Ad(f)^\ast(\Theta)(\xi)=f^\ast(\Theta(\mf(\xi)))\in \Omega(\cM,\cS)~~\label{Adfast}\\
&& \Ad(f)^\ast(T\Theta)=(\id_{\Omega(\cM)}\otimes\bAd(f^{-1})(T))(f^\ast(\Theta))\in \Omega(\cM, End(\cS))~~.\label{TAdfast}
\eeqan
\end{Proposition}

\begin{proof}
It suffices to consider the case $\Theta=\rho\otimes A$, where $\rho\in
\Omega(\cM)$ and $A\in \Gamma(\cM,End(\cS))$. Then:
\be
\Ad(f)^\ast(\Theta)(\xi)=(f_0^\ast(\rho)\otimes \bAd(f)^{-1}(A))(\xi)=f_0^\ast(\rho)\otimes \mf^{-1}(A(\mf(\xi)))=f^\ast(\Theta(\mf(\xi)))~~,
\ee
where we used relation \eqref{bAdxiinv}. We have $T\Theta=\rho\otimes
TA$ and:
\be
\Ad(f)^\ast(T\Theta)=f_0^\ast(\rho)\otimes \bAd(f)^{-1}(TA)=f_0^\ast(\rho)\otimes \bAd(f^{-1})(T)\bAd(f^{-1})(A)=(\id_{\Omega(\cM)}\otimes \bAd(f^{-1})(T))(\Ad(f)^\ast(\Theta))~~,
\ee
where we used the fact that $\bAd$ is a morphism of groups from
$\Aut^\ub(\cS)$ to $\Aut_\Alg(\End(\cS))$ (see Proposition
\ref{prop:bAd}).
\end{proof}

\noindent Let $\cG$ be a Riemannian metric on $\cM$. We extend the
musical isomorphism $\sharp_\cG:\Omega^1(\cM)\rightarrow \cX(\cM)$ of
$\cG$ to an isomorphism $\sharp_\cG\eqdef \sharp_\cG\otimes
\id_{\cS}:\Omega^1(\cM,\cS)\rightarrow \cX(\cM,\cS)$.

\begin{Proposition}
\label{prop:pullpush}
For any $\Theta\in \Omega^1(\cM, \cS)$, we have:
\be
[f^\ast(\Theta)]^{\sharp_{f_0^\ast(\cG)}}=(f^{-1})_\ast (\Theta^{\sharp_\cG})~~.
\ee
\end{Proposition}
\begin{proof}
It suffices to consider the case $\Theta=\rho\otimes \xi$ with $\rho\in
\Omega(\cM)$ and $\xi\in \Gamma(\cM,\cS)$. Then
$\Theta^{\sharp_\cG}=\rho^{\sharp_\cG}\otimes \xi$. For any $X\in \cX(\cM)$, we have:
\be
f_0^\ast(\rho)(X)=\rho((f_0)_\ast(X))\circ f_0=\cG((f_0)_\ast(X), \rho^{\sharp_\cG})\circ f_0=f_0^\ast(\cG)(X, (f_0^{-1})_\ast(\rho^{\sharp_\cG}))~~,
\ee
where we used \eqref{pullpush0}. This relation implies
$(f_0^\ast(\rho))^{\sharp_{f_0^\ast(\cG)}}=(f_0^{-1})_\ast(\rho^{\sharp_\cG})$. Since
$f^\ast(\Theta)=f_0^\ast(\rho)\otimes \mf^{-1}(\xi)$, we have
$[f^\ast(\Theta)]^{\sharp_{f_0^\ast(\cG)}}=f_0^\ast(\rho)^{\sharp_{f_0^\ast(\cG)}}\otimes
\mf^{-1}(\xi)=(f_0^{-1})_\ast(\rho^{\sharp_\cG})\otimes \mf^{-1}(\xi)=(f^{-1})_\ast(\Theta^{\sharp_\cG})$.
\end{proof}

\section{Local form of the equations of motion}
\label{app:localeom}

Let $\cM$ and $M$ be smooth manifolds and $\varphi:M\rightarrow
\cM$ be a smooth function. Assume that $M$ is four-dimensional and
endowed with a Lorentzian metric $g$. Let $(\cS, D, \omega)$ be a flat
symplectic vector bundle over $\cM$ and $J$ be a taming of
$(\cS,\omega)$. Notice that we do {\em not} assume that $J$ is
covariantly constant with respect to $D$. Let $L=\lambda(J)$ be the
positive polarization of $(\cS_\C,\omega_\C)$ defined by $J$.
Pick a D-flat symplectic local frame $\cE=(e_1\ldots e_n, f_1\ldots
e_n)$ of $\cS$ defined on an open subset $\cU\subset \cM$. Then there
exists a unique smooth map $\tau:=\btau^\cE(J|_U)\in
\cC^\infty(U,\SH_n)$ such that the sections $u_k\eqdef
u_k^\cE(\tau)\eqdef f_k+\sum_{l=1}^n \tau_{kl}e_l\in \Gamma(U,\cS_\C)$
(where $k=1\ldots n$) form a frame of $L|_U$.

Consider the pulled-back bundle $\cS^\varphi\eqdef \varphi^\ast(\cS)$ over
$M$, which is endowed with the pulled-back symplectic pairing
$\omega^\varphi\eqdef \varphi^\ast(\omega)$ and the pulled-back taming
$J^\varphi\eqdef \varphi^\ast(J)$, whose associated positive
polarization is $L^\varphi\eqdef \varphi^\ast(L)\subset
\cS^{\varphi,\C}$. Let $\cE^\varphi\eqdef (e^\varphi_1\ldots
e^\varphi_n,f^\varphi_1 \ldots f^\varphi_n) = \left\{
\cE^{\varphi}_{M} \right\}_{M=1,\hdots, 2n}$ be pulled-back symplectic
frame of $\cS^\varphi$, which is defined on the open set $U\eqdef
\varphi^{-1}(\cU)\subset M$, with $e^\varphi_i\eqdef
\varphi^\ast(e_i)$ and $f^\varphi_i\eqdef \varphi^\ast(f_i)$. This
frame is flat with respect to the pulled-back connection $D^\varphi$,
which is a symplectic flat connection on $\cS^\varphi$.

\subsection{Local form of the electromagnetic equations} 

Let $\cV\in \Omega^2(M,\cS^\varphi)$ be a two-form on $M$ valued in the
bundle $\cS^\varphi$, which we expand locally as:
\ben
\label{cVexp}
\cV=_U F_k \otimes e_k^\varphi + G_k \otimes f^\varphi_k~~,
\een
where $F_k,G_k\in \Omega^2(U)$. The $J$-projection $\cV_J\in \Omega^2(M,L^\varphi)$ expands as:
\be
\cV_J=_U F_k^+ \otimes e_k^\varphi + G_k^+\otimes f^\varphi_k
\ee 
with $F_k^+,G_k^+\in \Omega^2_\C(U)$, where $\Re F_k^+=F_k$ and $\Re
G_k^+=G_k$. Define:
\be
{\hat F}^+\eqdef
\left[\begin{array}{c}F^+_1\\\ldots\\ F^+_n\end{array}\right]~,~{\hat G}^+\eqdef
\left[\begin{array}{c}G^+_1\\\ldots\\ G^+_n\end{array}\right]~~.
\ee
By Proposition \ref{prop:polrels}, we have:
\ben
\label{GFmatrix}
{\hat G}^+=\tau^\varphi {\hat F}^+~~.
\een
The bundle-valued forms $\cV\in \Omega^2(U,\cS^\varphi)$ and
$\cV_J|_U\in \Omega^2(U,L^\varphi)$ can be identified in the local
symplectic frame $\cE^\varphi$ with the $2n$-vectors:
\be
{\hat \cV} = \left[\begin{array}{c} {\hat F}\\ {\hat G}\end{array}\right]~~,~~
{\hat \cV}_J = \left[\begin{array}{c} {\hat F}^+\\ {\hat G}^+\end{array}\right]~~.
\ee
and we have ${\hat \cV}=\Re {\hat \cV}_J$. 
The form $\cV$ is positively polarized iff $\ast_g^\C {\hat F}^+=-\i {\hat
  F}^+$ and $\ast_g^\C {\hat G}^+=-\i {\hat G}^+$, which gives:
\be
{\hat F}^+={\hat F}+\i \ast {\hat F}~~,~~{\hat G}^+={\hat G}+\i \ast {\hat G}~~.
\ee
The entries of the real and imaginary parts of
$\tau$:
\be
\theta=\Re\tau\in \cC^\infty(\cU,\Mat_s(n,\R))~~,~~\gamma=\Im\tau\in \cC^\infty(\cU,\Mat_s(n,\R))
\ee
capture the gauge coupling constants and ``theta angles''.
Setting $\theta^\varphi\eqdef \theta\circ (\varphi|_U)$ and
$\gamma^\varphi\eqdef \gamma\circ (\varphi|_U)$, we have
$\tau^\varphi=\theta^\varphi+i\gamma^\varphi$.
By Proposition \ref{prop:polrels}, the positive polarization condition amounts to
the relation: 
\ben
{\hat G}=\theta^\varphi~{\hat F}-\gamma^\varphi \ast_g {\hat F}~~,
\een
which in turn is equivalent with: 
\ben
\label{twsdreal}
\ast_g {\hat \cV} =\left[\begin{array}{cc} 
    (\gamma^\varphi)^{-1}\theta^\varphi~~& -(\gamma^\varphi)^{-1} \\ 
\gamma^\varphi+\theta^\varphi(\gamma^\varphi)^{-1}\theta^\varphi & -\theta^ \varphi (\gamma^\varphi)^{-1}\end{array}\right]{\hat \cV}~~.
\een
Since $D^\varphi(e_k^\varphi)=D^\varphi(f_k^\varphi)=0$, the 2-form $\cV$ is $\dd_{D^\varphi}$-closed
on $U$ iff.:
\ben
\label{BianchiEOM}
\dd {\hat F}=\dd {\hat G} =0~~,
\een
which amounts to the following equations for ${\hat F}$: 
\be
\dd {\hat F}=\dd (\theta^\varphi {\hat F}-\gamma^\varphi \ast {\hat F})=0~~.
\ee
This recovers the Maxwell equations and Bianchi identities as written
in the supergravity literature (see \cite{Andrianopoli} or
\cite[Sec. 6]{Ortin}) if one interprets the 2-forms $F_k\in
\Omega^2(U)$ as the field strengths and $G_k\in \Omega^2(U)$ are their
Lagrangian conjugates, defined with respect to the flat symplectic
frame $\cE$ of $\Delta$. Any other flat symplectic frame of $\Delta$
defined over $U$ has the form:
\be
\cE'=M\cE~~,
\ee
with $M\in \Sp(2n,\R)$. We have ${\hat \cV}=({\hat \cV}')^T\cE'$, with ${\hat
  \cV}'=M^{-T}{\hat \cV}$ and $\tau'\eqdef \tau^{\cE'}(J)$ is given by:
\be
\tau'=M\bullet \tau~~.
\ee
This reproduces the local formulation of electromagnetic duality
transformations.  

\begin{Remark}
Notice that reference \cite{Andrianopoli} uses an unusual definition
of the integration measure (see equation (2.3) in loc. cit), which
corresponds to working with the opposite orientation of $M$. Because
of this, the matrix $\gamma$ appearing in loc. cit. corresponds to
{\em minus} our $\gamma$. Taking this into account, our modular
matrix $\tau$ corresponds to the matrix denoted by $\cN$ in loc. cit.
\end{Remark}

\subsection{Local form of the  Einstein equations} 

Recall that the Einstein equation of the ESM theory is given by:
\begin{equation}
\label{eq:Einseinap}
\frac{1}{\kappa}\G(g) = \varphi^\ast(\cG) - e_\Sigma (g,\varphi) g + 2 \cV \loslash \cV \, .
\end{equation}
Using the local duality frame $\cE$ defined on $\cU\subset \cM$ and the
local coordinate systems $(\cU,y^A)$ of $\cM$ and $(U,x^\mu)$ of $M$ (where
$\varphi(U)\subset \cU$), define:
\be
Q_{MN} = Q(e_{M},e_{N})\in \cC^\infty(\cU,\R)
\ee 
and expand $\cV^M\in \Omega(U)$ as:
\be
\cV^M=\cV^M_{\mu\nu}\dd x^\mu\wedge \dd x^\nu~~\mathrm{where}~~\cV^M_{\mu\nu}\in \cC^\infty(U,\R)~~.
\ee
Then:
\be
\cV \loslash \cV=_U Q_{MN} g^{\rho\delta} \cV^{M}_{\mu\rho} \cV^{N}_{\delta\nu} \dd x^{\mu}\otimes \dd x^{\nu}~~.
\ee
Thus \eqref{eq:Einseinap} reduces to the following system of equations when restricted to $U$:
\ben
\label{eq:Einseinlocal}
\frac{1}{\kappa}\G_{\mu\nu}(g) =_{U} \cG_{AB}\partial_{\mu}\varphi^{A}\partial_{\nu}\varphi^{B}  - \frac{1}{2}g_{\mu\nu}\cG_{AB} \partial_{\rho} \varphi^{A} 
\partial^{\rho} \varphi^{B}   -  g_{\mu\nu}\Phi(\varphi) + 2\, Q_{MN} g^{\rho\delta} \cV^{M}_{\mu\rho} \cV^{N}_{\delta\nu}\, , ~~.
\een
This agrees with the Einstein equations of the local ESM theory, see equation \eqref{eq:localeinsteinconv}.

\subsection{Local form of the scalar equations} 

Recall that the scalar equation of the ESM theory take the form:
\begin{equation}
\label{eq:scalarap}
\theta_{\Sigma}(g,\varphi)  - \frac{1}{2} (\ast\cV ,\Psi^{\varphi}\cV) = 0.
\end{equation}
Using the local coordinates and the local duality frame introduced above, we find:
\be
(\ast\cV ,\Psi^{\varphi}\cV)=_U (\ast\cV , \Psi(\varphi)\cV)^A\partial_A^\varphi~~,
\ee
where the locally-defined functions $(\ast\cV , \Psi(\varphi)\cV)^A\in \cC^\infty(U,\R)$ are given by:
\be
(\ast\cV , \Psi(\varphi)\cV)^A = \cV^{M \mu\nu}\cV^{N}_{\mu\nu} \cG^{AB}(\varphi)(\partial_B J)_{MN}(\varphi)~~.
\ee
Hence \eqref{eq:scalarap} reduces to the following system of equations when restricted to $U$:
\beqan
\label{eq:scalarlocal}
\partial^\mu \partial_\mu
\varphi^A &+& \Gamma^A_{BC}(\varphi)\partial^\mu\varphi^B\partial_\mu
\varphi^C - g^{\mu\nu} K^\rho_{\mu\nu} \partial_\rho\varphi^A + \frac{1}{2} \cG^{AB}(\varphi) \partial_{B} \Phi(\varphi) - \nn\\
&-& \frac{1}{2}\cV^{M \mu\nu} \cV^{N}_{\mu\nu} \cG^{AB}(\varphi)(\partial_B J)_{MN}(\varphi) =_U 0~~.
\eeqan
The last term in \eqref{eq:scalarlocal} can be written as:
\begin{equation}
\label{eq:rexpress}
-\frac{1}{2}\cV^{M \mu\nu} \cV^{N}_{\mu\nu}(\partial_A J)_{MN}(\varphi) =_U \partial_{A} G^{\mu\nu}_{i}(\varphi) \ast F^{i}_{\mu\nu}\, . 
\end{equation}
Plugging equation \eqref{eq:rexpress} into equation
\eqref{eq:scalarlocal} we obtain the local scalar equations of the ESM
theory, in agreement with equation \eqref{eq:scalarlocalconv}.

\subsection{Summary of the local ESM equations} 
\label{app:localESM}

For reader's convenience, we summarize the local equations of motion
of the ESM theory using our conventions, which differ from those of
reference \cite{Ortin}. The local Lagrangian density with respect to
coordinate systems $(U, x^\mu)$ on $M$ and $(\cU,y^A)$ on $\cM$ such
that $U$ and $\cU$ are simply-connected and $\varphi(U)\subset \cU$
and with respect to a local flat symplectic frame $\cE=(e_i,f_i)$ of
$\Delta$ which is defined on $\cU$ is given by:
\begin{equation}
\label{eq:Localdensity}
\mathcal{L} = \frac{1}{2\kappa} \mathrm{R}(g) -
\frac{1}{2}(\cG_{AB}\circ \varphi) \partial_\mu \varphi^A\partial^\mu
\varphi^B - (\gamma_{ij}\circ \varphi) F^{i}_{\mu\nu} F^{j \mu \nu} -
(\theta_{ij}\circ \varphi) F^{i}_{\mu\nu} (\ast_g F^j)^{\mu \nu} -
\Phi\circ \varphi\, ,
\end{equation}
where $\gamma$ and $\theta$ are $n\times n$ real-valued matrix
functions defined on $\cU$ and $\gamma(p)$ is strictly
positive-definite for any $p\in \cU$. Here $F^i$ are closed 2-forms
defined on $U$ and $\ast_g$ is the Hodge operator of $g$.  Let
\be
G_{i} = (\theta_{ij}\circ \varphi) F^{j} - (\gamma_{ij}\circ
\varphi) \ast_g F^{j}\in \Omega^2(U)~
\ee
Notice that $G_i$ depends on $\varphi$. We will use the notation: 
\be
\partial_{B} G_i^{\mu\nu}\eqdef [(\partial_B\theta_{ij})\circ \varphi] F^{j\mu\nu} - 
[(\partial_B\gamma_{ij})\circ \varphi] (\ast_g F^{j})^{\mu\nu}~~.
\ee
Let $\cV^M$ be defined through
$\cV^i=G_i$ and $\cV^{i+n}=F^i$, where the index $i$ runs from $1$ to
$n$.  Then the equations of motion derived from \eqref{eq:Localdensity} are as follows:
\begin{itemize}
\item Local form of the Einstein equations:
\ben
\label{eq:localeinsteinconv}
\frac{1}{\kappa}\G_{\mu\nu}(g) =_{U} (\cG_{AB}\circ \varphi) \partial_{\mu}\varphi^{A}\partial_{\nu}\varphi^{B}  - \frac{1}{2}g_{\mu\nu}(\cG_{AB}\circ\varphi) \partial_{\rho} \varphi^{A} 
\partial^{\rho} \varphi^{B} + 2\, Q_{MN} \cV^{M}_{\mu\rho} g^{\rho\delta} \cV^{N}_{\delta\nu} - g_{\mu\nu} \Phi^\varphi \, .
\een
\item Local form of the scalar equations:
\beqan
\label{eq:scalarlocalconv}
\nabla^\mu \partial_\mu
\varphi^A + (\Gamma^A_{BC}\circ\varphi)\,\partial^\mu\varphi^B\partial_\mu
\varphi^C  + (\cG^{AB}\circ \varphi) \partial_{B}  G_i^{\mu\nu} \ast_g F^{i}_{\mu\nu}- (\cG^{AB}\circ \varphi) (\partial_{B} \Phi)\circ \varphi  = 0\, \,
\eeqan
\item Local form of the Maxwell Equations (including the Bianchi identities):
\ben
\label{BianchiEOMOrtin}
\dd F^i=\dd G_i =0~~.
\een
\end{itemize}

\section{Integral spaces and tori}
\label{app:integral}

In this appendix, we discuss various categories of integral spaces and tori and the relation 
between them. 

\subsection{Real integral spaces and real tori}

\begin{Definition}
A {\em real integral space} is a pair $(V,\Lambda)$, where $V$ is a
finite-dimensional vector space over $\R$ and $\Lambda\subset V$ is a
full lattice, i.e. a free subgroup of the Abelian group $(V,+)$ such
that $\Lambda\otimes_\Z \R\simeq V$.
\end{Definition}

\noindent The fullness condition implies that $\Lambda$ has rank equal
to $\dim_\R V$.

\begin{Definition}
Given two integral spaces $(V_1,\Lambda_1)$ and $(V_2,\Lambda_2)$, a
{\em morphism of integral spaces} $f:(V_1,\Lambda_1)\rightarrow
(V_2,\Lambda_2)$ is an $\R$-linear map $f:V_1\rightarrow V_2$ such
that $f(\Lambda_1)\subset \Lambda_2$. Such a morphism is called {\em
  finite} if $f$ is injective.
\end{Definition}

\noindent Integral spaces and morphisms of such
form a category $\Vect_0$, while integral spaces and finite morphisms
form a subcategory $\Vect_0^f$ of $\Vect_0$. Due to the fullness
condition, any integral space $(V,\Lambda)$ is completely determined
by the lattice $\Lambda$ and any morphism
$f:(V_1,\Lambda_1)\rightarrow (V_2,\Lambda_2)$ of integral spaces is
completely determined by its restriction $f_0:\Lambda_1\rightarrow
\Lambda_2$, which is a morphism of Abelian groups. It follows that
$\Vect_0$ is equivalent with the category of free finitely-generated
Abelian groups and group morphisms. 

By definition, a {\em real torus} is a compact and connected Abelian
Lie group $X$ defined over $\R$. Given two real tori $X_1,X_2$, a {\em
  morphism of real tori} $f:X_1\rightarrow X_2$ is a morphism of Lie
groups. The morphism is called {\em finite} if $\ker f$ is a finite
subgroup of $X_1$. With this definition, real tori and morphisms of
such form a category denoted $\Tor$, while real tori and finite
morphisms form a sub-category $\Tor^f$ of $\Tor$. 

There exist functors $X:\Vect_0\rightarrow \Tor$ and $\cT:\Tor\rightarrow \Vect_0$ 
defined as follows:
\begin{itemize}
\itemsep 0.0em
\item The {\em quotient functor} $X:\Vect_0\rightarrow \Tor$
  associates to each integral space $(V,\Lambda)$ the real torus
  $X(V,\Lambda)=V/\Lambda$ and to each morphism of integral spaces
  $f:(V_1,\Lambda_1)\rightarrow (V_2,\Lambda_2)$ the morphism of
  Abelian groups $X(f):V_1/\Lambda_1\rightarrow V_2/\Lambda_2$ given
  by $X(f)([v_1])=[f(v_1)]$ for all $v_1\in V_1$.
\item The {\em tangent functor} $\cT:\Tor\rightarrow \Vect_0$
  associates to each real torus $X$ the integral vector space
  $\cT(X)\eqdef (T_0X,\ker (\exp_X))$ (where $0$ is the neutral
  element of $X$ and $\exp_X:T_0X\rightarrow X$ is the exponential map
  of $X$, which is surjective) and to each morphism of real tori
  $f:X_1\rightarrow X_2$ the morphism of real integral spaces given by
  the differential $\cT(f)\eqdef \dd_0f:\cT(X_1)\rightarrow \cT(X_2)$
  of $f$ at $0$.
\end{itemize}
The following statement is well-known from Lie theory:

\begin{Proposition}
The functors $X$ and $\cT$ are mutually quasi-inverse equivalences
between $\Vect_0$ and $\Tor$, which restrict to mutually
quasi-inverse equivalences between $\Vect_0^f$ and $\Tor^f$.
\end{Proposition}

\subsection{Integral complex spaces and complex tori}

\begin{Definition}
An {\em integral complex space} is a triplet $(\cS,J,\Lambda)$ such
that:
\begin{enumerate}[1.]
\itemsep 0.0em
\item $\cS$ is a finite-dimensional vector space over $\R$.
\item $J$ is a complex structure on $\cS$.
\item $\Lambda$ is a full lattice inside $\cS$.
\end{enumerate}
\end{Definition}

\begin{Definition}
Given two integral complex spaces $(\cS_i,J_i,\Lambda_i)$ ($i=1,2$), a
{\em morphism of integral complex spaces} from $(\cS_1,J_1,\Lambda_1)$ to
$(\cS_2,J_2,\Lambda_2)$ is an $\R$-linear map $f:\cS_1\rightarrow
\cS_2$ such that $f(J_1)=J_2$ and $f(\Lambda_1)\subset \Lambda_2$. 
The morphism is called {\em finite} if $f$ is injective as a morphism of 
vector spaces. 
\end{Definition}

\noindent With this definition, integral complex spaces and morphism
of such form a category denoted $\Comp_0$. Keeping only the finite
morphisms gives a subcategory $\Comp_0^f$.  These categories map to 
$\Vect_0$ (respectively $\Vect_0^f$) through the functor which
forgets the complex structure.

A {\em complex torus} is a compact and connected complex
Lie group; such a Lie group is necessarily commutative \cite{Mumford}. 
A {\em morphism of complex tori} is a morphism of complex
Lie groups (in particular, every such morphism is a holomorphic
map). A morphism of complex tori is called {\em finite} if it has
finite kernel. With these definitions, complex tori and morphisms
between them form a category denoted $\TorComp$, while complex tori and
finite morphisms of such form a subcategory $\TorComp^f$. These
categories are mapped respectively to $\Tor$ and $\Tor^f$
through the functor which forgets the complex structure.

There functors $X$ and $\cT$ enrich to functors $X_c:\Comp_0\rightarrow \TorComp$ and
$\cT_c:\TorComp\rightarrow \Comp_0$ defined as follows:
\begin{itemize}
\itemsep 0.0em
\item Any integral complex space $(\cS,J,\Lambda)$
  determines a complex torus $X_c(\cS,J,\Lambda)\eqdef \cS/\Lambda$, endowed
  with the addition and complex structure induced from $\cS$.  Any
  morphism $f:(\cS_1,J_1,\Lambda_1)\rightarrow (\cS_2,J_2,\Lambda_2)$
  of integral complex spaces descends to a morphism
  of complex tori $X_c(f):X(\cS_1,J_1,\Lambda_1)\rightarrow
  X(\cS_2,J_2,\Lambda_2)$. 
\item For any complex torus $X$, we set $\cT_c(X)\eqdef (T_0X, J, \ker
  (\exp_X))$, where $T_0X$ is the tangent space to $X$ at the neutral
  element $0\in X$, $\exp_X:T_0X\rightarrow X$ is the exponential map
  of $X$ (which is surjective) and $J$ is the complex structure of
  $T_0X$ induced by the complex structure of $X$. Notice that
  $\cT(X)=(T_0X,\ker (\exp_X))$. For any morphism of complex tori
  $f:X_1\rightarrow X_2$, we set $\cT_c(f)\eqdef
  \dd_0f:T_0(X_1)\rightarrow T_0(X_2)$. The properties of the
  exponential map imply that $\cT_c(f)$ is a morphism of integral
  complex spaces from $\cT_c(X_1)$ to $\cT_c(X_2)$.
\end{itemize}

\noindent The following statement summarizes the well-known quotient
description of complex tori \cite{Mumford, BL2}:

\begin{Proposition}
The functors $X_c$ and $\cT_c$ are mutually quasi-inverse equivalences
of categories between $\Comp_0$ and $\TorComp$. The restrictions of
these functors give mutually quasi-inverse equivalences between
$\Comp_0^f$ and $\TorComp^f$.
\end{Proposition}

\subsection{Integral symplectic spaces and symplectic tori}

\begin{Definition}
A {\em (real) integral symplectic space} is a triple
$(V,\omega,\Lambda)$ such that:
\begin{enumerate}[1.]
\itemsep 0.0em
\item $(V,\omega)$ is a finite-dimensional symplectic vector
  space over $\R$.
\item $(V,\Lambda)$ is an integral space. 
\item $\omega$ is integral with respect to $\Lambda$, i.e. we have
  $\omega(\Lambda,\Lambda)\subset \Z$.
\end{enumerate}
\end{Definition}

\begin{Definition}
Let $(V_1,\omega_1,\Lambda_1)$ and $(V_2,\omega_2,\Lambda_2)$ be two
integral symplectic spaces. A {\em morphism} from
$(V_1,\omega_1,\Lambda_1)$ to $(V_2,\omega_2,\Lambda_2)$ is a
symplectic morphism $f:(V_1,\omega_1)\rightarrow (V_2,\omega_2)$ which
is also a morphism of integral spaces, i.e. which satisfies that
$f(\Lambda_1)\subset \Lambda_2$. 
\end{Definition}

Notice that any morphism of integral symplectic spaces is
injective (since any morphism of symplectic vector spaces is).  Let
$\Symp_0$ denote the category of integral symplectic spaces and
morphisms between such. This category fibers over $\Vect_0^f$ through
the functor which forgets the lattice. 

A {\em real symplectic torus} is a pair $(X,\Omega)$,
where $X$ is a real torus and $\Omega$ is a translationally-invariant
symplectic form on $X$. A {\em morphism of real symplectic tori}
$f:(X_1,\Omega_1)\rightarrow (X_2,\Omega_2)$ is a morphism of real
tori $f:X_1\rightarrow X_2$ such that $f^\ast(\Omega_2)=\Omega_1$;
such a morphism is necessarily finite, i.e. $\ker f$ is automatically
a finite subgroup of $X_1$. These definitions give a category
$\TorSymp$ of symplectic tori and maps of such. The category
$\TorSymp$ fibers over $\Tor^f$ through the functor which
forgets the symplectic form.

The functors $X$ and $\cT$ enrich to functors $X_s:\Symp_0\rightarrow
\TorSymp$ and $\cT_s:\TorSymp\rightarrow \Symp_0$, which are
defined as follows:
\begin{itemize}
\item For any integral symplectic space $(V,\Lambda,\omega)$, we set
  $X_s(V,\Lambda,\omega)\eqdef (X(V,\Lambda),\Omega(\omega))$, where
  $\Omega(\omega)$ is the translationally-invariant symplectic form
  induced by $\omega$ on $X(V,\Lambda)=V/\Lambda$. For any morphism
  $f:(V_1,\omega_1,\Lambda_1) \rightarrow (V_2,\omega_2,\Lambda_2)$,
  we set $X_s(f)\eqdef X(f)$.
\item For any symplectic torus $(X,\Omega)$, we set
  $\cT_s(X,\Omega)\eqdef (V,\omega,\Lambda)$, where
  $(V,\Lambda)=\cT(X)$ and $\omega\eqdef \Omega_0$ is the value of
  $\Omega$ at the neutral element of $X$. For any morphism
  $f:(X_1,\omega_1)\rightarrow (X_2,\omega_2)$, we set $\cT_s(f)\eqdef
  \cT(f)$.
\end{itemize}

\noindent With these definitions, we have:

\begin{Proposition}
The functors $X_s$ and $\cT_s$ are mutually quasi-inverse equivalences
between the categories $\Symp_0$ and $\TorSymp_\R$, which restrict to
mutually quasi-inverse equivalences between the categories $\Symp_0^f$
and $\TorSymp_\R^f$.
\end{Proposition}

\paragraph{Integral symplectic bases.}

Let $(V,\omega,\Lambda)$ be an integral symplectic space of dimension $\dim_\R V=2n$. Let:
\be
\Sp(V,\omega,\Lambda)\eqdef \{A\in \Sp(V,\omega)|A(\Lambda)=\Lambda\} 
\ee
denote the automorphism group of $(V,\omega,\Lambda)$ in the category $\Symp_0$. 

\begin{Definition}
An {\em integral symplectic basis} of $(V,\omega,\Lambda)$ is
a basis $\bcE \eqdef (\lambda_1\ldots \lambda_n,\mu_1\ldots \mu_n)$ of
the lattice $\Lambda$ (i.e. a basis of the free $\Z$-module
$\Lambda$) such that:
\be
\omega(\lambda_i,\lambda_j)=\omega(\mu_i,\mu_j)=0~~,~~\omega(\lambda_i,\mu_j)=t_j\delta_{ij}~~\forall i,j=1\ldots n~~,
\ee
where $t_1,\ldots, t_n\in \Z_{>0}$ are strictly positive integers
satisfying the divisibility conditions $t_1|t_2|\ldots | t_n$.
\end{Definition}

\begin{Remark}
An integral symplectic basis of $(V,\omega,\Lambda)$ need not be a
symplectic basis of the vector space $V$.
\end{Remark}

\noindent The elementary divisor theorem implies that any integral
symplectic space admits integral symplectic bases. Moreover, the
integers $t_1\ldots t_n$ (which are called the {\em elementary
  divisors} of $(V,\omega,\Lambda)$) do not depend on the choice of
integral symplectic basis. Let:
\be
\Div^n\eqdef \{(t_1,\ldots, t_n)\in (\Z_{>0})^n~|~~t_1|t_2|\ldots |t_n\}~~.
\ee
and: 
\be
\bdelta(n)\eqdef (1,\ldots, 1)\in \Div^n~~.
\ee
The ordered system of elementary divisors:
\be
\bt(V,\omega,\Lambda)\eqdef (t_1,t_2,\ldots , t_n)\in \Div^n
\ee
is called the {\em type} of $(V,\omega,\Lambda)$. The integral symplectic space
$(V,\omega,\Lambda)$ is called {\em principal} if
$\bt(V,\omega,\Lambda)=\bdelta(n)$. This is equivalent with
the requirement that the Liouville volume of (any) elementary cell of
$\Lambda$ with respect to $\omega$ be equal to $1$. For any
$\bt=(t_1,\ldots, t_n)\in \Div^n$, let:
\be
\mathrm{D}_\bt\eqdef \diag(t_1,\ldots, t_n)\in \Mat(n,\Z)~~.
\ee
and:
\be
\Gamma_\bt\eqdef \left[\begin{array}{cc}I_n &0 \\ 0 & \mathrm{D}_\bt\end{array}\right]\in \Mat(2n,\Z)~~.
\ee
Let $\Lambda_\bt\subset \R^{2n}$ be the lattice defined through:
\ben
\label{LambdaDef}
\Lambda_\bt=\{(m_1\ldots m_n, m_{n+1}t_1,\ldots, m_{2n}t_n)|m_1,\ldots, m_{2n}\in \Z\}~~.
\een
Then $(\R^{2n},\omega_n, \Lambda_\bt)$ is an integral symplectic space, where
$\omega_n$ denotes the canonical symplectic pairing on $\R^{2n}$. 
Notice that $\Lambda_{\bdelta(n)}=\Z^{2n}$. Also notice that $\Lambda_\bt\subset
\Z^{2n}$ is a sub-lattice of $\Z^{2n}$ of index $t_1t_2\ldots t_n$.

\paragraph{The modified Siegel modular group.}

\begin{Definition}
The {\em modified Siegel modular group\footnote{This is the group denoted by
    $G_D$ in \cite[Section 8]{BL}. It is isomorphic to the group
    denoted by $\Gamma_D=\Sp^D_{2n}(\Z)$ on page 216 of loc. cit., which has a
    different action on the Siegel upper half space.} of type $\bt\in
  \Div^n$} is the subgroup of $\Sp(2n,\R)$ defined through:
\be
\Sp_\bt(2n,\Z)\eqdef \{A\in \Sp(2n,\R)|\Gamma_\bt A\Gamma_\bt^{-1}\in \GL(2n,\Z)\}~~.
\ee
\end{Definition}

\noindent   We have:
\be
\Sp_\bt(2n,\Z)=\{A\in \Sp(2n,\R)|A\Lambda_\bt\subset \Lambda_\bt\}=\{A\in \Sp(2n,\R)|A\Lambda_\bt= \Lambda_\bt\}\simeq \Sp(\R^{2n},\omega_n,\Lambda_\bt)~~,
\ee
and $\Sp_{\bdelta_n}(2n,\Z)=\Sp(2n,\Z)$. Since
$\Lambda_\bt\subset \Z^{2n}$, we have $\Sp(2n,\Z)\subset
\Sp_\bt(2n,\Z)$.  

\begin{Proposition}
\label{prop:intsympbasis}
Let $(V,\omega,\Lambda)$ be a $2n$-dimensional integral symplectic space
of type $\bt$. Then any integral symplectic basis
$\bcE=(\lambda_1\ldots \lambda_n,\mu_1\ldots \mu_n)$ of this space
induces an isomorphism of integral symplectic spaces between
$(V,\omega,\Lambda)$ and $(\R^{2n},\omega_n,\Lambda_\bt)$ and an
isomorphism of groups between $\Sp(V,\omega,\Lambda)$ and
$\Sp_\bt(2n,\Z)$.
\end{Proposition}

\begin{proof}
The vectors $e_1\eqdef \lambda_1,\ldots e_n\eqdef \lambda_n, f_1\eqdef
\frac{1}{t_1}\mu_1,\ldots f_n\eqdef \frac{1}{t_n}\mu_n$ form a
symplectic basis $\cE=(e_1\ldots e_n, f_1\ldots f_n)$ of $(V,\omega)$,
which induces the stated isomorphism of integral symplectic
spaces. Recall that $\Omega_n$ denotes the matrix of $\omega_n$ in the
canonical basis. An invertible $\R$-linear operator $A\in \Aut_\R(V)$
is a symplectomorphism iff its matrix ${\hat A}$ in the basis $\cE$
belongs to $\Sp(2n,\R)$, which amounts to the condition:
\be
{\hat A}^t\Omega_{n}{\hat A}=\Omega_{n}~~.
\ee
On the other hand, $A$ preserves $\Lambda$ iff its matrix
$\boldsymbol{\hat{A}}$ in the basis $\bcE$ belongs to
$\Mat(2n,\Z)$. The two matrices are related through:
\be
\boldsymbol{\hat{A}}=\Gamma_\bt \hat{A}\Gamma_\bt^{-1}~~.
\ee
Hence the condition $A(\Lambda)\subset \Lambda$ amounts to the
requirement $\Gamma_\bt \hat{A}\Gamma_\bt^{-1} \in \Mat(2n,\Z)$.
This gives the desired isomorphism of groups.
\end{proof}

\subsection{Integral tamed symplectic spaces and polarized Abelian varieties}

\begin{Definition}
An {\em integral tamed symplectic space} is an ordered system
$(V,\omega,J,\Lambda)$ such that:
\begin{enumerate}[1.]
\itemsep 0.0em
\item $(V,\omega,\Lambda)$ is an integral symplectic space.
\item $J$ is a taming of the symplectic space $(V,\omega)$. 
\end{enumerate}
The {\em type} of an integral tamed symplectic space
$(V,\omega,J,\Lambda)$ is the type of the underlying
integral symplectic space $(V,\omega,\Lambda)$.
\end{Definition}

\begin{Definition}
Given two integral tamed symplectic spaces $(V_i,\omega_i,J_i,\Lambda_i)$
($i=1,2$), a {\em morphism} from $(V_1,\omega_1,J_1,\Lambda_1)$ to
$(V_2,\omega_2,J_2,\Lambda_2)$ is a linear map $f:V_1\rightarrow
V_2$ such that:
\begin{enumerate}[1.]
\item $f$ is a morphism of integral symplectic spaces from
  $(V_1,\omega_1,\Lambda_1)$ to $(V_2,\omega_2,\Lambda_2)$.
\item $f$ is a morphism of complex vector spaces from
  $(V_1,J_1)$ to $(V_2,J_2)$.
\end{enumerate}
\end{Definition}

\noindent Any morphism of integral tamed symplectic spaces is
necessarily injective. Integral tamed symplectic spaces and morphisms
of such form a category denoted $\TamedSymp_0$. This category is
naturally equivalent with the category $\Herm_0$ of integral Hermitian
spaces and morphisms of such, whose definition should be obvious to
the reader. There exist forgetful functors from $\TamedSymp_0$ to $\Symp_0$ 
and $\Comp_0$.

Let $X$ be a complex torus of complex dimension $n$. We identify the
Neron-Severi group of $X$ with the image $NS(X)\subset H^2(X,\Z)$ of
the morphism $c_1:\Pic(X)\rightarrow H^2(X,\Z)$, which fits into an
exact sequence:
\be
0\longrightarrow \Pic^0(X)\hookrightarrow \Pic(X)\stackrel{c_1}{\longrightarrow} NS(X)\longrightarrow 0~~.
\ee
Recall \cite{BL2} that $NS(X)$ is a free Abelian group whose rank
$\rho(X)$ (the Picard number of $X$) satisfies $\rho(X)\leq n^2$. Since
$H^2(X,\Z)$ is torsion-free, it can be identified with its image in
$H^2(X,\C)$, in which case $NS(X)$ can be identified with the group
$H^{1,1}(X,\C)\cap H^2(X,\Z)$ of integral $(1,1)$-cohomology classes.
By the Kodaira embedding theorem, a holomorphic line bundle on $X$ is
ample iff it is positive. The ample cone $NS_+(X)\subset NS(X)$
consists of the first Chern classes of all ample holomorphic line
bundles defined on $X$. This can be identified with the set
$\cK(X)\cap H^2(X,\Z)$ of integral K\"{a}hler classes, where $\cK(X)$ is
the K\"{a}hler cone of $X$. A {\em polarization} of $X$ is an element
$c\in NS_+(X)$; this can be identified with an integral K\"{a}hler class
of $X$. The complex torus $X$ is called an Abelian variety if its
admits polarizations, i.e. if $NS_+(X)\neq \emptyset$. This is
equivalent with the condition that $X$ admits a K\"{a}hler-Hodge metric
and to the condition that $X$ admits a holomorphic embedding
into some complex projective space. In this case, $X$ is a smooth projective 
algebraic variety, i.e. it can be presented as the zero locus of a system of 
homogeneous polynomial equations in that projective space. 
Let $\cT_c(X)=(\cS,J,\Lambda)$.  Any cohomology class $\alpha\in \cK(X)$
contains a unique translationally-invariant representative, whose
corresponding K\"{a}hler metric is itself translationally invariant and
hence induced by a Hermitian form $H$ on the complex the vector space
$(\cS,J)$. Equivalently, such a K\"{a}hler metric is determined by the
symplectic form $\omega=\Im H$, which has the property that
$(\cS,J,\omega)$ is a tamed symplectic vector space. This gives a
bijection between the K\"{a}hler cone $\cK(X)$ and the space of all
symplectic pairings on $\cS$ for which $J$ is a taming. The K\"{a}hler
class $\alpha$ is integral iff $\omega$ is integral with respect to
$\Lambda$. This gives a bijection between the set $NS_+(X)$ of
polarizations of $X$ and the set of all symplectic pairings $\omega$
on $\cS$ which make $(\cS,J,\omega,\Lambda)$ into an integral
Hermitian space. We let $c(\omega)$ denote the polarization defined by
such a symplectic pairing $\omega$. 

A {\em polarized Abelian variety}
is a pair $(X,c)$, where $X$ is an Abelian variety and $c\in N_+(X)$
is a polarization of $X$. Given two polarized Abelian varieties
$(X_1,c_1)$ and $(X_2,c_2)$, a {\em morphism of polarized Abelian
  varieties} $f:(X_1,c_1)\rightarrow (X_2,c_2)$ is a morphism of
complex tori $f:X_1\rightarrow X_2$ such that $f^\ast(c_2)=c_1$. Such
a morphism necessarily has finite kernel. Let $\AbVar$
denote the category of complex polarized Abelian varieties and
morphisms between such. There exist obvious forgetful functors 
from $\AbVar$ to $\TorSymp$ and $\TorComp$. 

\paragraph{Equivalence between $\TamedSymp_0$ and $\AbVar$.}

Consider the functor $X_h:\TamedSymp_0\rightarrow \AbVar$ defined
as follows:
\begin{itemize}
\item Given an integral tamed symplectic space
  $(\cS,\omega,J,\Lambda)$, we set $X_h(\cS,\omega,J,\Lambda)\eqdef
  (X_c(\cS,J,\Lambda),c(\omega))$, where $c(\omega)\eqdef
  \beta_X^{-1}(\omega)\in NS_+(X)$ is the polarization of
  $X_c(\cS,J,\Lambda)$ defined by $\omega$.
\item Given a finite morphism
  $f:(\cS_1,\omega_1,J_1,\Lambda_1)\rightarrow
  (\cS_2,\omega_2,J_2,\Lambda_2)$ of integral tamed symplectic spaces, let
  $X_h(f):=X_c(f):X_c(\cS_1,J_1,\Lambda_1)\rightarrow
  X_2:=X_c(\cS_2,J_2,\Lambda_2)$ be the finite morphism of complex tori
  induced by $f$. Then $X_h(f)^\ast(c_2)=c_1$, where $c_1$ and $c_2$ are
  the polarizations induced by $\omega_1$ and $\omega_2$ on $X_1$ and
  $X_2$. Thus $X_h(f)$ is a morphism of polarized Abelian varieties from
  $X_h(\cS_1,J_1,\omega_1,\Lambda_1)$ to
  $X_h(\cS_2,J_2,\omega_2,\Lambda_2)$.
\end{itemize}

\noindent The following well-known statement underlies the description of Abelian 
varieties through period matrices satisfying the Riemann bilinear relations:

\begin{Proposition}
The functor $X_h:\TamedSymp\rightarrow \AbVar$ is an equivalence of categories. 
\end{Proposition}

\paragraph{The modular parameter relative to an integral symplectic basis.}

Let $(\cS,J,\omega,\Lambda)$ be an integral tamed symplectic space of complex
dimension $n$ and type $\bt\in \Div^n$.  Let $\bcE\eqdef
(\lambda_1,\ldots, \lambda_n, \mu_1\ldots \mu_n)$ be an integral
symplectic basis of $(\cS, \omega,\Lambda)$.  Let $e_i\eqdef
\lambda_i$ and $f_i\eqdef \frac{1}{t_i}\mu_i$. Then $\cE\eqdef
(e_1\ldots e_n, f_1\ldots f_n)$ is a symplectic basis of
$(\cS,\omega)$ and $\cE^2\eqdef (f_1\ldots f_n)$ is a basis (over
$\C$) of the complex vector space $(\cS,J)$. We have:
\be
e_i=\sum_{i,j=1}^n \btau_{ij}f_j~~,
\ee  
where $\btau\eqdef \btau^\cE \in \SH_n$ is the modular parameter of
the tamed symplectic vector space $(\cS,J,\omega)$ with respect to the
symplectic basis $\cE$ (see Appendix \ref{app:space_tamings}). 

\begin{Definition}
The modular parameter of $(\cS,J,\omega,\Lambda)$ relative to the
integral symplectic basis $\bcE=(\mu_1,\ldots, \mu_n,\lambda_1,\ldots,
\lambda_n)$ is defined through:
\be
\btau^\bcE\eqdef \btau^\cE~~,
\ee
where $\tau^\cE$ is the modular parameter of the tamed symplectic
space $(\cS,J,\omega)$ with respect to the symplectic basis $\cE$ of
$(\cS,\omega)$ defined through $\cE\eqdef (e_1\ldots e_n,
\frac{1}{t_1}\mu_1\ldots \frac{1}{t_n}f_n)$.
\end{Definition}

\noindent The period matrix of $(\cS,J,\omega)$ with respect to the
real basis $\bcE$ and complex basis $\cE_2$ is given by \cite[Section
  8.1]{BL}:
\be
\Pi_{\bcE}^{\cE_2}=[\btau^\bcE, \mathrm{D_\bt}]~~\mathrm{where}~~\mathrm{D}_\bt\eqdef \diag(t_1\ldots t_n)~~.
\ee
The pair $(X_h(\cS,D,J,\omega),\bcE)$ is a polarized Abelian variety
of type $\bt$ {\em with symplectic basis} in the sense of
\cite[Section 8]{BL} and $\btau^\bcE$ coincides with its modular
parameter. As explained in loc cit, the moduli space of such objects
is the Siegel upper half space $\SH_n$, while the moduli space of
polarized Abelian varieties of complex dimension $n$ and type $\bt$
equals the quotient $\SH_n/\Sp_\bt(2n,\Z)$, where $\Sp_\bt(2n,\Z)$
acts on $\SH_n$ by matrix fractional transformations. This action is
properly discontinuous since $\Sp_\bt(2n,Z)$ is a discrete subgroup of
$\Sp(2n,\R)$.

\subsection{Summary of correspondences}

Summarizing everything, we have the following commutative diagram of
categories and functors, whose vertical arrows are equivalences and
whose remaining arrows are forgetful functors. The functors $X_s$ and $X_h$ 
arise in Section \ref{sec:Dirac}. 
\ben
\label{diag:cat}
\scalebox{1.2}{
\xymatrix{
\TamedSymp_0 \ar_{X_h}[dd]\ar[rd] \ar[r] & \Comp_0^f\ar[rd] \ar@<7pt>^(.70){X_c}[dd]\\
& \Symp_0\!\!\!\ar_(.30)[dd]{X_s}[dd]  \ar[r] & \Vect_0^f\ar^X[dd]\\
\AbVar \ar[rd] \ar[r] & ~~\TorComp^f \ar[rd]\\
& \TorSymp  \ar[r] & \Tor^f\\
}}
\een

\end{document}